\documentclass[12pt,reqno,a4paper]{amsart}
\usepackage{amsmath,amssymb,dsfont,graphicx,cite,latexsym,
epsf,cancel,epic,eepic,mathrsfs}
\usepackage[urlcolor=cyan,colorlinks=false]{hyperref}
\usepackage{tikz}
\usepackage{caption}
\usepackage{subcaption}
\usepackage{tikz,pgfplots}

\usepackage{mathtools}
\usepackage{algorithm}
\usepackage{algpseudocode}

\graphicspath{ {../Figures} }

\usetikzlibrary{arrows.meta}
\usetikzlibrary{arrows}
\newtheorem{theorem}{Theorem}
\newtheorem{proposition}[theorem]{Proposition}
\newtheorem{lemma}[theorem]{Lemma}

\newtheorem{rem}[theorem]{Remark}

\newtheorem{definition}[theorem]{Definition}
\newtheorem{corollary}[theorem]{Corollary}
\makeatletter
\expandafter\let\expandafter
\reset@font\csname reset@font\endcsname
\def\subeqnarray{\arraycolsep1pt
   \def\@eqnnum\stepcounter##1{\stepcounter{subequation}
       {\reset@font\rm(\theequation\alph{subequation})}}
\jot5mm     \eqnarray}

\makeatother
\newcommand{\bbZ}{{\mathbb Z}}
\newcommand{\bbR}{{\mathbb R}}
\newcommand{\bbC}{{\mathbb C}}

\def\epsilon{\varepsilon}

\def\tilde{\widetilde}

\def\fe{\mathfrak{e}}

\def\bbR{\mathbb R}
\def\bbC{\mathbb C}

\def\bbZ{\mathbb Z}

\newcommand{\br}[1]{\left( #1 \right)}

\newcommand{\cubr}[1]{\left\{ #1 \right\} }
\newcommand{\restr}[2]{\ensuremath{\left.#1\right|_{#2}}}
\DeclarePairedDelimiter{\abs}{\lvert}{\rvert}
\DeclarePairedDelimiter{\norm}{\lVert}{\rVert}

\DeclareMathOperator*{\argmin}{arg\,min}
\DeclareMathOperator{\myRe}{Re}
\DeclareMathOperator{\myIm}{Im}
\DeclareMathOperator{\res}{res}
\DeclareMathOperator{\Jac}{Jac}
\DeclareMathOperator{\Lip}{Lip}
\DeclareMathOperator{\Div}{div}
\DeclareMathOperator{\Hess}{Hess}

\newbox\meibox
\def\placeunder#1#2#3#4{\setbox\meibox%
\vbox{\hbox{\hskip#4$\hphantom{#2}$}\hbox{$\hphantom{#1}$}}%
\vtop{\baselineskip=0pt\lineskiplimit=\baselineskip%
\lineskip=#3\hbox to \wd\meibox{\hfil\hskip#4$#2$\hfil}%
\hbox to \wd\meibox{\hfil$#1$\hfil}}}
\def\intprod{\mathbin{\hbox to 6pt{%
                 \vrule height0.4pt width5pt depth0pt
                 \kern-.4pt
                 \vrule height6pt width0.4pt depth0pt\hss}}}

\begin{document}
\title
{Dimers and M-curves}

%
\author{Alexander I. Bobenko \and Nikolai Bobenko \and Yuri B. Suris}

\thanks{E-mail: {\tt  bobenko@math.tu-berlin.de, nikolai.bobenko@unige.ch, suris@math.tu-berlin.de}}

%


%

\begin{abstract}
   In this paper we develop a general approach to dimer models analogous to Krichever's scheme in the theory of integrable systems. 
   We start with a Riemann surface and the simplest generic meromorphic functions on it and demonstrate how to obtain integrable dimer models.
   These are dimer models on doubly periodic bipartite graphs with quasi-periodic positive weights.
   Dimer models with periodic weights and Harnack curves are recovered as a special case.
   This generalization from Harnack curves to general M-curves leads to transparent algebro-geometric structures.
   In particular explicit formulas for the Ronkin function and surface tension as integrals of meromorphic differentials on M-curves are obtained. Furthermore we describe the variational principle for the height function in the quasi-periodic case.
   Based on Schottky uniformizations of Riemann surfaces we present concrete computational results including computing the weights and sampling dimer configurations with them.
   The computational results are in complete agreement with the theoretical predictions.
\end{abstract}

\makeatletter
\def\@tocline#1#2#3#4#5#6#7{\relax
  \ifnum #1>\c@tocdepth 
  \else
    \par \addpenalty\@secpenalty\addvspace{#2}%
    \begingroup \hyphenpenalty\@M
    \@ifempty{#4}{%
      \@tempdima\csname r@tocindent\number#1\endcsname\relax
    }{%
      \@tempdima#4\relax
    }%
    \parindent\z@ \leftskip#3\relax \advance\leftskip\@tempdima\relax
    \rightskip\@pnumwidth plus4em \parfillskip-\@pnumwidth
    #5\leavevmode\hskip-\@tempdima
      \ifcase #1
       \or\or \hskip 1em \or \hskip 2em \else \hskip 3em \fi%
      #6\nobreak\relax
    \hfill\hbox to\@pnumwidth{\@tocpagenum{#7}}\par
    \nobreak
    \endgroup
  \fi}
\makeatother

\maketitle
\vspace{0.5 cm}

\begin{center}
   \bf Dedicated to Igor Krichever.
\end{center}

\tableofcontents

\section{Introduction}


Over the past two decades, dimer models became an essential part of integrable statistical mechanics, see for example~\cite{kenyon_laplacian_2002, Kenyon_Okounkov_Sheffield_2006, boutillier_minimal_2023}. Various aspects of integrable structure such as the invariance of the partition function under the spider move (Z-invariance), expression of the partition function, correlation functions and inverse of Dirac operators in terms of spectral curves, etc. proved to be a powerful tool in the investigation of models~\cite{Kenyon_Okounkov_2007, boutillier_elliptic_2020, boutillier_minimal_2023, berggren_geometry_2023, Borodin_2x2_2023}. 
Moreover many integrable systems and systems with cluster algebra structure were identified within dimer models~\cite{goncharov_dimers_2013}.

In this work we consider dimer models on planar bipartite graphs. In the foundational work~\cite{Kasteleyn1967} the partition function for finite graphs was shown to be the determinant of the Kasteleyn operator, which is a weighted signed adjacency matrix. Furthermore the partition function per fundamental domain was explicitly computed for periodic lattices with constant weights.

These results were generalized in~\cite{kenyon_laplacian_2002} for isoradial graphs with critical weights. Further generalization for arbitrary doubly periodic weights was provided in~\cite{Kenyon_Okounkov_Sheffield_2006} with the discovery of a remarkable connection to algebraic geometry. The Kasteleyn operator was used to define the spectral curve and it was shown that it is a Harnack curve, which is a classical object in algebraic geometry \cite{Mikhalkin_2000}. Many characteristics of dimer models were expressed purely in terms of algebro-geometric data. In particular a characterization of all ergodic Gibbs measures in terms of points in the Newton polygon was provided. Furthermore the free energy of a model with magnetic field was shown to correspond to the Ronkin function~\cite{Kenyon_Okounkov_Sheffield_2006} and limit shapes for some boundary data were computed \cite{Kenyon_Okounkov_2007}.

The solution of the inverse problem was given in~\cite{fock_inverse_2015} where the weights of the model were expressed in terms of algebro-geometric data. A comprehensive investigation of the corresponding models was carried out in~\cite{boutillier_elliptic_2020} for spectral curves of genus $1$ and in~\cite{boutillier_minimal_2023} for general genus. For any minimal graph with Fock's weights an explicit expression for inverse Kasteleyn operators was obtained yielding alternative expressions for the ergodic Gibbs measures in the doubly periodic case.

Research on dimer models and their relation to integrable systems continues to flourish. For example, an explicit description of limit shapes of the Aztec diamond as well as convergence of local fluctuations to the corresponding Gibbs measures was shown in~\cite{berggren_geometry_2023}. The genus 1 case was studied before that in~\cite{Borodin_2x2_2023}. Recently a relation to circle patterns and Miquel dynamics was established in~\cite{Kenyon_Russkikh_circle_patters_2022}. See also~\cite{astala_dimer_2023} for an elaborate investigation of variational problems of dimer limit shapes and \cite{chelkak_t-embeddings_2023} for an investigation of fluctuations via t-embeddings.

Somewhat similar methods based on the analysis on Riemann surfaces are known in the theory of integrable systems under the name of finite gap (or algebro-geometric) integration theory~\cite{Dubrovin_1981, belokolos_algebro-geometric_1994}. In 1977 Krichever~\cite{Krichever_1977} suggested an integration scheme based on the notion of Baker-Akhiezer (BA) function. This is a function on a Riemann surface which is uniquely determined by its special analytic properties, and is also expressed explicitly in terms of theta functions and Abelian integrals. From these properties one derives linear differential or difference equations satisfied by the BA function. The coefficients in these equations are also derived from explicit formulas for the BA function. Special (multiphase) solutions of many integrable systems, like Korteweg-de Vries (KdV), Kadomtsev-Petviashvili (KP), sine-Gordon etc. equations were constructed this way, see for example \cite{belokolos_algebro-geometric_1994}. 

In this paper we develop a general approach to dimer models analogous to Krichever's scheme in the theory of integrable systems. 
We start with a Riemann surface and the simplest generic meromorphic functions on it and demonstrate how to obtain integrable dimer models (Section~\ref{sec:3_baker_akhiezer}). At each vertex of a graph there is a BA-function. They satisfy certain local linear relations whose coefficients are the weights of the corresponding dimer model. Moreover, exactly as in the finite gap theory, we derive the explicit formulas for these coefficients from the formulas for the BA-functions, and they turn out to coincide with the Fock weights. It is remarkable that as a special case of our construction we obtain some fundamental algebro-geometric identities, like Fay's identity \cite{Fay_1973, Mumford_theta_1} as well as 
the invariance of the dimer model under the spider move (Section~\ref{sec:4_fay_and_spider}). The latter is closely related to the multidimensional consistency \cite{bobenko_discrete_2008} of the model, which is proved along the same lines (Section~\ref{sec:5_quadGraphs_and_consistency}).

Whereas positivity of weights is crucial for statistical mechanics, Fock's original paper \cite{fock_inverse_2015}, as well as some other papers on integrable dimers associated to Riemann surfaces \cite{goncharov_dimers_2013, george_spectra_2019} deal with complex Riemann surfaces that lead to complex weights. Similar problems of characterization of real and non-singular solutions were encountered in the theory of integrable systems from its early days \cite{Dubrovin_1981, belokolos_algebro-geometric_1994}. The class of Riemann surfaces leading to real, non-singular solutions of the KP-equation and to self-adjoint operators was identified as M-curves, \cite{Dubrovin_Natanzon_1988}. An M-curve is a Riemann surface $\mathcal R$ of genus $g$ with an anti-holomorphic involution $\tau$ and maximal number $g+1$ of real ovals $X=X_0\cup X_1\cup\ldots X_g$.  In particular, in 1985 Krichever \cite{Krichever_1985} constructed a class of two-dimensional difference operators with positive weights, coming from M-curves. 

In the context of dimer models in \cite{Kenyon_Okounkov_2006} it was established that doubly periodic dimer models are in one to one correspondence with Harnack curves. Their Riemann surfaces are M-curves. The above mentioned Krichever's operators can be identified with the Dirac operator with Fock's weights on the regular hexagonal lattice.  For planar bipartite graphs with arbitrary faces a construction of Dirac operators with positive weights is presented in Section~\ref{sec:6_M-curves} and in \cite{boutillier_minimal_2023}\footnote{We discussed this construction with the authors of \cite{boutillier_minimal_2023} in private correspondence in 2021.}. In this setup one obtains quasi-periodic weights. Harnack curves appear in the special case of periodic weights (Section~\ref{sec:7_Torus_and_Harnack}). We demonstrate that a generalization from Harnack curves to general M-curves leads to a general approach with general transparent algebro-geometric structures and self-contained presentation. It implies new results even in the classical $g=0$ case.

An explicit formula for the free energy of a doubly periodic dimer model was derived in \cite{Kenyon_Okounkov_Sheffield_2006}. It was shown that it is given by the Ronkin function, defined on the amoeba of the corresponding Harnack curve. A remarkable generalization of all these notions was given by Krichever in \cite{krichever_amoebas_2014} without any relation to dimer models. In particular he defined Harnack data $\mathcal{S}$ for M-curves, with the corresponding amoeba map and Ronkin function, inheriting most of the classical properties. In Section~\ref{sec:8_Ronkin_function} we present a unified picture including diffeomorphisms of the (open) factor $\mathcal{R}_+^\circ=(\mathcal{R}\setminus X)/ \tau$ to the generalized amoeba $\mathcal{A}_{\mathcal S}^\circ$ and to the generalized Newton polygon $\Delta_{\mathcal S}^\circ$. Further we present a new formula for Krichever's Ronkin function along with its Legendre dual. Our representation is given via an integral of a holomorphic one-form. It is more convenient for investigation and computation than Krichever's original formula which is an integral of a two-form over a certain domain of $\mathcal{R}_+$. Note that for dimer models with doubly periodic weights the (Legendre) dual Ronkin function is known as surface tension \cite{Kenyon_Okounkov_Sheffield_2006}.

A variational principle for the height function of a doubly periodic dimer model was obtained in \cite{cohn_variational_2000}. It was shown that it is the minimizer of the surface tension functional with prescribed boundary conditions. The functional is convex which implies the uniqueness of the minimizer. In Section~\ref{sec:9_heightFunction} we extend the notion of surface tension to quasi periodic weights and show that the corresponding height function is also given by its minimizer.

In Section~\ref{sec:10_Schottky_Uniformization} we show that all the above characteristics of dimer models including the weights, amoeba and polygon maps, surface tension as well the height function can be effectively computed. For this purpose we use the method proposed in~\cite{Bobenko_1987} based on the Schottky uniformization of the corresponding M-curve. The characteristics of dimers are then represented via converging (Poincaré theta) series. This representation turns out to be useful even in the case of genus $0$ which corresponds to Kenyon's isoradial weights. 

Finally in Section~\ref{sec:12_computation} we present concrete computational results including computing general Fock weights and sampling dimer configurations with them. The computational results are in complete agreement with the theoretical predictions (See Fig.~\ref{fig:simulation_with_amoeba}).

\begin{figure}[h]
	\centering
	\includegraphics[width=.95\linewidth]{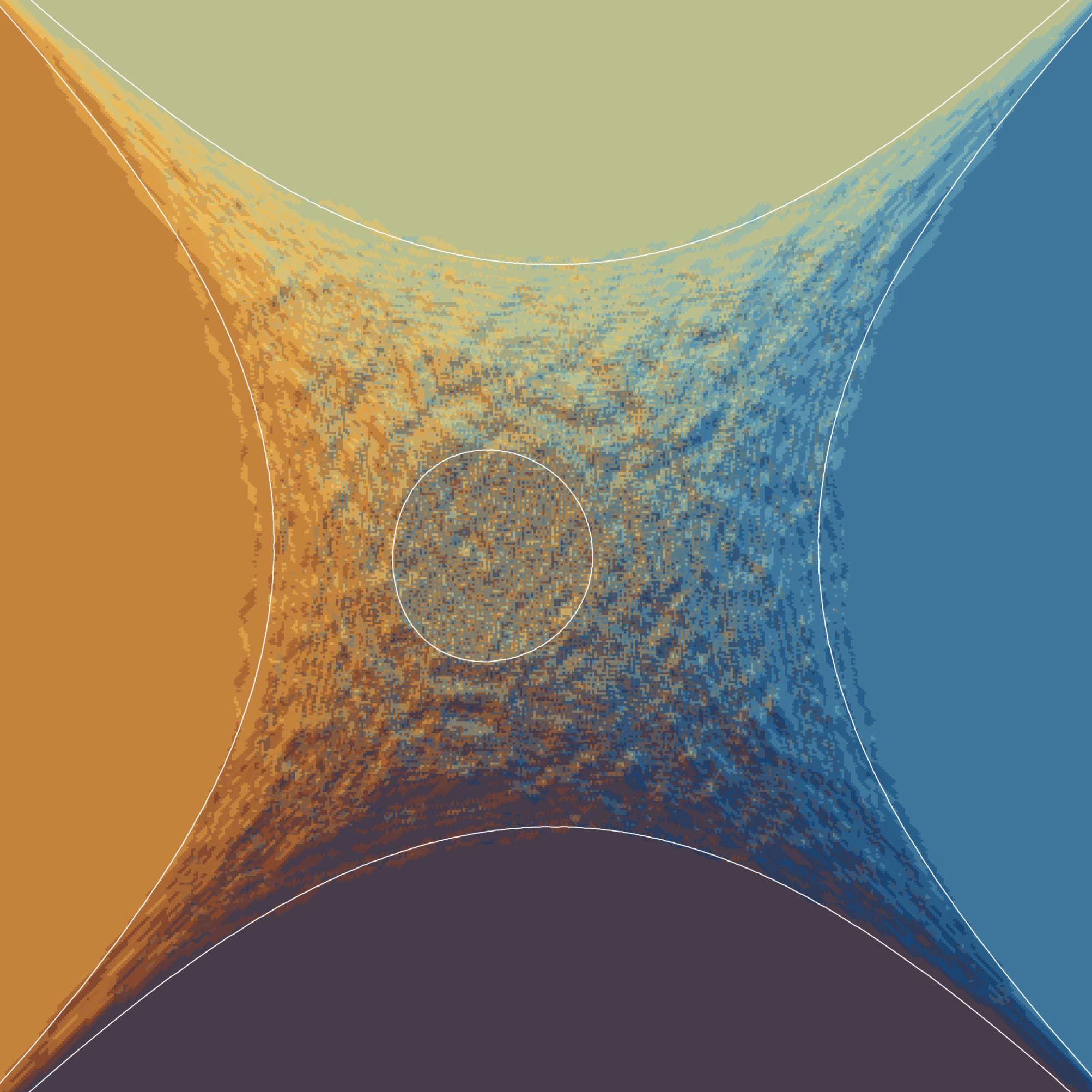}
	\caption[Random height function sampled according to Fock weights with a volume constraint and overlaid with corresponding amoeba.]{ Random height function sampled according to Fock weights with a volume constraint and overlaid with theoretical prediction.\footnotemark }
	\label{fig:simulation_with_amoeba}
\end{figure}

\vspace*{1cm}
{\bf Acknowledgements.} The first and third authors were supported by the DFG Collaborative Research Center TRR 109 "Discretization in Geometry and Dynamics".

The second author was supported by the NCCR SwissMAP as well as by the Swiss NSF grants 200400 and 197226.

\section{Dimer models on bipartite graphs}
\label{sect dimer models}

While one can consider the dimer model on an arbitrary (finite connected)  graph, we will restrict our attention to graphs with certain additional structures.
\medskip 

\textbf{Dimer model on planar graphs.} We consider our graphs as \emph{planar graphs}, i.e., graphs embedded in surfaces. Thus, the underlying graph is understood as $G=(V,E,F)$, with $V$ being the set of vertices, $E$ the set of edges, and $F$ the set of faces (connected components of the complement of $E$ in the surface). Actually, we will assume that the graph $G$ is \emph{simple}, so that each face of $G$ is an $n$-gon with $n\ge 3$. We can introduce the \emph{dual graph} $G^*=(V^*,E^*,F^*)$, for which 
$V^*=F$, $F^*=V$, and $E$, $E^*$ are in a one-to-one correspondence, the corresponding edges are denoted $e\in E$, $e^*\in E^*$.

A \emph{dimer configuration} on $G$ is a subset $D\subset E$ such that each vertex of $G$ is adjacent to exactly one edge from $D$. A purely combinatorial problem of counting all possible dimer configurations for a given graph turns into an appealing and important \emph{dimer problem} of statistical mechanics if one supplies dimer configurations with probabilities, or Boltzmann weights. More precisely, one assumes that $G$ has an \emph{edge weighting} which is a function $\nu:E\to\bbR_+$, and one defines the Boltzmann weight of a dimer configuration $D$ by
\begin{equation}\label{Boltzmann weight}
\nu(D)=\prod_{e\in D}\nu(e).
\end{equation} 
Then, the corresponding probability measure on the set of dimer configurations is given by
\begin{equation}
	\label{eq:def_Boltzmann_measure}
	\mathbb{P}(D)=\frac{\nu(D)}{Z(G,\nu)},\quad Z(G,\nu)=\sum_D \nu(D),
\end{equation}
where the summation in the formula for $Z(G,\nu)$ (the partition function) is extended over all possible dimer configurations on $G$.

We quote the famous Kasteleyn's formula for the dimers partition function for graphs embedded in the plane. For this, the following ingredients are necessary. 
\footnotetext{This is a height function colored according to its local slope. We stress that it is practical to compute both the weights as well as the amoeba map numerically. This particular configuration was sampled on a $(500 \times 500)$ square in $\mathbb{Z}^2$ with a volume constraint of $0.3$ volume per unit. The Schottky data used is given by $\cubr{\beta^-, \alpha^+, \beta^+, \alpha^-} = \cubr{-2.4, -0.4, 0.4, 2.4}$ and $(A,B,\mu) = (0.25 + i 0.9, 0.25 - i 0.9, 0.015)$. }
\begin{definition}\label{def Kasteleyn or 1}
Fix an orientation of the plane, and let the faces of the graph $G$ embedded in the plane inherit this orientation. For each face $f\in F$, orient $\partial f$ as the boundary of the oriented face $f$. This gives the \emph{natural orientation} of the edges of $\partial f$. For an arbitrary orientation $\sigma$ of the edges of $G$, denote by $N^\sigma(\partial f)$ the number of edges of the boundary $\partial f$ for which $\sigma$ induces the opposite to the natural orientation. An edge orientation $\sigma$ is called a \emph{Kasteleyn orientation}, if for each $f\in F$, the quantity $N^\sigma(\partial f)$ is odd.
\end{definition}
\begin{theorem}
Every graph embedded in the plane admits a Kasteleyn orientation $\sigma$. If $\sigma$ is such an orientation, the following formula holds true:
\begin{equation}\label{Z=pf}
Z(G,\nu)=|{\rm Pf}(A^\sigma)|,
\end{equation}
where $A^\sigma$ is the skew-symmetric matrix whose rows and columns are indexed by $V$, and the entries are given by 
$$
a_{ij}^\sigma=\left\{\begin{array}{cl}  \nu(e), & {\rm if}\; \{i,j\}=e\in E, \;{\rm and}\;
\sigma\; {\rm orients}\; e\; {\rm from}\; j\;{\rm to}\; i, \\ 
-\nu(e), & {\rm if}\; \{i,j\}=e\in E, \;{\rm and}\;
\sigma\; {\rm orients}\; e\; {\rm from}\; i\;{\rm to}\; j, \\0, & {\rm if}\; \{i,j\}\not\in E. \end{array}\right.
$$
\end{theorem}
\medskip

\textbf{Dimer model on bipartite graphs.} These results and techniques can be further specialized and refined under an additional assumption (which we adopt hereafter)
that the graph $G$ is \emph{bipartite}, that is, the set of vertices is represented as $V=B\cup W$ ($B$ being denoted as black vertices and $W$ as white vertices), and each edge $e\in E$ connects a black vertex with a white one, $e=(bw)$. As an immediate corollary, each face of $G$ is a $(2n)$-gon with $n\ge 2$.

For a bipartite graph, one can introduce an orientation of edges different from the natural one, say from the black vertex to the white one (let us call this the bw-orientation). Then it is easy to see that  Definition~\ref{def Kasteleyn or 1} is equivalent to the following one.
\begin{definition}\label{def Kasteleyn or 2}
For an arbitrary orientation $\sigma$ of the edges of $G$, denote by $M^\sigma(\partial f)$ the number of edges of the boundary $\partial f$ which $\sigma$ orients from white to black, i.e., opposite to the bw-orientation. An edge orientation $\sigma$ is called a \emph{Kasteleyn orientation}, if for each $(2n)$-gonal $f\in F$, we have $M^\sigma(\partial f)\equiv n+1\pmod 2$.
\end{definition}

For a bipartite graph, the matrix $A^\sigma$ has (upon a suitable ordering of the index set $V$) a block structure:
$$
A^\sigma=\begin{pmatrix} 0 & K_{WB} \\ K_{BW} & 0\end{pmatrix},
$$
where $K_{WB}$ is a matrix whose rows are indexed by $W$ and whose columns are indexed by $B$ (and, of course, $K_{BW}=-K_{WB}^{\rm T}$). We call $K_{WB}$ the \emph{Kasteleyn matrix}. Its definition can be formulated as follows:
$$
K_{wb}=\left\{\begin{array}{cl}  \nu(e), & {\rm if}\; \{b,w\}=e\in E, \;{\rm and}\;
\sigma\; {\rm orients}\; e\; {\rm from}\; b\;{\rm to}\; w, \\ 
-\nu(e), & {\rm if}\; \{b,w\}=e\in E, \;{\rm and}\;
\sigma\; {\rm orients}\; e\; {\rm from}\; w\;{\rm to}\; b, \\0, & {\rm if}\; \{b,w\}\not\in E. \end{array}\right.
$$
Formula \eqref{Z=pf} takes the form
\begin{equation}\label{Z=det}
Z(G,\nu)=|\det(K_{WB})|,
\end{equation} 

An important issue is the \emph{gauge transformations} of the dimer model. We call two weight functions $\nu$, $\nu'$ gauge equivalent if, for each $e=\{i,j\}\in E$, we have $\nu'(e)=\nu(e)\lambda(i)\lambda(j)$ with some $\lambda:V\to\bbR_+$. One easily sees that the probability measure $\mathbb{P}(D)$ does not change under a gauge transformation of weights. For a bipartite $G$, one has the following characterization of gauge equivalent weights.
\begin{proposition}
Two weight functions $\nu$, $\nu'$ are gauge equivalent iff for each $f\in F$ with $\partial f=e_1\cup e_2\cup\ldots \cup e_{2n}$, there holds
\begin{equation}\label{gauge inv}
W_f=
\frac{\nu(e_1)\nu(e_3)\ldots \nu(e_{2n-1})}{\nu(e_2)\nu(e_4)\ldots\nu(e_{2n})}=
\frac{\nu'(e_1)\nu'(e_3)\ldots \nu'(e_{2n-1})}{\nu'(e_2)\nu'(e_4)\ldots\nu'(e_{2n})}.
\end{equation}
\end{proposition}
For the proof observe that, under the bw-orientation of edges, $\log \nu$ is a discrete 1-form on $G$. Equation \eqref{gauge inv} is equivalent to $\log \nu'-\log \nu$ being closed, which for graphs embedded in the plane is equivalent to $\log \nu'-\log \nu$ being exact, i.e., $\log \nu'-\log \nu=df$ for some $f:V\to\bbR_+$. Now set $\lambda(b)=\exp(-f(b))$ and $\lambda(w)=\exp(f(w))$ for any $b\in B$, $w\in W$. 

One can extend the notion of gauge invariance by admitting complex-valued gauge functions $\lambda:V\to\bbC^*$, then the gauge equivalent weight functions are also allowed to take values in $\bbC^*$ rather than in $\bbR_+$.

It is important that the characterization of a Kasteleyn orientation for bipartite graphs is conveniently given in terms of the face gauge invariants \eqref{gauge inv}: an edge orientation $\sigma$ is a Kasteleyn orientation, if for each $f\in F$ with the boundary $(2n)$-gon $w_1b_1\ldots w_nb_n$, the entries of the Kasteleyn matrix $K_{WB}$ satisfy
\begin{equation}\label{Kasteleyn cond}
{\rm sign}\left(\frac{K_{w_1b_1}K_{w_2b_2}\ldots K_{w_nb_n}}{K_{w_2b_1}K_{w_3b_2}\ldots K_{w_1b_n}}\right)=(-1)^{n+1}.
\end{equation}
Of course, this condition will ensure the representation \eqref{Z=det} for the partition function also for the case when the entries $K_{wb}$ are from $\bbC^*$, with the understanding that the expression in the parentheses on the left-hand side must be real for each $f\in F$.
\medskip

\textbf{Height function.} Any dimer configuration $D$ on a planar bipartite graph $G$ defines, via
$$
\omega_D(bw)=-\omega_D(wb)=\left\{\begin{array}{l} 1, \; {\rm if}\; (bw)\in D, \\ 0, \; {\rm otherwise}\end{array}\right.
$$
a discrete 1-form on $G$. This form has a special property of having divergence $1$ at black vertices and $-1$ at white ones, where the divergence of $\omega$ at a vertex $v$ is defined as $d^*\omega(v) = \sum_{w\sim v} \omega(vw)$. Such 1-forms are called \emph{unit forms}.

Fix a unit form $\omega_0$. Then $\omega=\omega_D-\omega_0$ is divergence free and therefore the dual form $\omega^*$ on the dual graph $G^*$ is closed. Recall that the dual form $\omega^*$ is defined by $\omega^*(e^*)=\omega(e)$, where the edge $e^*\in E^*$ dual to $e\in E$ is directed so that it intersects the positively directed $e$ from left to right. The form $\omega^*$ is closed  in the sense that its integral (sum) around the boundary of any face $f^*\in F^*$ vanishes. Therefore, it is exact, $\omega^*=dh$, with a function $h:V^*=F\to \bbR$ defined, up to an additive constant, by integrating $\omega^*$ along paths in $G^*$. This function is called the \emph{height function} of the dimer configuration $D$. Note that it depends on the choice of a unit form $\omega_0$, but that for two dimer configurations $D_1, D_2$ the difference $h(D_1) - h(D_2)$ does not depend on the choice of $\omega_0$, making this a more natural and intrinsic object than the individual height function $h(D)$.

Actually, a slight modification of this construction will be used. On the square grid, a usual choice of $\omega_0$ is the  Thurston's original one \cite{thurston1989groups}, with $\omega_0 = \frac{1}{4}$ on each positively (from black to white) directed edge. Likewise, on the hexagonal grid one can take $\omega_0 = \frac{1}{3}$ on each positively directed edge. 
However, with these standard choices, for a finite subgraph of a grid the boundary vertices will not have divergence $\pm 1$. Therefore, the form $\omega^*$ will not be closed around the boundary vertices. This is usually remedied by a modification of the dual graph, in which the edges $e^*$ dual to the boundary edges of $G$ are considered to have distinct ``remote'' vertices (of valence 1) rather than to have a common vertex corresponding to the outer (unbounded) face of $G$, see Fig.~\ref{Fig discr Abel}. In such a modified dual graph there are no faces dual to the boundary vertices of $G$, i.e., no cycles around the boundary vertices of $G$. Note that the height function at the boundary (valence 1) vertices of $G^*$ does not depend on the choice of the dimer configuration $D$. 
\medskip



\textbf{Diamond graph.} For the discussion of integrability properties of the dimer model, it is useful to introduce the concept of the \emph{diamond graph} $G^\diamond$ associated to $G$ and constructed as follows. The set of vertices of $G^\diamond$ is 
 $V\sqcup V^*$. Each pair of dual edges, say $e=(bw)\in E(G)$ and $e^*=(f_1f_2)\in E(G^*)$ define a quadrilateral $(bf_1wf_2)$, see Fig.~\ref{Fig discr Abel}. These quadrilaterals constitute the faces of the cell decomposition (quad-graph) 
$G^\diamond$. Let us stress that edges of $G^\diamond$ belong neither to $E(G)$ nor to $E(G^*)$ (the latter are diagonals of the quadrilaterals). 

A \emph{strip, or a train track} in $G^\diamond$ is a sequence of quadrilateral faces $q_j\in F(G^\diamond)$ such that any pair $q_{j-1}$, $q_j$ is adjacent along the edge $\fe_j = q_{j-1}\cap q_j$, and $\fe_j$, $\fe_{j+1}$ are opposite edges of $q_j$. The edges $\fe_j$ are called traverse edges of the strip. We visualize the strips as \emph{strands} passing through the midpoints of the diagonals $e=(bw)$ of the quadrilaterals, and entering and leaving each quadrilateral through two opposite sides. Through the midpoint of any edge $e\in E(G)$ there pass two strips. We use the convention that these strands are directed so that white vertices lie to the left, and black vertices lie to the right of the strands. See Figure~\ref{Fig discr Abel}. 

\begin{figure}[h]
\centering
\begin{subfigure}[p]{.45\textwidth}
	\centering
	\fontsize{10pt}{12pt}\selectfont
	\def\svgwidth{.9\textwidth}
	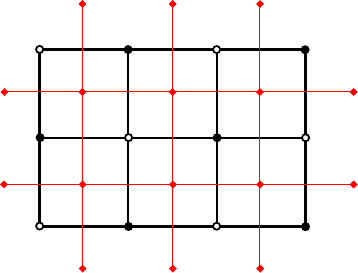
\end{subfigure}
\begin{subfigure}[p]{.35\textwidth}
	\centering
	\fontsize{10pt}{12pt}\selectfont
	\def\svgwidth{0.8\linewidth}
\begingroup%
  \makeatletter%
  \providecommand\color[2][]{%
    \errmessage{(Inkscape) Color is used for the text in Inkscape, but the package 'color.sty' is not loaded}%
    \renewcommand\color[2][]{}%
  }%
  \providecommand\transparent[1]{%
    \errmessage{(Inkscape) Transparency is used (non-zero) for the text in Inkscape, but the package 'transparent.sty' is not loaded}%
    \renewcommand\transparent[1]{}%
  }%
  \providecommand\rotatebox[2]{#2}%
  \newcommand*\fsize{\dimexpr\f@size pt\relax}%
  \newcommand*\lineheight[1]{\fontsize{\fsize}{#1\fsize}\selectfont}%
  \ifx\svgwidth\undefined%
    \setlength{\unitlength}{318.09269726bp}%
    \ifx\svgscale\undefined%
      \relax%
    \else%
      \setlength{\unitlength}{\unitlength * \real{\svgscale}}%
    \fi%
  \else%
    \setlength{\unitlength}{\svgwidth}%
  \fi%
  \global\let\svgwidth\undefined%
  \global\let\svgscale\undefined%
  \makeatother%
  \begin{picture}(1,1.41666335)%
    \lineheight{1}%
    \setlength\tabcolsep{0pt}%
    \put(0,0){\includegraphics[width=\unitlength,page=1]{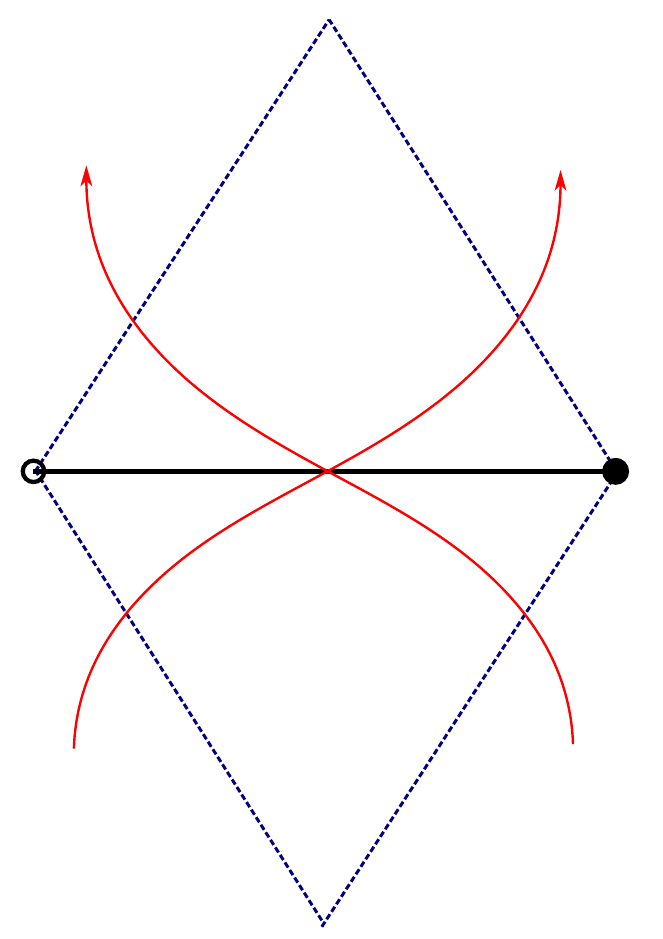}}%
    \put(-0.00214136,0.62226423){\color[rgb]{0,0,0}\makebox(0,0)[lt]{\lineheight{1.25}\smash{\begin{tabular}[t]{l}$w$\end{tabular}}}}%
    \put(0.94835858,0.62780658){\color[rgb]{0,0,0}\makebox(0,0)[lt]{\lineheight{1.25}\smash{\begin{tabular}[t]{l}$b$\end{tabular}}}}%
    \put(0.74329453,1.08504415){\color[rgb]{1,0,0}\makebox(0,0)[lt]{\lineheight{1.25}\smash{\begin{tabular}[t]{l}$\alpha$\end{tabular}}}}%
    \put(0.17798255,1.08504415){\color[rgb]{1,0,0}\makebox(0,0)[lt]{\lineheight{1.25}\smash{\begin{tabular}[t]{l}$\beta$\end{tabular}}}}%
    \put(0,0){\includegraphics[width=\unitlength,page=2]{Figures/FockWeightWithDiamond.pdf}}%
    \put(0.55762833,1.39264038){\color[rgb]{0,0,0}\makebox(0,0)[lt]{\lineheight{1.25}\smash{\begin{tabular}[t]{l}$f_2$\end{tabular}}}}%
    \put(0.55485712,0.01815644){\color[rgb]{0,0,0}\makebox(0,0)[lt]{\lineheight{1.25}\smash{\begin{tabular}[t]{l}$f_1$\end{tabular}}}}%
    \put(0,0){\includegraphics[width=\unitlength,page=3]{Figures/FockWeightWithDiamond.pdf}}%
  \end{picture}%
\endgroup%

\end{subfigure}
\caption{An example of a bipartite graph $G$ and its dual graph $G^*$ in red. The height function is defined on the vertices of $G^*$ (left). A quadrilateral face of $G^\diamond$, with two directed strips through it (right).}
\label{Fig discr Abel}
\end{figure}

\medskip
\textbf{Infinite graphs.} Our principal interest in this paper is towards understanding the thermodynamic limit of the dimer model on infinite regular lattices in $\bbR^2$, with the main examples being the regular square lattice and the regular hexagonal lattice. The diamond graphs for them are the regular square lattice again (rotated by 45$^\circ$), resp. the regular dual kagome lattice. Both are instances of \emph{quasicrystallic rhombic embeddings}, as introduced in \cite[Sect. 6.8]{belokolos_algebro-geometric_1994}. Those are quad-graphs embedded in the plane with all quadrilateral faces being rhombi with sides of the unit length from a finite set $\{ \pm \theta_1,\ldots,\pm \theta_d\}\subset\mathbb S^1$. As discussed in \cite[Sect. 6.8]{belokolos_algebro-geometric_1994}, such quad-graph can be seen as a sort of projection of a certain two-dimensional subcomplex (stepped surface) of a multi-dimensional regular square lattice $\bbZ^d$. While the presentation in this paper will be mainly restricted to two above mentioned lattices (square and hexagonal), it admits a generalization to arbitrary quasicrystallic rhombic embeddings. A still more general class of quad-graphs for which similar theory can be developed constitute minimal isoradial immersions studied in \cite{boutillier_elliptic_2020} and \cite{boutillier_minimal_2023}.

\section{Baker-Akhiezer functions and Fock weights}
\label{sec:3_baker_akhiezer}

In this section, we recall the definition of Fock weights \cite{fock_inverse_2015}, which can be constructed on an arbitrary bipartite graph, with a data provided by an arbitrary compact Riemann surface. After that, we introduce \emph{Baker-Akhiezer functions}, which are a key device in the theory of integrable systems, widely used in various contexts by I.M. Krichever (who also established this terminology).

Here are the main ingredients of the construction. 
\begin{itemize}
\item $\mathcal{R}$ a compact Riemann surface of genus $g$, with a fixed homology basis $a_1,\ldots,a_g$, $b_1,\ldots,b_g$.
\item $\boldsymbol{\omega}=(\omega_1,\ldots,\omega_g)$ the set of normalized holomorphic differentials, with $a$-periods $\int_{a_j}\omega_k=\delta_{jk}$ and $b$-periods $\int_{b_j}\omega_k=B_{jk}$;
\item $J(\mathcal{R})=\mathbb C^g/(\mathbb Z^g+B\mathbb Z^g)$ the Jacobi variety of $\mathcal{R}$;
\item $A:\mathcal{R}\to J(\mathcal{R})$, $P\mapsto A(P)=\int_{P_0}^{P}\boldsymbol{\omega}\pmod {\mathbb Z^g+B\mathbb Z^g}$ the Abel map.

\item $\theta(z)=\theta(z|B)$ the theta function on $\mathbb C^g$ with the period matrix $B$. It is defined by
\begin{equation}\label{theta}
\theta(z)=\theta(z|B)=\sum_{m\in\mathbb Z^g}\exp\Big(\pi i\langle Bm,m\rangle+2\pi i \langle z,m\rangle\Big).
\end{equation}
Its behavior under shifts of the argument by a period is given by
\begin{equation}\label{theta monodromy}
\theta(z+m+Bn)=\exp\big(-\pi i\langle n,Bn\rangle-2\pi i\langle z, n\rangle\big)\theta(z), \quad m,n\in\mathbb Z^g. 
\end{equation}
We will also need theta-functions with (half-integer) characteristics $\Delta=\left[\begin{array}{l} \Delta_1 \\ \Delta_2\end{array}\right]$, where
each of $\Delta_1,\Delta_2$ is of the form $\frac{1}{2}(\delta_1,\ldots,\delta_g)^{\rm T}$ with $\delta_j\in\{0,1\}$:
\begin{eqnarray}\label{theta with char}
\theta[\Delta](z) & = & \sum_{m\in \mathbb Z^g} \exp\Big(\pi i\langle B(m+\Delta_1),m+\Delta_1\rangle+2\pi i\langle z+\Delta_2,m+\Delta_1\rangle \Big)
  \nonumber\\
    & = & \exp\Big(\pi i\langle B\Delta_1,\Delta_1\rangle+2\pi i\langle z+\Delta_2,\Delta_1\rangle\Big) \theta(z+\Delta_2+B\Delta_1).
\end{eqnarray} 
 The characteristic $\Delta$ is called even (odd), if $4\langle\Delta_1,\Delta_2\rangle\equiv 0$ (resp. $\equiv 1\pmod 2$). The function $\theta[\Delta](z)$ is even (odd), iff the characteristic $\Delta$ is even (resp. odd). 

\item The prime form $E(\alpha,\beta)$, which can be expressed through a theta function with an odd characteristic $\Delta$:
\begin{equation}\label{E thru theta}
E(\alpha,\beta)=\frac{\theta[\Delta]\Big(\int_\alpha^\beta\boldsymbol{\omega}\Big)}{h_\Delta(\alpha)h_\Delta(\beta)}
\end{equation}
with a holomorphic spinor $h_\Delta$ (the result being independent of $\Delta$).
\end{itemize}

Given a planar bipartite graph $G$ (as discussed in Section~\ref{sect dimer models}), and a compact Riemann surface $\mathcal{R}$, we start with a \emph{labeling} of the strips of $G^\diamond$, assigning to each (directed) strip a point $\alpha\in \mathcal{R}$. The following notions were introduced by Fock in \cite{fock_inverse_2015}.

\begin{definition}\label{def discr Abel} 
For a labeled graph $G^\diamond$, the \emph{discrete Abel map}  $\eta:V\cup F =V(G^\diamond)\to J(\mathcal{R})$ is defined as follows: set $\eta(v_0)=0$ for some $v_0\in V(G^\diamond)$, and then extend it to the whole of $V(G^\diamond)$ according to the following recursive rule (we use here notations from Fig.~\ref{Fig discr Abel}): 
\begin{eqnarray}
\eta(b) & = & \eta(f_1)-A(\beta)=\eta(f_2)-A(\alpha), \\ 
\eta(w) & = & \eta(f_2)+A(\beta)=\eta(f_1)+A(\alpha),
\end{eqnarray}
which also implies
\begin{equation}
\eta(b)=\eta(w)-A(\alpha)-A(\beta), \qquad \eta(f_2)=\eta(f_1)+A(\alpha)-A(\beta).
\end{equation}
\end{definition}
In a more informal way, this rule can be formulated as follows: if a (directed) edge of $G^\diamond$ traverses the strip labelled by $\alpha$ from right to left, then it contributes $A(\alpha)$, while traversing the strip from left to right contributes $-A(\alpha)$ to the discrete Abel map. This is well defined (closed along cycles in $G^\diamond$).

\begin{definition}\label{def Fock weights}
For a labeled bipartite graph $G$, \emph{Fock weights} are coefficients on (directed) edges of $G$ depending on the labels $\alpha$, $\beta$ of the two strips intersecting the edge defined by the formula:
\begin{equation}\label{Fock coeff}
K_{wb}(\alpha,\beta)=\frac{E(\alpha,\beta)}{\theta(\eta(f_1)+D)\theta(\eta(f_2)+D)},
\end{equation}
where $\eta$ is the discrete Abel map, and $D\in \mathbb C^g$ is arbitrary. Here we use notations of Fig.~\ref{Fig Fock coeff}), the rule is as follows: going from $w$ to $b$, and then turning left at the midpoint, one first meets the strand labeled $\alpha$, then the strand labeled $\beta$. 
\end{definition}

\begin{figure}[h]
	\centering
	\fontsize{10pt}{12pt}\selectfont
	\def\svgwidth{0.35\linewidth}
\begingroup%
  \makeatletter%
  \providecommand\color[2][]{%
    \errmessage{(Inkscape) Color is used for the text in Inkscape, but the package 'color.sty' is not loaded}%
    \renewcommand\color[2][]{}%
  }%
  \providecommand\transparent[1]{%
    \errmessage{(Inkscape) Transparency is used (non-zero) for the text in Inkscape, but the package 'transparent.sty' is not loaded}%
    \renewcommand\transparent[1]{}%
  }%
  \providecommand\rotatebox[2]{#2}%
  \newcommand*\fsize{\dimexpr\f@size pt\relax}%
  \newcommand*\lineheight[1]{\fontsize{\fsize}{#1\fsize}\selectfont}%
  \ifx\svgwidth\undefined%
    \setlength{\unitlength}{295.17427603bp}%
    \ifx\svgscale\undefined%
      \relax%
    \else%
      \setlength{\unitlength}{\unitlength * \real{\svgscale}}%
    \fi%
  \else%
    \setlength{\unitlength}{\svgwidth}%
  \fi%
  \global\let\svgwidth\undefined%
  \global\let\svgscale\undefined%
  \makeatother%
  \begin{picture}(1,0.94844645)%
    \lineheight{1}%
    \setlength\tabcolsep{0pt}%
    \put(0,0){\includegraphics[width=\unitlength,page=1]{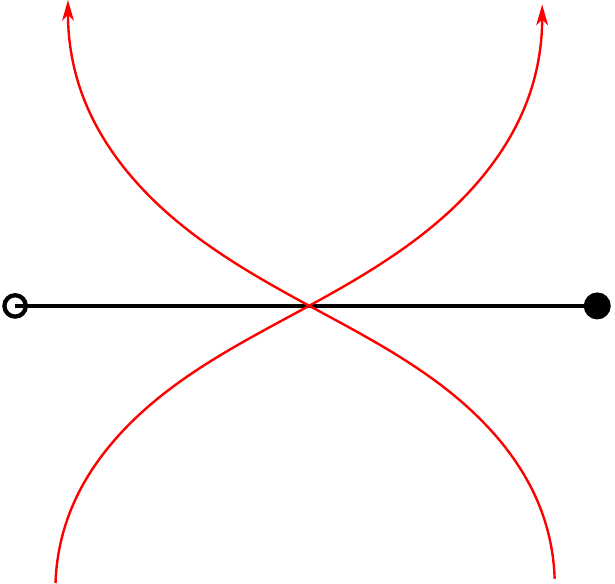}}%
    \put(-0.00230763,0.35244165){\color[rgb]{0,0,0}\makebox(0,0)[lt]{\lineheight{1.25}\smash{\begin{tabular}[t]{l}$w$\end{tabular}}}}%
    \put(0.94434895,0.35841433){\color[rgb]{0,0,0}\makebox(0,0)[lt]{\lineheight{1.25}\smash{\begin{tabular}[t]{l}$b$\end{tabular}}}}%
    \put(0.77114367,0.86011243){\color[rgb]{1,0,0}\makebox(0,0)[lt]{\lineheight{1.25}\smash{\begin{tabular}[t]{l}$\alpha$\end{tabular}}}}%
    \put(0.16193878,0.86011243){\color[rgb]{1,0,0}\makebox(0,0)[lt]{\lineheight{1.25}\smash{\begin{tabular}[t]{l}$\beta$\end{tabular}}}}%
    \put(0.40382899,0.31660609){\color[rgb]{0,0,0}\makebox(0,0)[lt]{\lineheight{1.25}\smash{\begin{tabular}[t]{l}$K_{wb}(\alpha,\beta)$\end{tabular}}}}%
    \put(0,0){\includegraphics[width=\unitlength,page=2]{Figures/FockWeight.pdf}}%
  \end{picture}%
\endgroup%

	\caption{Coefficient $K_{wb}(\alpha,\beta)$ of the Dirac operator on an edge $wb$}
	\label{Fig Fock coeff}
\end{figure}

The matrix $K_{WB}=(K_{wb}(\alpha,\beta))_{w\in W, \, b\in B}$ will play the role of the Kasteleyn matrix of a dimer model on $G$, as explained in Section~\ref{sect dimer models}, provided these weights are gauge equivalent to a system of positive weights. A construction enforcing this positivity property, based on $M$-curves $\mathcal{R}$, is the subject of Section 6.

In the next sections, we will be focusing on the integrability properties of the \emph{Dirac operator}
\begin{equation}\label{Dirac}
K_{WB}: \mathbb C^B\to\mathbb C^W
\end{equation}
defined, for $\psi\in\bbC^B$, by
\begin{equation}
(K_{WB} \psi)_w=\sum_{b\sim w} K_{wb}(\alpha,\beta)\psi_b.
\end{equation}
The study of integrability is based on the construction of \emph{Baker-Akhiezer functions} $(\psi_b(P))_{b\in B}$, which are functions on the Riemann surface $\mathcal{R}$. In the following their argument is denoted by $P\in \mathcal{R}$. These functions turn out to lie in the kernel of the Dirac operator $K_{WB}$.

We introduce more combinatorial structure for the set $B$. Namely, we define the \emph{black graph} $(B,E_B)$, by declaring $b_1,b_2\in B$ to be neighbors, connected by an edge $(b_1b_2)\in E_B$ , if there exists $w\in W$ such that the path $b_1wb_2$ is a part of the boundary of a face of $G$. 

The following two analytic facts are required for introducing the Baker-Akhiezer functions.

\begin{proposition}\label{Mumford on prime forms}
Let $\omega_{\alpha-\beta}$ be the unique Abelian differential of the third kind with vanishing $a$-periods and with exactly two poles at $\alpha$ and $\beta$ with the residues $+1$ resp. $-1$ at these points. Then
\begin{equation}
\exp\left(\int_Q^P \omega_{\alpha-\beta}\right)=\frac{E(P,\alpha)}{E(P,\beta)} \frac{E(Q,\beta)}{E(Q,\alpha)}.
\end{equation}
\end{proposition}

\begin{proposition}\label{BA mero}
For any $D\in J(\mathcal{R})$, $\alpha,\beta\in \mathcal{R}$, the following expression defines a genuine meromorphic function on $\mathcal{R}$:
\begin{equation}
\psi(P)=\frac{\theta\big(A(P)+A(\alpha)-A(\beta)+D\big)}{\theta\big(A(P)+D\big)}\ \frac{E(P,\alpha)}{E(P,\beta)}.
\end{equation}
\end{proposition}

\begin{definition}
For $b\in B$, we define the \emph{Baker-Akhiezer functions} $\psi_b:\mathcal{R}\to\mathbb C$ by the following rule. Fix some $b_0\in B$ and set $\psi_{b_0}(P)=1$. For any $b\in B$, let there be a path $b_0b_1\ldots b_n$ in the black graph connecting $b_0$ to $b_n=b$ (so that all $b_k\in B$, and there exist $w_k\in W$ such that each two-edge path $b_{k-1}w_kb_k$ is a part of the boundary of a face of $G$). Then, in notations of Fig.~\ref{Fig black transition},
\begin{equation}\label{black BA}
\psi_b(P)=\frac{\theta\big(A(P)+\eta(b)+D\big)}{\theta\big(A(P)+\eta(b_0)+D\big)}\prod_{k=1}^n\frac{E(P,\alpha_k^-)}{E(P,\alpha_k^+)}
\end{equation}
\end{definition}
Since 
\begin{equation}\label{discr Abel black}
\eta(b)-\eta(b_0)=\sum_{k=1}^n (A(\alpha_k^-)-A(\alpha_k^+)), 
\end{equation}
all $\psi_b(P)$ are, by Proposition~\ref{BA mero}, meromorphic functions on $\mathcal{R}$. The transition from $\psi_{b_{k-1}}(P)$ to $\psi_{b_k}(P)$ is described by the transition function
\begin{equation}\label{black BA transition}
\dfrac{\psi_{b_k}(P)}{\psi_{b_{k-1}}(P)}= \frac{\theta\big(A(P)+\eta(b_k)+D\big)}{\theta\big(A(P)+\eta(b_{k-1})+D\big)}\frac{E(P,\alpha_k^-)}{E(P,\alpha_k^+)},
\end{equation}
which adds a zero $\alpha_k^-$ and a pole $\alpha_k^+$ (all $\psi_b(P)$ have additionally poles at the zeros of $\theta(A(P)+\eta(b_0)+D)$, which form a fixed effective divisor governed by the point $D$).

\begin{rem}
    Note that this definition is independent of the choice of path $b_0 b_1 \ldots b_n$. Indeed, $\psi_{b_n}/\psi_{b_0}$ contains terms only for every train track separating $b_n$ and $b_0$.
\end{rem}

\begin{figure}
    \centering
    \fontsize{10pt}{12pt}\selectfont
    \def\svgwidth{0.4\linewidth}
\begingroup%
  \makeatletter%
  \providecommand\color[2][]{%
    \errmessage{(Inkscape) Color is used for the text in Inkscape, but the package 'color.sty' is not loaded}%
    \renewcommand\color[2][]{}%
  }%
  \providecommand\transparent[1]{%
    \errmessage{(Inkscape) Transparency is used (non-zero) for the text in Inkscape, but the package 'transparent.sty' is not loaded}%
    \renewcommand\transparent[1]{}%
  }%
  \providecommand\rotatebox[2]{#2}%
  \newcommand*\fsize{\dimexpr\f@size pt\relax}%
  \newcommand*\lineheight[1]{\fontsize{\fsize}{#1\fsize}\selectfont}%
  \ifx\svgwidth\undefined%
    \setlength{\unitlength}{540.6921315bp}%
    \ifx\svgscale\undefined%
      \relax%
    \else%
      \setlength{\unitlength}{\unitlength * \real{\svgscale}}%
    \fi%
  \else%
    \setlength{\unitlength}{\svgwidth}%
  \fi%
  \global\let\svgwidth\undefined%
  \global\let\svgscale\undefined%
  \makeatother%
  \begin{picture}(1,0.60420425)%
    \lineheight{1}%
    \setlength\tabcolsep{0pt}%
    \put(0,0){\includegraphics[width=\unitlength,page=1]{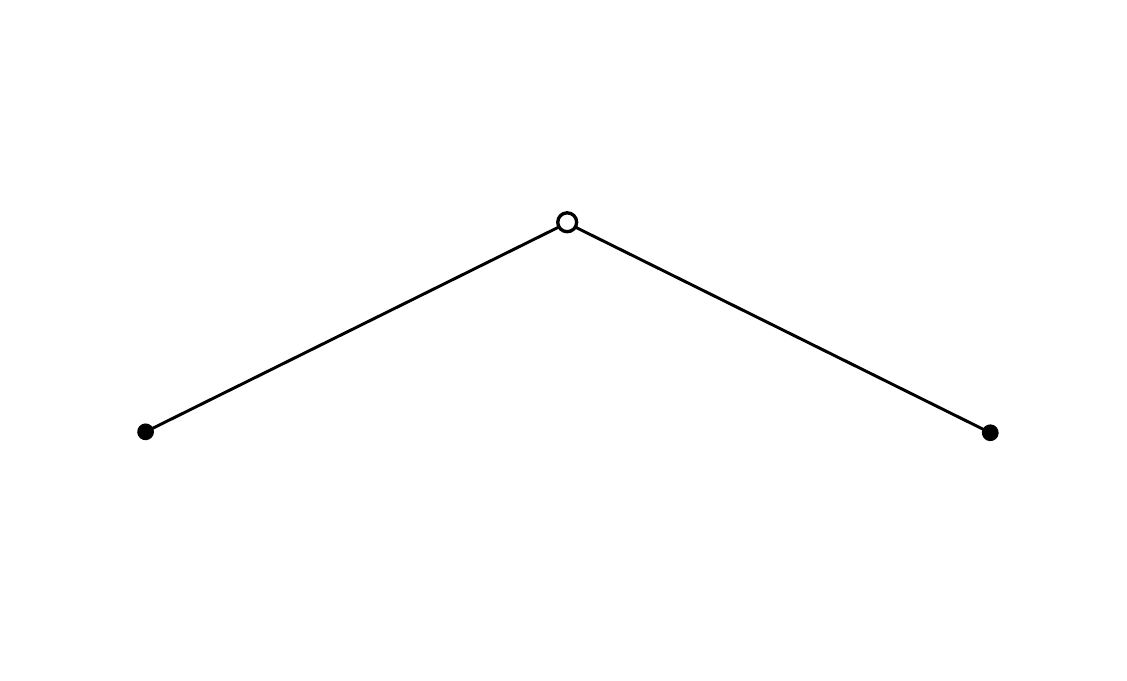}}%
    \put(0.10128983,0.16567287){\color[rgb]{0,0,0}\makebox(0,0)[lt]{\lineheight{1.25}\smash{\begin{tabular}[t]{l}$b_{k-1}$\end{tabular}}}}%
    \put(0.85433228,0.16567287){\color[rgb]{0,0,0}\makebox(0,0)[lt]{\lineheight{1.25}\smash{\begin{tabular}[t]{l}$b_k$\end{tabular}}}}%
    \put(0.48675117,0.44912836){\color[rgb]{0,0,0}\makebox(0,0)[lt]{\lineheight{1.25}\smash{\begin{tabular}[t]{l}$w$\end{tabular}}}}%
    \put(0,0){\includegraphics[width=\unitlength,page=2]{Figures/TransitionFunctionBlackToBlack.pdf}}%
    \put(0.93340874,0.37425249){\color[rgb]{1,0,0}\makebox(0,0)[lt]{\lineheight{1.25}\smash{\begin{tabular}[t]{l}$\beta$\end{tabular}}}}%
    \put(0.37626844,0.11355237){\color[rgb]{1,0,0}\makebox(0,0)[lt]{\lineheight{1.25}\smash{\begin{tabular}[t]{l}$\alpha^-_k$\end{tabular}}}}%
    \put(0.68768881,0.11355237){\color[rgb]{1,0,0}\makebox(0,0)[lt]{\lineheight{1.25}\smash{\begin{tabular}[t]{l}$\alpha^+_k$\end{tabular}}}}%
  \end{picture}%
\endgroup%

	\caption{Strips along a path in $G$ corresponding to an edge $(b_{k-1}b_k)$ in the black graph. The increment of the discrete Abel map $\eta(b_k)-\eta(b_{k-1})=A(\alpha_k^-)-A(\alpha_k^+)$ only depends on the strands $\alpha_k^-$ and $\alpha_k^+$ intersecting that edge}
	\label{Fig black transition}
\end{figure}

\begin{figure}[h]
    \centering
    \fontsize{10pt}{12pt}\selectfont
    \def\svgwidth{0.5\textwidth}
    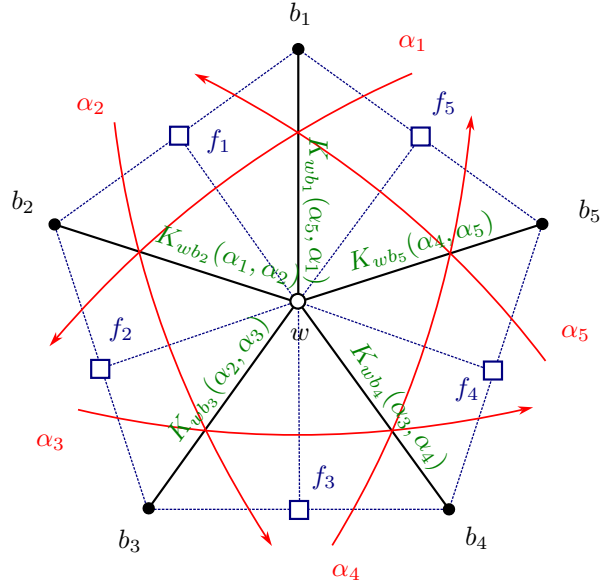
	\caption{Star of a white vertex}
	\label{Fig white star}
\end{figure}
\begin{theorem}\label{Th black Dirac}
For any $P\in \mathcal{R}$, the vector $(\psi_b(P))_{b\in B}\in\bbC^B$ is in the kernel of the Dirac operator $K_{WB}$. In other words: consider the star of a white vertex $w\in W$, consisting of $b_1,b_2,\ldots,b_n\in B$, see Fig.~\ref{Fig white star}. The associated Baker-Akhiezer functions satisfy the following relation: 
\begin{equation}\label{black Dirac}
\sum_{k=1}^n K_{wb_k}(\alpha_{k-1},\alpha_k)\psi_{b_k}(P)=0,
\end{equation}
where
\begin{equation}\label{black Dirac weights}
K_{wb_k}(\alpha_{k-1},\alpha_k)=\frac{E(\alpha_{k-1},\alpha_k)}{\theta\big(\eta(f_{k-1})+D\big)\theta\big(\eta(f_k)+D\big)}.
\end{equation}
\end{theorem}

\noindent
\emph{Proof.} For the sake of brevity, we write in this proof $\psi_k$ for $\psi_{b_k}$. Since the transition from $\psi_{k-1}$ to $\psi_k$ adds the zero $\alpha_{k-2}$ and the pole $\alpha_k$, and due to the corresponding cancellations, we find:
\begin{eqnarray*}
\frac{\psi_2(P)}{\psi_1(P)} & = & \frac{\theta\big(A(P)+\eta(b_2)+D\big)}{\theta\big(A(P)+\eta(b_1)+D\big)}\ \frac{E(P,\alpha_n)}{E(P,\alpha_2)}, \\
\frac{\psi_k(P)}{\psi_1(P)} & = & \frac{\theta\big(A(P)+\eta(b_k)+D\big)}{\theta\big(A(P)+\eta(b_1)+D\big)}\ \frac{E(P,\alpha_n)}{E(P,\alpha_k)}\ \frac{E(P,\alpha_1)}{E(P,\alpha_{k-1})} \quad {\rm for}\quad 3\le k\le n-1,\\
\frac{\psi_n(P)}{\psi_1(P)} & = & \frac{\theta\big(A(P)+\eta(b_n)+D\big)}{\theta\big(A(P)+\eta(b_1)+D\big)}\ \frac{E(P,\alpha_1)}{E(P,\alpha_{n-1})}.
\end{eqnarray*}
We determine the coefficients $c_2,\ldots,c_n$ in
$$
1+\sum_{k=2}^n c_k\frac{\psi_k(P)}{\psi_1(P)} = 0
$$
requiring that the residues at the poles $\alpha_2,\ldots,\alpha_{n-1}$ vanish (this gives $n-2$ conditions), then, upon multiplication with $\theta(A(P)+\eta(b_1)+D)$, this linear combination equals a constant, and then requiring that this linear combination vanishes at $\alpha_n$ (the $(n-1)$-st condition), then it vanishes identically. 

The point $\alpha_2$ is the pole of $\psi_2/\psi_1$ and of $\psi_3/\psi_1$, with the residues 
$$
\frac{\theta\big(\eta(b_2)+A(\alpha_2)+D\big)}{\theta\big(\eta(b_1)+A(\alpha_2)+D\big)}\ E(\alpha_2,\alpha_n)
\quad{\rm and}\quad
\frac{\theta\big(\eta(b_3)+A(\alpha_2)+D\big)}{\theta\big(\eta(b_1)+A(\alpha_2)+D\big)}\ \frac{E(\alpha_2,\alpha_n)}{E(\alpha_2,\alpha_3)}\ E(\alpha_2,\alpha_1),
$$
respectively. Therefore, we kill this pole by choosing
$$
\frac{c_2}{c_3}=\frac{\theta\big(\eta(b_3)+A(\alpha_2)+D\big)}{\theta\big(\eta(b_2)+A(\alpha_2)+D\big)}\ \frac{E(\alpha_1,\alpha_2)}{E(\alpha_2,\alpha_3)}
=\frac{\theta\big(\eta(f_3)+D\big)}{\theta\big(\eta(f_1)+D\big)}\ \frac{E(\alpha_1,\alpha_2)}{E(\alpha_2,\alpha_3)}. 
$$
For $3\le k\le n-2$, the point $\alpha_k$ is the pole of $\psi_k/\psi_1$ and of $\psi_{k+1}/\psi_1$, with the residues
$$
\frac{\theta\big(\eta(b_k)+A(\alpha_k)+D\big)}{\theta\big(\eta(b_1)+A(\alpha_k)+D\big)}\ E(\alpha_k,\alpha_n)\ \frac{E(\alpha_k,\alpha_1)}{E(\alpha_k,\alpha_{k-1})}
$$
and
$$
\frac{\theta\big(\eta(b_{k+1}+A(\alpha_k)+D\big)}{\theta\big(\eta(b_1)+A(\alpha_k)+D\big)}\ \frac{E(\alpha_k,\alpha_n)}{E(\alpha_k,\alpha_{k+1})}\ E(\alpha_k,\alpha_1),
$$
respectively. Therefore, we kill this pole by choosing
$$
\frac{c_k}{c_{k+1}}=\frac{\theta\big(\eta(b_{k+1})+A(\alpha_k)+D\big)}{\theta\big(\eta(b_k)+A(\alpha_k)+D\big)}\ \frac{E(\alpha_{k-1},\alpha_k)}{E(\alpha_k,\alpha_{k+1})}=\frac{\theta\big(\eta(f_{k+1})+D\big)}{\theta\big(\eta(f_{k-1})+D\big)}\ \frac{E(\alpha_{k-1},\alpha_k)}{E(\alpha_k,\alpha_{k+1})}.
$$
Finally, the point $\alpha_{n-1}$ is the pole of $\psi_{n-1}/\psi_1$ and of $\psi_n/\psi_1$, with the residues
$$
\frac{\theta\big(\eta(b_{n-1})+A(\alpha_{n-1})+D\big)}{\theta\big(\eta(b_1)+A(\alpha_{n-1})+D\big)}\ E(\alpha_{n-1},\alpha_n)\ \frac{E(\alpha_{n-1},\alpha_1)}{E(\alpha_{n-1},\alpha_{n-2})} 
$$
and
$$
\frac{\theta\big(\eta(b_n)+A(\alpha_{n-1})+D\big)}{\theta\big(\eta(b_{n-1})+A(\alpha_{n-1})+D\big)}\ E(\alpha_{n-1},\alpha_1).
$$
Therefore, we kill this pole by choosing
$$
\frac{c_{n-1}}{c_n}=
\frac{\theta\big(\eta(b_n)+A(\alpha_{n-1})+D\big)}{\theta\big(\eta(b_{n-1})+A(\alpha_{n-1})+D\big)}\ \frac{E(\alpha_{n-2},\alpha_{n-1})}{E(\alpha_{n-1},\alpha_n)}=
\frac{\theta\big(\eta(f_n)+D\big)}{\theta\big(\eta(f_{n-2})+D\big)}\ \frac{E(\alpha_{n-2},\alpha_{n-1})}{E(\alpha_{n-1},\alpha_n)}.
$$
The final condition of vanishing at $\alpha_n$ gives:
$$
1+c_n \frac{\theta\big(\eta(b_n)+A(\alpha_n)+D\big)}{\theta\big(\eta(b_1)+A(\alpha_n)+D\big)}\ \frac{E(\alpha_n,\alpha_1)}{E(\alpha_n,\alpha_{n-1})}=0,
$$
or
$$
c_n=\frac{\theta\big(\eta(b_1)+A(\alpha_n)+D\big)}{\theta\big(\eta(b_n)+A(\alpha_n)+D\big)}\ \frac{E(\alpha_{n-1},\alpha_n)}{E(\alpha_n,\alpha_1)}=
\frac{\theta\big(\eta(f_1)+D\big)}{\theta\big(\eta(f_{n-1})+D\big)}\ \frac{E(\alpha_{n-1},\alpha_n)}{E(\alpha_n,\alpha_1)}.
$$
Dividing all coefficients by a common factor $\theta\big(\eta(f_n)+D\big)\theta\big(\eta(f_1)+D\big)/E(\alpha_n,\alpha_1)$, we find:
$$
c_k=\frac{E(\alpha_{k-1},\alpha_k)}{\theta\big(\eta(f_{k-1})+D\big)\theta\big(\eta(f_k)+D\big)}, \quad k=1,\ldots,n. \qquad\qed
$$
\medskip

\begin{rem}
	One can introduce Baker-Akhiezer ``functions'' on white vertices, which will lie in the kernel of the adjoint Dirac operator with the Fock weights. However, these dual objects are meromorphic 1-forms rather than functions on $\mathcal{R}$.
\end{rem} 

\section{Fay identity and spider move}
\label{sec:4_fay_and_spider}

A particular case of Theorem~\ref{Th black Dirac} for $n=3$ leads to a simple proof of the famous Fay formula. In this case, a different enumeration of the strips and of the faces leads to a more symmetric notation, see Fig.~\ref{Fig Fay}.

\begin{figure}[h]
    \centering
    \fontsize{10pt}{12pt}\selectfont
    \def\svgwidth{0.6\textwidth}
    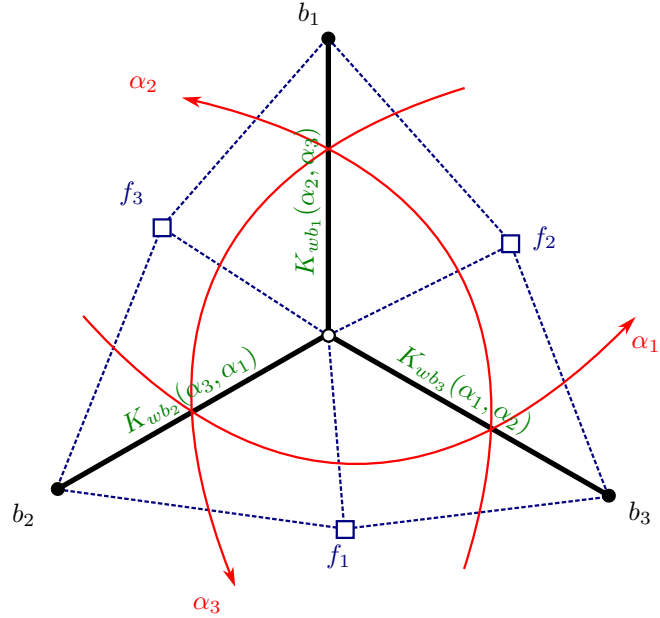
	\caption{Coeffficients of the Dirac operator on a triangle}
	\label{Fig Fay}
\end{figure}

Moreover, one gets a symmetric expression for the discrete Abel map on the faces of $G$, defined by
\begin{eqnarray*}
\eta(f_2)-\eta(f_1) & = & A(\alpha_1)-A(\alpha_2), \\ 
\eta(f_3)-\eta(f_3) & = & A(\alpha_2)-A(\alpha_3), \\
\eta(f_1)-\eta(f_3) & = & A(\alpha_3)-A(\alpha_1), 
\end{eqnarray*}
if one sets
$$
\eta(f_1)=A(\alpha_2)+A(\alpha_3), \quad \eta(f_2)=A(\alpha_3)+A(\alpha_1), \quad \eta(f_3)=A(\alpha_1)+A(\alpha_2).
$$
This corresponds also to
\begin{equation}\label{Fay dAbel b}
\eta(b_1)=A(\alpha_1), \quad \eta(b_2)=A(\alpha_2), \quad \eta(b_3)=A(\alpha_3).
\end{equation}
Thus, we get:

\begin{proposition}
The Baker-Akhiezer functions $\psi_1(P)$, $\psi_2(P)$, $\psi_3(P)$ are linearly dependent:
\begin{equation}\label{Fock Dirac triang}
K_{wb_1}(\alpha_2,\alpha_3)\psi_1+K_{wb_2}(\alpha_3,\alpha_1)\psi_2+K_{wb_3}(\alpha_1,\alpha_2)\psi_3=0,
\end{equation}
or, more detailed:
\begin{eqnarray} \label{Fay denom}
\frac{E(\alpha_2,\alpha_3)}{\theta\big(A(\alpha_1)+A(\alpha_2)+D\big)\theta\big(A(\alpha_1)+A(\alpha_3)+D\big)} \psi_1(P)\nonumber\\
+\frac{E(\alpha_3,\alpha_1)}{\theta\big(A(\alpha_2)+A(\alpha_3)+D\big)\theta\big(A(\alpha_1)+A(\alpha_2)+D\big)} \psi_2(P)\nonumber\\
+\frac{E(\alpha_1,\alpha_2)}{\theta\big(A(\alpha_1)+A(\alpha_3)+D\big)\theta\big(A(\alpha_2)+A(\alpha_3)+D\big)} \psi_3(P) & = & 0. \qquad
\end{eqnarray}
\end{proposition}

Clearing denominators, this can be written as
\begin{eqnarray} \label{Fay num}
\theta\big(A(\alpha_2)+A(\alpha_3)+D\big)E(\alpha_2,\alpha_3) \psi_1(P)\nonumber\\
+\theta\big(A(\alpha_1)+A(\alpha_3)+D\big)E(\alpha_3,\alpha_1) \psi_2(P)\nonumber\\
+\theta\big(A(\alpha_1)+A(\alpha_2)+D\big)E(\alpha_1,\alpha_2) \psi_3(P) & = & 0. \qquad
\end{eqnarray}
\medskip

\begin{rem}
    Recall that, according to \eqref{black BA transition} and \eqref{Fay dAbel b},
    \begin{eqnarray}\label{Fay psi}
    \frac{\psi_2(P)}{\psi_1(P)} & = & \frac{\theta\big(A(P)+A(\alpha_2)+D\big)}{\theta\big(A(P)+A(\alpha_1)+D\big)}\ \frac{E(P,\alpha_2)}{E(P,\alpha_1)}, \nonumber\\
    \frac{\psi_3(P)}{\psi_1(P)} & = & \frac{\theta\big(A(P)+A(\alpha_3)+D\big)}{\theta\big(A(P)+A(\alpha_1)+D\big)}\ \frac{E(P,\alpha_3)}{E(P,\alpha_1)}.
    \end{eqnarray}
    Substituting \eqref{Fay psi} into \eqref{Fay num} and clearing denominators again, we find the \emph{Fay identity} in the following form:
    \begin{eqnarray} \label{Fay bilinear}
    \theta\big(A(\alpha_2)+A(\alpha_3)+D\big)\theta\big(A(P)+A(\alpha_1)+D\big)E(\alpha_2,\alpha_3)E(P,\alpha_1) \nonumber\\
    +\theta\big(A(\alpha_1)+A(\alpha_3)+D\big)\theta\big(A(P)+A(\alpha_2)+D\big)E(\alpha_3,\alpha_1)E(P,\alpha_2) \nonumber\\
    +\theta\big(A(\alpha_1)+A(\alpha_2)+D\big)\theta\big(A(P)+A(\alpha_3)+D\big)E(\alpha_1,\alpha_2)E(P,\alpha_3)  & = & 0, \qquad
    \end{eqnarray}
    where the four points $P,\alpha_1,\alpha_2,\alpha_3$ play symmetric roles.
\end{rem}

\begin{figure}[h]
\centering
\begin{subfigure}[p]{.49\textwidth}
	\centering
    \fontsize{10pt}{12pt}\selectfont
    \def\svgwidth{0.9\linewidth}
    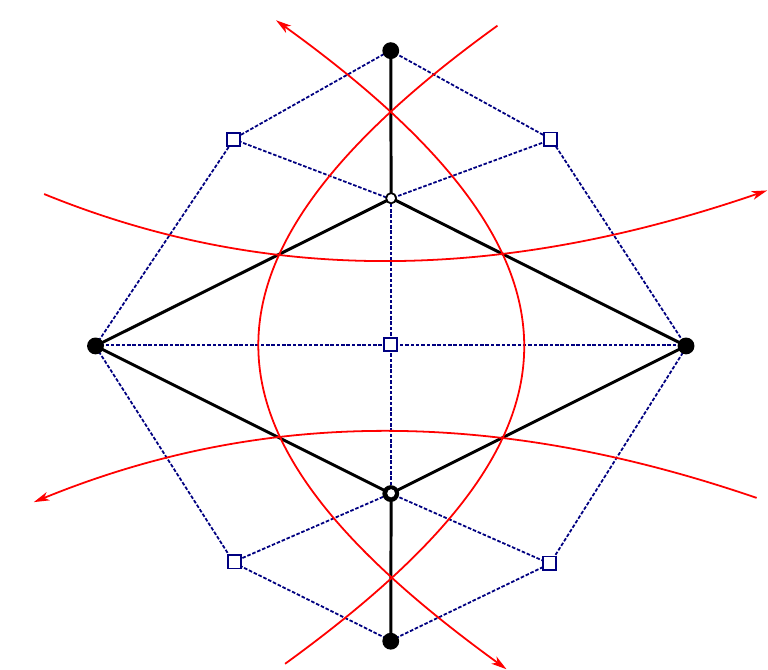
\end{subfigure}
\begin{subfigure}[p]{.49\textwidth}
	\centering
    \fontsize{10pt}{12pt}\selectfont
    \def\svgwidth{0.9\linewidth}
    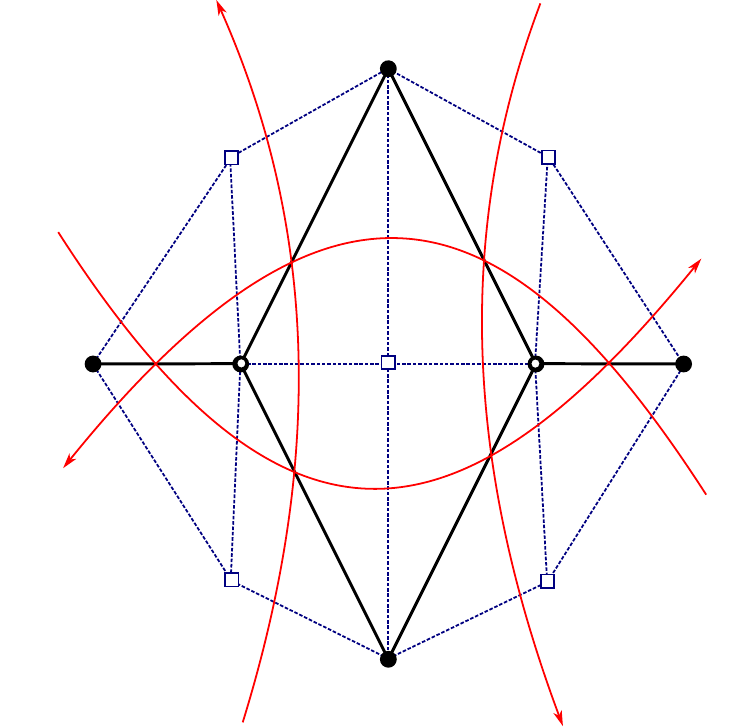
\end{subfigure}
\caption{Spider move}
\label{Fig spider}
\end{figure}

\medskip
\begin{proposition}\label{prop spider}
Equation \eqref{Fock Dirac triang} is consistent with the spider move.
\end{proposition}

The statement means that for the system of Dirac equations on the black graph, solutions do not change under local changes of the graph called `´spider moves'' and depicted on Fig.~\ref{Fig spider}. This is a local statement, and to demonstrate it, one can consider suitable boundary value problems for two systems consisting of equations \eqref{Fock Dirac triang} on the both graphs, and show that they admit the same solution. A suitable boundary problem, defining the solution uniquely, consists in prescribing the values at any two points connected by an edge of the black graph, say at $b_1$ and $b_2$. To show is that the solutions at $b_3$ and $b_4$ for both graphs coincide.

\begin{proof}
The argument is a combination of three observations.
\begin{itemize}
\item The discrete Abel maps at the black vertices of the both graphs coincide. This follows from eq. \eqref{discr Abel black}. 
\item The BA functions $\psi_b(P)$ provide us then with an infinite-dimensional family of (linearly independent) solutions of the system of equations \eqref{Fock Dirac triang} on both graphs. 
\item For each of the graphs, the space of solutions with prescribed initial data at $b_1$, $b_2$ is two-dimensional, therefore the solutions can be represented as linear combination of two linearly independent solutions from the infinite-dimensional family of BA functions. Therefore the solutions of the boundary value problem at $b_3$, $b_4$ coincide for both graphs. 
\end{itemize}
\end{proof}

This argument (solutions of suitable boundary value problems for two systems of linear equations coincide, if these systems admit a common infinite-dimensional family of linearly independent solutions) works in similar situations below, for instance, for the proof of the 3D consistency of the system of Dirac equations on the square lattice in the next section.

\section{Dirac operator on quad-graphs and multidimensional consistency} 
\label{sec:5_quadGraphs_and_consistency}

Consider the graph $G$ for which the black vertices form the lattice $\mathbb Z^2$, and the white lattices form the lattice $(\mathbb Z+\frac{1}{2})^2$, compare Fig.~\ref{Fig sq lattice}.


Setting the zero of the discrete Abel map at $b_0=(0,0)$, one easily computes, with the help of eq. \eqref{discr Abel black},
$$
\eta(b_{mn})=\sum_{i=0}^{m-1} A(\alpha_i^--\alpha_i^+)+\sum_{j=0}^{n-1} A(\beta_j^--\beta_j^+).
$$
The Baker-Akhiezer functions on $B=\mathbb Z^2$ are given by
\begin{equation}\label{BA Z^2}
\psi_{mn}(P)=\frac{\theta\big(A(P)+\eta(b_{mn})+D\big)}{\theta\big(A(P)+D\big)}
\prod_{i=0}^{m-1}\frac{E(P,\alpha_i^-)}{E(P,\alpha_i^+)}\prod_{j=0}^{n-1}\frac{R(P,\beta_j^-)}{E(P,\beta_j^+)}.
\end{equation}
On each elementary square of $B=\mathbb Z^2$, the BA functions $\psi_{mn}(P)$ satisfy the four-point Dirac equation (quad-equation) \eqref{black Dirac}. This system of quad-equations in the case $g=1$ was the subject of \cite{BS_linear}. In this case, if all $\alpha_i^{\pm}=\alpha^{\pm}$, $\beta_j^{\pm}=\beta^{\pm}$ are independent of $i$, resp. of $j$, and if $2A(\alpha^--\alpha^+)=2A(\beta^--\beta^+)\equiv 0$ in $J(\mathcal{R})$, then the coefficients of the quad-equations are 2-periodic in both directions. 

Remarkably, one can extend the construction of BA functions \eqref{BA Z^2} to a multi-dimensional lattice $\mathbb Z^d$. Each edge $e$ of $\mathbb Z^d$ is labelled by a two-point divisor $\alpha_e^--\alpha_e^+$, so that opposite edges of any elementary square have the same labels. Thus, now the strips of equally labeled edges are associated with hyperplanes orthogonally bisecting the edges. For any black point $b\in \mathbb  Z^d$ connected to $b_0=0$ by a path of edges $e_1,\ldots, e_p$, the discrete Abel map is given by
$$
\eta(b)=\sum_{k=1}^p A(\alpha_{e_k}^--\alpha_{e_k}^+).
$$
Obviously, it is independent of the path. The same is true for the BA functions, which are defined by
$$
\psi_b(P)=\frac{\theta\big(A(P)+\eta(b)+D\big)}{\theta\big(A(P)+D\big)}
\prod_{k=1}^{p}\frac{E(P,\alpha_{e_k}^-)}{E(P,\alpha_{e_k}^+)}.
$$
Quadruples of the BA functions at the vertices elementary squares of $\mathbb Z^d$ still satisfy four-point Dirac equations \eqref{black Dirac}. The system of Dirac equations is multi-dimensionally consistent in the sense of \cite{bobenko_discrete_2008}. As explained in this book, multi-dimensional consistency of quad-equations essentially follows from their 3D consistency.

\begin{proposition}
The system of four-point Dirac equations on $\mathbb Z^d$ is 3D consistent.
\end{proposition}

The statement means that a suitable boundary value problem for the system consisting of four-point Dirac equations on the six faces of a 3D cube admits a unique solution. The labelling for this system is shown on Fig.~\ref{Fig 3D consistency}. A suitable boundary problem consists in prescribing four values, say at $b_0=(000)$, $b_1=(100)$, $b_2=(010)$ and $b_{13}=(101)$. To show is that the solutions at $b_{23}=(011)$ and $b_{123}=(111)$ for both graphs coincide.

\begin{figure}[h]
\centering
\begin{subfigure}[p]{.45\textwidth}
    \centering
    \fontsize{10pt}{12pt}\selectfont
    \def\svgwidth{0.9\linewidth}
    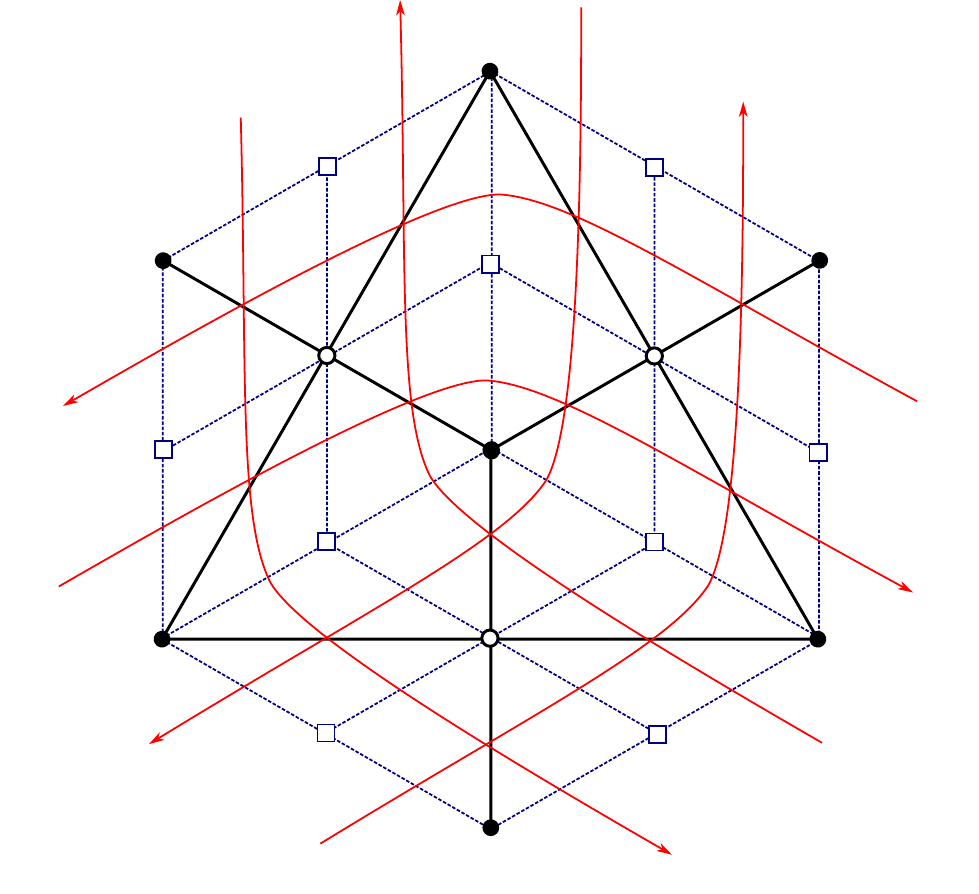
\end{subfigure}
\begin{subfigure}[p]{.45\textwidth}
    \centering
    \fontsize{10pt}{12pt}\selectfont
    \def\svgwidth{0.9\linewidth}
    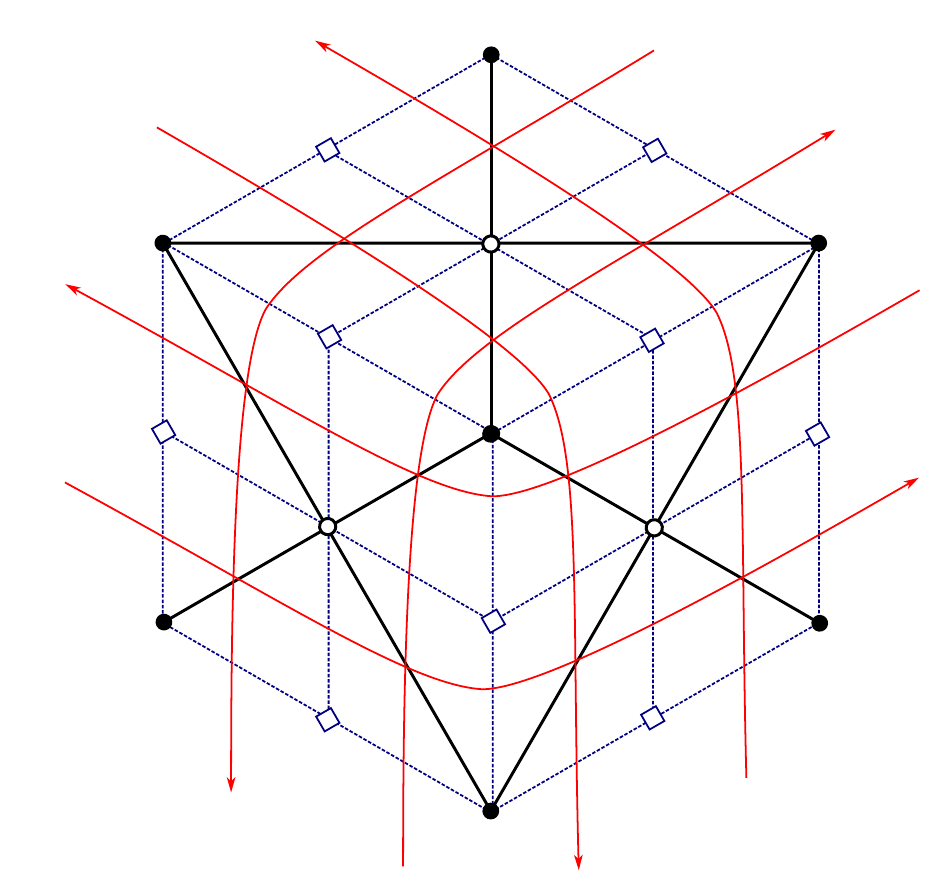
\end{subfigure}
\caption{3D consistency}
\label{Fig 3D consistency}
\end{figure}

\begin{proof} It goes along the same lines as the proof of Proposition \ref{prop spider}.
\begin{itemize}
\item The discrete Abel maps at the black vertices of the both graphs coincide, as pointed out above. 
\item The BA functions $\psi_b(P)$ provide us with an infinite-dimensional family of (linearly independent) solutions of the system of Dirac equations on all six faces of the cube. 
\item For each of the graphs, the space of solutions with prescribed initial data at $b_0$, $b_1$, $b_2$ and $b_{13}$ is four-dimensional, therefore the solutions can be represented as linear combination of four linearly independent solutions from the infinite-dimensional family of BA functions. Therefore the solutions of the boundary value problem at the common vertices $b_{12}=(110)$ and $b_{23}=(011)$ coincide for both graphs. 
\end{itemize}
\end{proof}

Since every quasicrystallic rhombic embedding can be realized as a combinatorial surface in $\mathbb Z^d$ whose faces are elementary squares of $\mathbb Z^d$, we can restrict the above construction to any such surface. Obviously, functions $\psi_b(P)$ at the vertices of every face of the surface satisfy the four-point Dirac equation. Strips on the surface are the intersections of combinatorial coordinate hyperplanes in $\mathbb Z^d$ with the surface.

Such surfaces (with their respective systems of strips and Dirac equations) can be transformed one into another by local flips like on Fig. \ref{Fig 3D consistency}. These flips do not affect physical properties of the corresponding dimer model.

\section{Kasteleyn property of Fock weights from M-curves}
\label{sec:6_M-curves}

We now turn to the study of the dimer model on $G$ with the Fock weights given by \eqref{Fock coeff}. Recall that there exists a gauge equivalent real-valued weight function whose signs define a Kasteleyn orientation, if for every face of $G$ condition \eqref{Kasteleyn cond} is satisfied. In the present context, including the labeling of strands by points of $\mathcal{R}$, this condition looks as follows:
\begin{equation}\label{Kasteleyn cond 1}
{\rm sign}\left(\frac{K_{w_1b_1}(\alpha_1,\beta_1)}{K_{w_2b_1}(\beta_2,\alpha_1)} \cdot\frac{K_{w_2b_2}(\alpha_2,\beta_2)}{K_{w_3b_2}(\beta_3,\alpha_2)}\cdots \frac{K_{w_nb_n}(\alpha_n,\beta_n)}{K_{w_1b_n}(\beta_1,\alpha_{n})}\right)=(-1)^{n+1}.
\end{equation}
for any face of the graph $G$, see Fig.~\ref{Fig face}.

\begin{figure}[h]
    \centering
    \fontsize{10pt}{12pt}\selectfont
    \def\svgwidth{0.5\textwidth}
    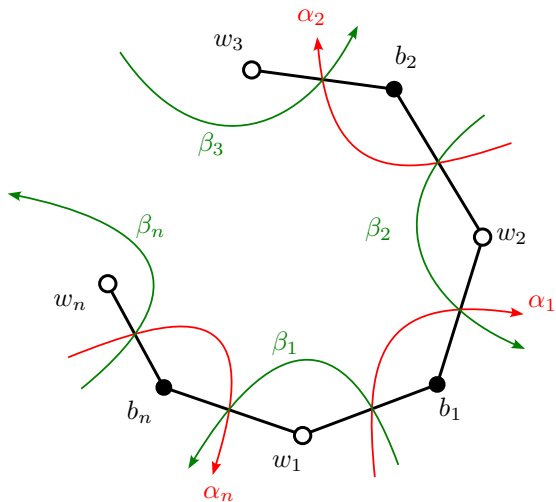
  \caption{A face of the graph $G$}
  \label{Fig face}
\end{figure}

We show how construct the Fock weights with the Kasteleyn property from $M$-curves.

\begin{definition}
A Riemann surface $\mathcal{R}$ of genus $g$ with an anti-holomorphic involution $\tau:\mathcal{R}\to \mathcal{R}$ is called an $M$-\emph{curve} if it has the maximal possible number $g+1$ of real ovals $X_0,\ldots,X_g$, see Fig.~\ref{Fig M-curve}.  
\end{definition}
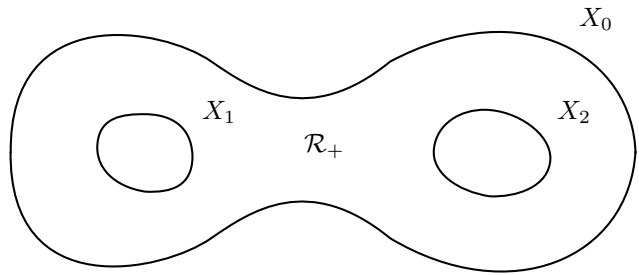
\begin{figure}[h]
  \centering
  \fontsize{10pt}{12pt}\selectfont
  \def\svgwidth{0.5\textwidth}
\begingroup%
  \makeatletter%
  \providecommand\color[2][]{%
    \errmessage{(Inkscape) Color is used for the text in Inkscape, but the package 'color.sty' is not loaded}%
    \renewcommand\color[2][]{}%
  }%
  \providecommand\transparent[1]{%
    \errmessage{(Inkscape) Transparency is used (non-zero) for the text in Inkscape, but the package 'transparent.sty' is not loaded}%
    \renewcommand\transparent[1]{}%
  }%
  \providecommand\rotatebox[2]{#2}%
  \newcommand*\fsize{\dimexpr\f@size pt\relax}%
  \newcommand*\lineheight[1]{\fontsize{\fsize}{#1\fsize}\selectfont}%
  \ifx\svgwidth\undefined%
    \setlength{\unitlength}{344.28048744bp}%
    \ifx\svgscale\undefined%
      \relax%
    \else%
      \setlength{\unitlength}{\unitlength * \real{\svgscale}}%
    \fi%
  \else%
    \setlength{\unitlength}{\svgwidth}%
  \fi%
  \global\let\svgwidth\undefined%
  \global\let\svgscale\undefined%
  \makeatother%
  \begin{picture}(1,0.41784042)%
    \lineheight{1}%
    \setlength\tabcolsep{0pt}%
    \put(0,0){\includegraphics[width=\unitlength,page=1]{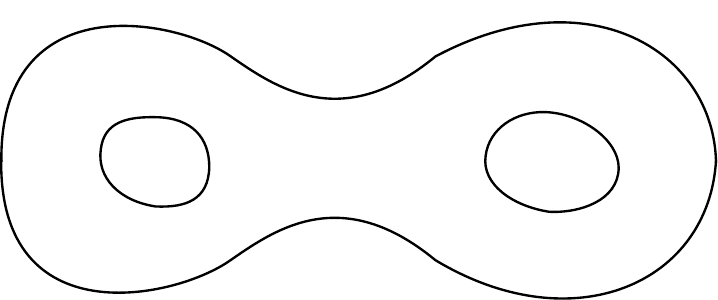}}%
    \put(0.90843802,0.39607704){\color[rgb]{0,0,0}\makebox(0,0)[lt]{\lineheight{1.25}\smash{\begin{tabular}[t]{l}$X_0$\end{tabular}}}}%
    \put(0.30707433,0.24634845){\color[rgb]{0,0,0}\makebox(0,0)[lt]{\lineheight{1.25}\smash{\begin{tabular}[t]{l}$X_1$\end{tabular}}}}%
    \put(0.87336741,0.24634845){\color[rgb]{0,0,0}\makebox(0,0)[lt]{\lineheight{1.25}\smash{\begin{tabular}[t]{l}$X_2$\end{tabular}}}}%
    \put(0.47011299,0.19367473){\color[rgb]{0,0,0}\makebox(0,0)[lt]{\lineheight{1.25}\smash{\begin{tabular}[t]{l}$\mathcal{R}_+$\end{tabular}}}}%
  \end{picture}%
\endgroup%

  \caption{Upper half $\mathcal{R}_+$ of an $M$-curve $\mathcal{R}$.}
  \label{Fig M-curve}
\end{figure}

Let $\mathcal{R}$ be an $M$-curve. Choose the canonical base of cycles: $a_1,\ldots,a_g,b_1,\ldots,b_g$ with the intersection numbers $a_j\circ a_k=b_j\circ b_k=0$, $a_j\circ b_k=\delta_{jk}$ as follows: $a_1=X_1, \ldots,a_g=X_g$, so that $\tau(a_j)=a_j$, and $\tau(b_j)=-b_j$. Choose a basis of holomorphic differentials $\omega_1,\ldots,\omega_g$ on $\mathcal{R}$ normalized by the condition $\int_{a_k}\omega_j=\delta_{jk}$, then $\overline{\tau^*\omega_j}=\omega_j$ on $a_k$, while the $b$-periods $B_{jk}=\int_{b_k}\omega_j$ are purely imaginary. As usual, the matrix $B=(B_{jk})_{j,k=1}^g$ is symmetric with a positive definite imaginary part (which is in this case just the matrix $-iB$).

\begin{lemma}\label{lemma thetas real}
\quad

\begin{itemize}
\item For any $x\in\mathbb R^g$, we have $\theta(x)>0$.
\item For any $x\in\mathbb R^g$ and for any half-integer characteristic $\Delta$, we have $\theta[\Delta](x)\in\mathbb R$.
\end{itemize}
\end{lemma}
\noindent
\begin{proof}
The fact that $\theta(x)\in \mathbb R$ for $x\in\mathbb R^g$ follows from $\overline{\theta(z)}=\theta(-\bar{z})=\theta(\bar{z})$ (the latter equality due to $\theta(z)$ being an even function). To prove the stronger positivity statement, we use the multi-dimensional analog of the imaginary Jacobi transformation:
$$
\theta(x|B)=\frac{1}{\sqrt{\det(-iB)}}\exp\big(-\pi i \langle x,B^{-1}x\rangle\big)\theta(B^{-1}x|-B^{-1}).
$$
Here the first factor is positive, since the matrix $-iB$ is positive definite, the second factor is positive, since the matrix $-iB^{-1}$ is real, and all terms in the defining series of the third factor,
$$
\theta(B^{-1}x|-B^{-1})=\sum_{m\in\mathbb Z^g}\exp\Big(-\pi \langle iB^{-1}m,m\rangle+2\pi \langle iB^{-1}x,m\rangle\Big),
$$ 
 are positive for the same reason. To prove the second statement, we derive from \eqref{theta with char}, \eqref{theta monodromy} that $\overline{\theta[\Delta](z)}=\theta[\Delta](\bar{z})$. 
\end{proof}
%

Choose $\alpha_i,\beta_i\in X_0$ (all points associated to strips are on the real oval which does not belong to the chosen homology basis of $\mathcal{R}$). Choose the initial integration point of the Abel map $P_0\in X_0$, as well. Then all $A(\alpha_i), A(\beta_i)\in \mathbb R^g$. This implies that the discrete Abel map $\eta$ is $\mathbb R^g$-valued.  Further, choose $D\in\mathbb R^g$. Then all theta-functions $\theta(\eta(f)+D)$ in the denominators of the Fock weights are positive, according to Lemma~\ref{lemma thetas real}. The Kasteleyn condition \eqref{Kasteleyn cond 1} turns into

\begin{align}\label{Kasteleyn cond 2}
& {\rm sign}\left(\frac{E(\alpha_1,\beta_1)E(\alpha_2,\beta_2)\cdots E(\alpha_n,\beta_n)}{E(\beta_2,\alpha_1)E(\beta_3,\alpha_2)\cdots E(\beta_1,\alpha_n)}\right) \nonumber\\
& ={\rm sign}\left(\frac{\theta[\Delta]\Big(\int_{\alpha_1}^{\beta_1}\boldsymbol{\omega}\Big)
\theta[\Delta]\Big(\int_{\alpha_2}^{\beta_2}\boldsymbol{\omega}\Big)\cdots
\theta[\Delta]\Big(\int_{\alpha_n}^{\beta_n}\boldsymbol{\omega}\Big)}
{\theta[\Delta]\Big(\int_{\beta_2}^{\alpha_1}\boldsymbol{\omega}\Big)
\theta[\Delta]\Big(\int_{\beta_3}^{\alpha_2}\boldsymbol{\omega}\Big)\cdots 
\theta[\Delta]\Big(\int_{\beta_1}^{\alpha_{n}}\boldsymbol{\omega}\Big)}\right)  =  (-1)^{n+1}.
\end{align}
The expression in the parentheses on the l.h.s. is real by Lemma~\ref{lemma thetas real}.

\begin{theorem}
Consider an orientation preserving homeomorphism $X_0\to S^1$, and denote by $\hat\alpha_i,\hat\beta_i\in S^1$ the images of $\alpha_i,\beta_i$ under this map. Then the Kasteleyn condition is equivalent to
\begin{equation}\label{Kasteleyn cond 3}
{\rm sign}\left(\frac{(\hat\beta_1-\hat\alpha_1)(\hat\beta_2-\hat\alpha_2)\cdots(\hat\beta_n-\hat\alpha_n)}
{(\hat\alpha_1-\hat\beta_2)(\hat\alpha_2-\hat\beta_3)\cdots (\hat\alpha_{n}-\hat\beta_1)}\right)=(-1)^{n+1}
\end{equation}
for all faces $f$ of $G$.
\end{theorem}
\begin{proof}
The expression on the l.h.s. of \eqref{Kasteleyn cond 2} does not change sign if one moves $\alpha_i,\beta_i$ on $X_0$ preserving their cyclic order. Indeed, by such transformations $\theta[\Delta]\Big(\int_{\alpha_i}^{\beta_i}\boldsymbol{\omega}\Big)$ never vanish. One can move all the points (preserving their cyclic order) in a small neighborhood of some $P\in X_0$, covered by some coordinate chart $z$ of $\mathcal{R}$ with $\tau^*z=\bar z$. Setting $\tilde\alpha_i=z(\alpha_i), \tilde\beta_i=z(\beta_i)\in \mathbb R$, we obtain for the expression in question a formula like \eqref{Kasteleyn cond 3} with hats replaces by tildes. But this multi-ratio is, first,  M\"obius invariant and second, does not change sign under moving the points on $S^1$ while preserving their cyclic order.
\end{proof}

\subsection{Example: regular hexagonal lattice}

See Fig.~\ref{Fig hex lattice} for notations. Considering a typical hexagon, we see that the Kasteleyn condition reads:
\begin{equation}
{\rm sign} \frac{(\hat\gamma_1-\hat\beta_1)(\hat\alpha_2-\hat\gamma_2)(\hat\beta_2-\hat\alpha_1)}{(\hat\beta_1-\hat\alpha_2)(\hat\gamma_2-\hat\beta_2)(\hat\alpha_1-\hat\gamma_1)}=1.
\end{equation}

\begin{figure}[h]
\centering
\begin{subfigure}[p]{.7\textwidth}
  \centering
    \fontsize{10pt}{12pt}\selectfont
    \def\svgwidth{0.9\linewidth}
    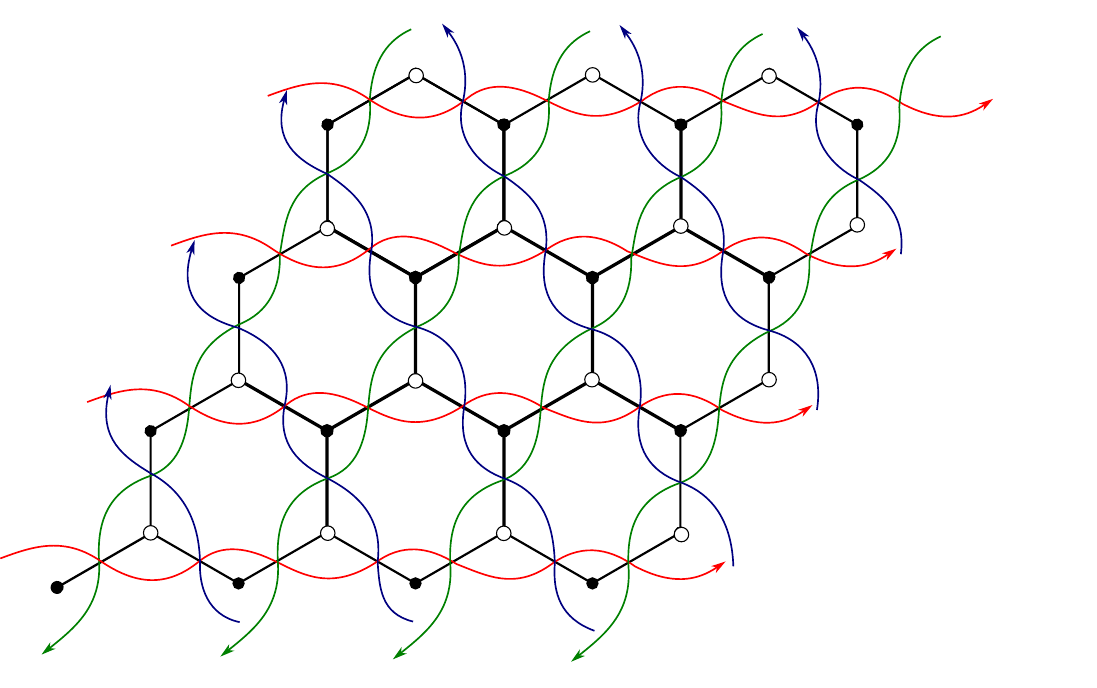
\end{subfigure}
\begin{subfigure}[p]{.28\textwidth}
  \centering
    \fontsize{10pt}{12pt}\selectfont
    \def\svgwidth{0.9\linewidth}
\begingroup%
  \makeatletter%
  \providecommand\color[2][]{%
    \errmessage{(Inkscape) Color is used for the text in Inkscape, but the package 'color.sty' is not loaded}%
    \renewcommand\color[2][]{}%
  }%
  \providecommand\transparent[1]{%
    \errmessage{(Inkscape) Transparency is used (non-zero) for the text in Inkscape, but the package 'transparent.sty' is not loaded}%
    \renewcommand\transparent[1]{}%
  }%
  \providecommand\rotatebox[2]{#2}%
  \newcommand*\fsize{\dimexpr\f@size pt\relax}%
  \newcommand*\lineheight[1]{\fontsize{\fsize}{#1\fsize}\selectfont}%
  \ifx\svgwidth\undefined%
    \setlength{\unitlength}{298.53233119bp}%
    \ifx\svgscale\undefined%
      \relax%
    \else%
      \setlength{\unitlength}{\unitlength * \real{\svgscale}}%
    \fi%
  \else%
    \setlength{\unitlength}{\svgwidth}%
  \fi%
  \global\let\svgwidth\undefined%
  \global\let\svgscale\undefined%
  \makeatother%
  \begin{picture}(1,0.98848374)%
    \lineheight{1}%
    \setlength\tabcolsep{0pt}%
    \put(0,0){\includegraphics[width=\unitlength,page=1]{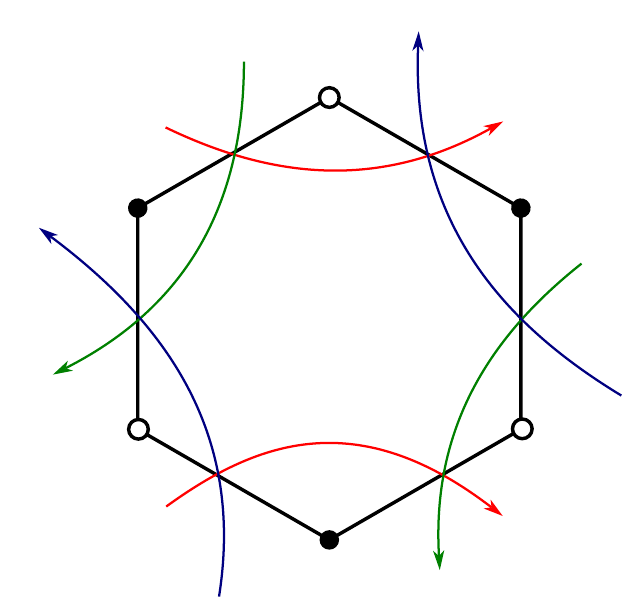}}%
    \put(-0.00046807,0.65043186){\color[rgb]{0,0,0.50196078}\makebox(0,0)[lt]{\lineheight{1.25}\smash{\begin{tabular}[t]{l}$\gamma_1$\end{tabular}}}}%
    \put(0.63933933,0.96347071){\color[rgb]{0,0,0.50196078}\makebox(0,0)[lt]{\lineheight{1.25}\smash{\begin{tabular}[t]{l}$\gamma_2$\end{tabular}}}}%
    \put(0.82469128,0.12046251){\color[rgb]{1,0,0}\makebox(0,0)[lt]{\lineheight{1.25}\smash{\begin{tabular}[t]{l}$\beta_1$\end{tabular}}}}%
    \put(0.83155622,0.77674581){\color[rgb]{1,0,0}\makebox(0,0)[lt]{\lineheight{1.25}\smash{\begin{tabular}[t]{l}$\beta_2$\end{tabular}}}}%
    \put(0.03111038,0.32778213){\color[rgb]{0,0.50196078,0}\makebox(0,0)[lt]{\lineheight{1.25}\smash{\begin{tabular}[t]{l}$\alpha_1$\end{tabular}}}}%
    \put(0.68464776,0.00787842){\color[rgb]{0,0.50196078,0}\makebox(0,0)[lt]{\lineheight{1.25}\smash{\begin{tabular}[t]{l}$\alpha_2$\end{tabular}}}}%
  \end{picture}%
\endgroup%

\end{subfigure}
\caption{Regular hexagonal lattice on a $3\times 3$ fundamental domain (left) and the strips for a typical hexagonal face (right).}
\label{Fig hex lattice}
\end{figure}

This condition is satisfied if the points $\hat\alpha_i,\hat\beta_i,\hat\gamma_i$ lie on $S^1$ in three clusters, so that, with respect to the cyclic order on $S^1$,
\begin{equation}\label{clusters hex}
\hat\alpha_i<\hat\beta_j<\hat\gamma_k<\hat\alpha_i \quad {\rm for\;\;all}\quad i,j,k,
\end{equation}
see Fig.~\ref{Fig hex clusters}.

\begin{figure}[h]
  \centering
  \fontsize{10pt}{12pt}\selectfont
  \def\svgwidth{0.3\textwidth}
\begingroup%
  \makeatletter%
  \providecommand\color[2][]{%
    \errmessage{(Inkscape) Color is used for the text in Inkscape, but the package 'color.sty' is not loaded}%
    \renewcommand\color[2][]{}%
  }%
  \providecommand\transparent[1]{%
    \errmessage{(Inkscape) Transparency is used (non-zero) for the text in Inkscape, but the package 'transparent.sty' is not loaded}%
    \renewcommand\transparent[1]{}%
  }%
  \providecommand\rotatebox[2]{#2}%
  \newcommand*\fsize{\dimexpr\f@size pt\relax}%
  \newcommand*\lineheight[1]{\fontsize{\fsize}{#1\fsize}\selectfont}%
  \ifx\svgwidth\undefined%
    \setlength{\unitlength}{310.10995671bp}%
    \ifx\svgscale\undefined%
      \relax%
    \else%
      \setlength{\unitlength}{\unitlength * \real{\svgscale}}%
    \fi%
  \else%
    \setlength{\unitlength}{\svgwidth}%
  \fi%
  \global\let\svgwidth\undefined%
  \global\let\svgscale\undefined%
  \makeatother%
  \begin{picture}(1,0.76822248)%
    \lineheight{1}%
    \setlength\tabcolsep{0pt}%
    \put(0,0){\includegraphics[width=\unitlength,page=1]{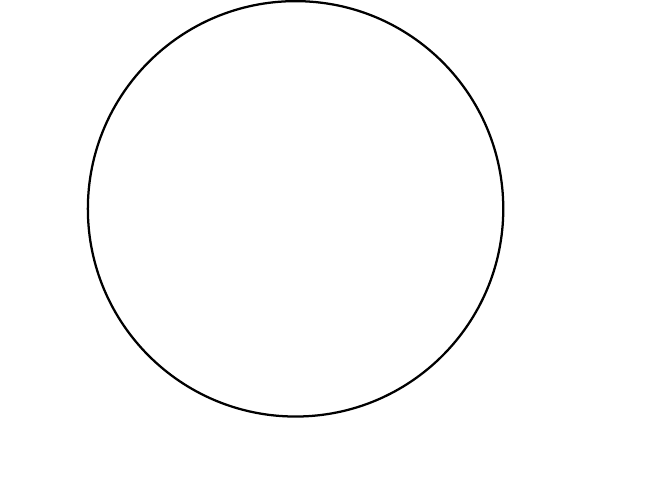}}%
    \put(0.13039287,0.17394255){\color[rgb]{0,0,0}\makebox(0,0)[lt]{\lineheight{1.25}\smash{\begin{tabular}[t]{l}$X_0$\end{tabular}}}}%
    \put(0,0){\includegraphics[width=\unitlength,page=2]{Figures/AnglesConditionOnCircleHex.pdf}}%
    \put(0.85750893,0.64591364){\color[rgb]{1,0,0}\makebox(0,0)[lt]{\lineheight{1.25}\smash{\begin{tabular}[t]{l}$\beta_j$\end{tabular}}}}%
    \put(0.65283217,0.00758429){\color[rgb]{0,0.50196078,0}\makebox(0,0)[lt]{\lineheight{1.25}\smash{\begin{tabular}[t]{l}$\alpha_i$\end{tabular}}}}%
    \put(-0.0004506,0.65712879){\color[rgb]{0,0,0.50196078}\makebox(0,0)[lt]{\lineheight{1.25}\smash{\begin{tabular}[t]{l}$\gamma_k$\end{tabular}}}}%
  \end{picture}%
\endgroup%

  \caption{Clusters of points on $X_0\simeq S^1$ corresponding to the strips of the regular hexagonal lattice, ensuring the Kasteleyn property}
  \label{Fig hex clusters}
\end{figure}

\subsection{Example: regular square lattice}

See Fig.~\ref{Fig sq lattice} for notations. There are two types of elementary squares. For the first one, the Kasteleyn condition reads:
\begin{equation}
{\rm sign} \frac{(\hat\alpha_i^+-\hat\beta_k^-)(\hat\alpha_j^--\hat\beta_k^+)}{(\hat\beta_k^--\hat\alpha_j^-)(\hat\beta_k^+-\hat\alpha_i^+)}=-1.
\end{equation}
For the second one, the Kasteleyn condition reads:
\begin{equation}
{\rm sign} \frac{(\hat\beta_j^--\hat\alpha_i^+)(\hat\beta_k^+-\hat\alpha_i^-)}{(\hat\alpha_i^+-\hat\beta_k^+)(\hat\alpha_i^--\hat\beta_j^-)}=-1.
\end{equation}

\begin{figure}[h]
\centering
\begin{subfigure}[p]{.7\textwidth}
  \centering
    \fontsize{10pt}{12pt}\selectfont
    \def\svgwidth{0.9\linewidth}
    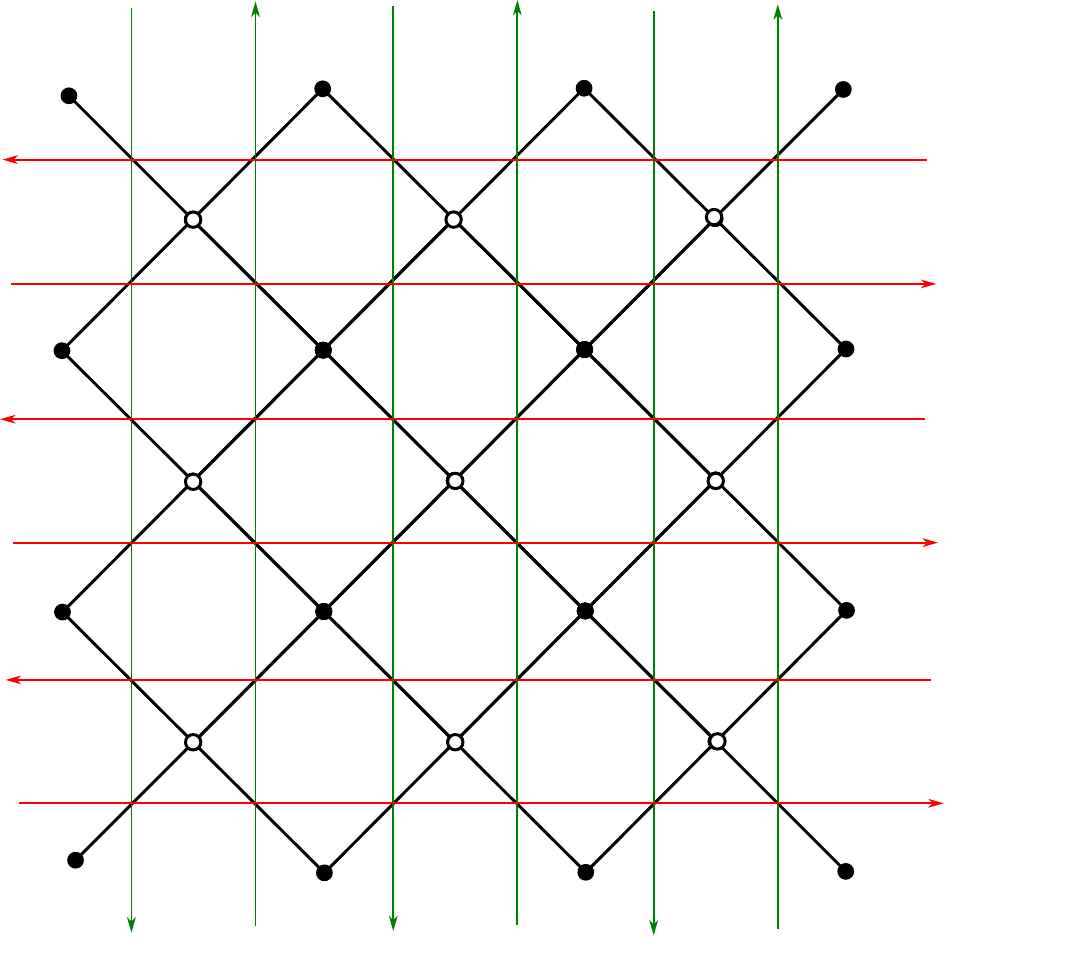
\end{subfigure}
\begin{subfigure}[p]{.28\textwidth}
  \begin{subfigure}[t]{\textwidth}
    \centering
      \fontsize{10pt}{12pt}\selectfont
      \def\svgwidth{0.9\linewidth}
\begingroup%
  \makeatletter%
  \providecommand\color[2][]{%
    \errmessage{(Inkscape) Color is used for the text in Inkscape, but the package 'color.sty' is not loaded}%
    \renewcommand\color[2][]{}%
  }%
  \providecommand\transparent[1]{%
    \errmessage{(Inkscape) Transparency is used (non-zero) for the text in Inkscape, but the package 'transparent.sty' is not loaded}%
    \renewcommand\transparent[1]{}%
  }%
  \providecommand\rotatebox[2]{#2}%
  \newcommand*\fsize{\dimexpr\f@size pt\relax}%
  \newcommand*\lineheight[1]{\fontsize{\fsize}{#1\fsize}\selectfont}%
  \ifx\svgwidth\undefined%
    \setlength{\unitlength}{322.25247532bp}%
    \ifx\svgscale\undefined%
      \relax%
    \else%
      \setlength{\unitlength}{\unitlength * \real{\svgscale}}%
    \fi%
  \else%
    \setlength{\unitlength}{\svgwidth}%
  \fi%
  \global\let\svgwidth\undefined%
  \global\let\svgscale\undefined%
  \makeatother%
  \begin{picture}(1,0.85773971)%
    \lineheight{1}%
    \setlength\tabcolsep{0pt}%
    \put(0,0){\includegraphics[width=\unitlength,page=1]{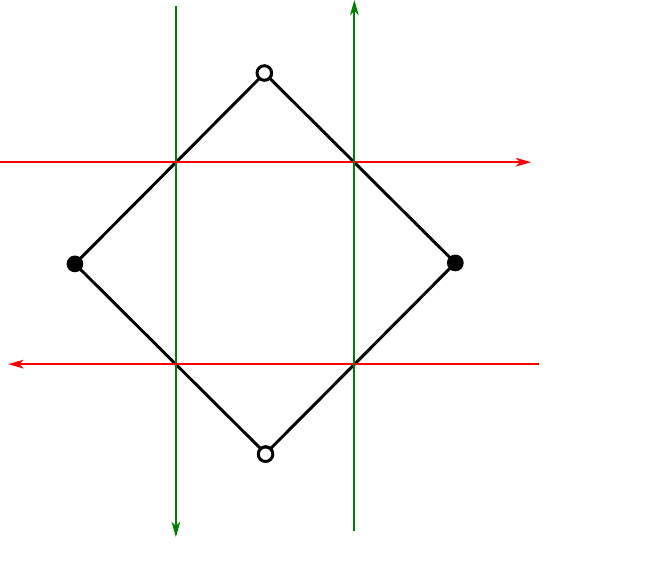}}%
    \put(0.75295355,0.6568396){\color[rgb]{1,0,0}\makebox(0,0)[lt]{\lineheight{1.25}\smash{\begin{tabular}[t]{l}$\beta_{j+1}^-$\end{tabular}}}}%
    \put(0.75676921,0.35412353){\color[rgb]{1,0,0}\makebox(0,0)[lt]{\lineheight{1.25}\smash{\begin{tabular}[t]{l}$\beta_j^+$\end{tabular}}}}%
    \put(0.23273982,0.0068903){\color[rgb]{0,0.50196078,0}\makebox(0,0)[lt]{\lineheight{1.25}\smash{\begin{tabular}[t]{l}$\alpha_i^-$\end{tabular}}}}%
    \put(0.51256123,0.00561841){\color[rgb]{0,0.50196078,0}\makebox(0,0)[lt]{\lineheight{1.25}\smash{\begin{tabular}[t]{l}$\alpha_i^+$\end{tabular}}}}%
  \end{picture}%
\endgroup%

  \end{subfigure}
  \begin{subfigure}[b]{\textwidth}
    \centering
      \fontsize{10pt}{12pt}\selectfont
      \def\svgwidth{0.9\linewidth}
\begingroup%
  \makeatletter%
  \providecommand\color[2][]{%
    \errmessage{(Inkscape) Color is used for the text in Inkscape, but the package 'color.sty' is not loaded}%
    \renewcommand\color[2][]{}%
  }%
  \providecommand\transparent[1]{%
    \errmessage{(Inkscape) Transparency is used (non-zero) for the text in Inkscape, but the package 'transparent.sty' is not loaded}%
    \renewcommand\transparent[1]{}%
  }%
  \providecommand\rotatebox[2]{#2}%
  \newcommand*\fsize{\dimexpr\f@size pt\relax}%
  \newcommand*\lineheight[1]{\fontsize{\fsize}{#1\fsize}\selectfont}%
  \ifx\svgwidth\undefined%
    \setlength{\unitlength}{301.01809837bp}%
    \ifx\svgscale\undefined%
      \relax%
    \else%
      \setlength{\unitlength}{\unitlength * \real{\svgscale}}%
    \fi%
  \else%
    \setlength{\unitlength}{\svgwidth}%
  \fi%
  \global\let\svgwidth\undefined%
  \global\let\svgscale\undefined%
  \makeatother%
  \begin{picture}(1,0.91824627)%
    \lineheight{1}%
    \setlength\tabcolsep{0pt}%
    \put(0,0){\includegraphics[width=\unitlength,page=1]{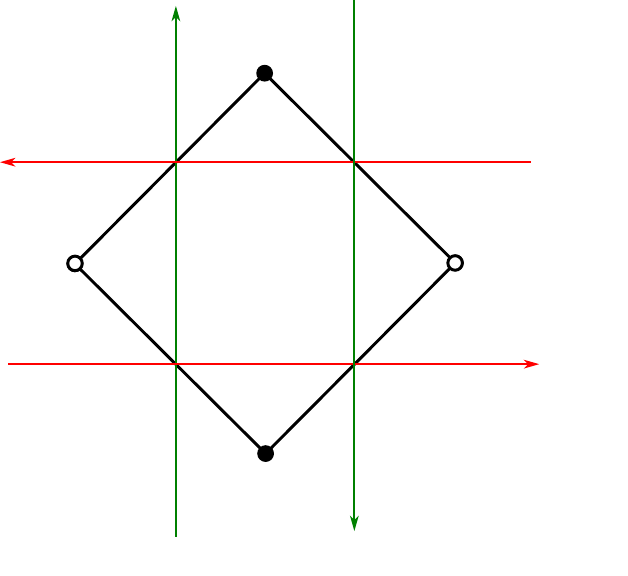}}%
    \put(0.8060683,0.70317428){\color[rgb]{1,0,0}\makebox(0,0)[lt]{\lineheight{1.25}\smash{\begin{tabular}[t]{l}$\beta_j^+$\end{tabular}}}}%
    \put(0.81015312,0.37910406){\color[rgb]{1,0,0}\makebox(0,0)[lt]{\lineheight{1.25}\smash{\begin{tabular}[t]{l}$\beta_j^-$\end{tabular}}}}%
    \put(0.24915772,0.00737635){\color[rgb]{0,0.50196078,0}\makebox(0,0)[lt]{\lineheight{1.25}\smash{\begin{tabular}[t]{l}$\alpha_i^+$\end{tabular}}}}%
    \put(0.54871825,0.00601474){\color[rgb]{0,0.50196078,0}\makebox(0,0)[lt]{\lineheight{1.25}\smash{\begin{tabular}[t]{l}$\alpha_{i+1}^-$\end{tabular}}}}%
  \end{picture}%
\endgroup%

\end{subfigure}
\end{subfigure}
\caption{Regular square lattice on a $3 \times 3$ fundamental domain (left) and the strips for two types of square faces (right).}
\label{Fig sq lattice}
\end{figure}

Both conditions are satisfied if the points $\hat\alpha_i^\pm,\hat\beta_i^\pm$ lie on $S^1$ in four clusters, so that, with respect to the cyclic order on $S^1$,
\begin{equation}\label{clusters sq}
\hat\alpha_i^+<\hat\beta_j^+<\hat\alpha_k^-<\hat\beta_l^-<\hat\alpha_i^+ \quad {\rm for\;\;all}\quad i,j,k,l,
\end{equation}
see Fig.~\ref{Fig sq clusters}.

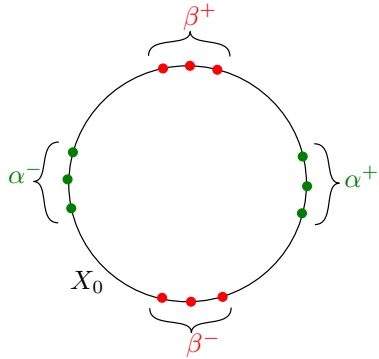
\begin{figure}[h]
  \centering
  \fontsize{10pt}{12pt}\selectfont
  \def\svgwidth{0.3\textwidth}
\begingroup%
  \makeatletter%
  \providecommand\color[2][]{%
    \errmessage{(Inkscape) Color is used for the text in Inkscape, but the package 'color.sty' is not loaded}%
    \renewcommand\color[2][]{}%
  }%
  \providecommand\transparent[1]{%
    \errmessage{(Inkscape) Transparency is used (non-zero) for the text in Inkscape, but the package 'transparent.sty' is not loaded}%
    \renewcommand\transparent[1]{}%
  }%
  \providecommand\rotatebox[2]{#2}%
  \newcommand*\fsize{\dimexpr\f@size pt\relax}%
  \newcommand*\lineheight[1]{\fontsize{\fsize}{#1\fsize}\selectfont}%
  \ifx\svgwidth\undefined%
    \setlength{\unitlength}{522.44217592bp}%
    \ifx\svgscale\undefined%
      \relax%
    \else%
      \setlength{\unitlength}{\unitlength * \real{\svgscale}}%
    \fi%
  \else%
    \setlength{\unitlength}{\svgwidth}%
  \fi%
  \global\let\svgwidth\undefined%
  \global\let\svgscale\undefined%
  \makeatother%
  \begin{picture}(1,0.8976315)%
    \lineheight{1}%
    \setlength\tabcolsep{0pt}%
    \put(0,0){\includegraphics[width=\unitlength,page=1]{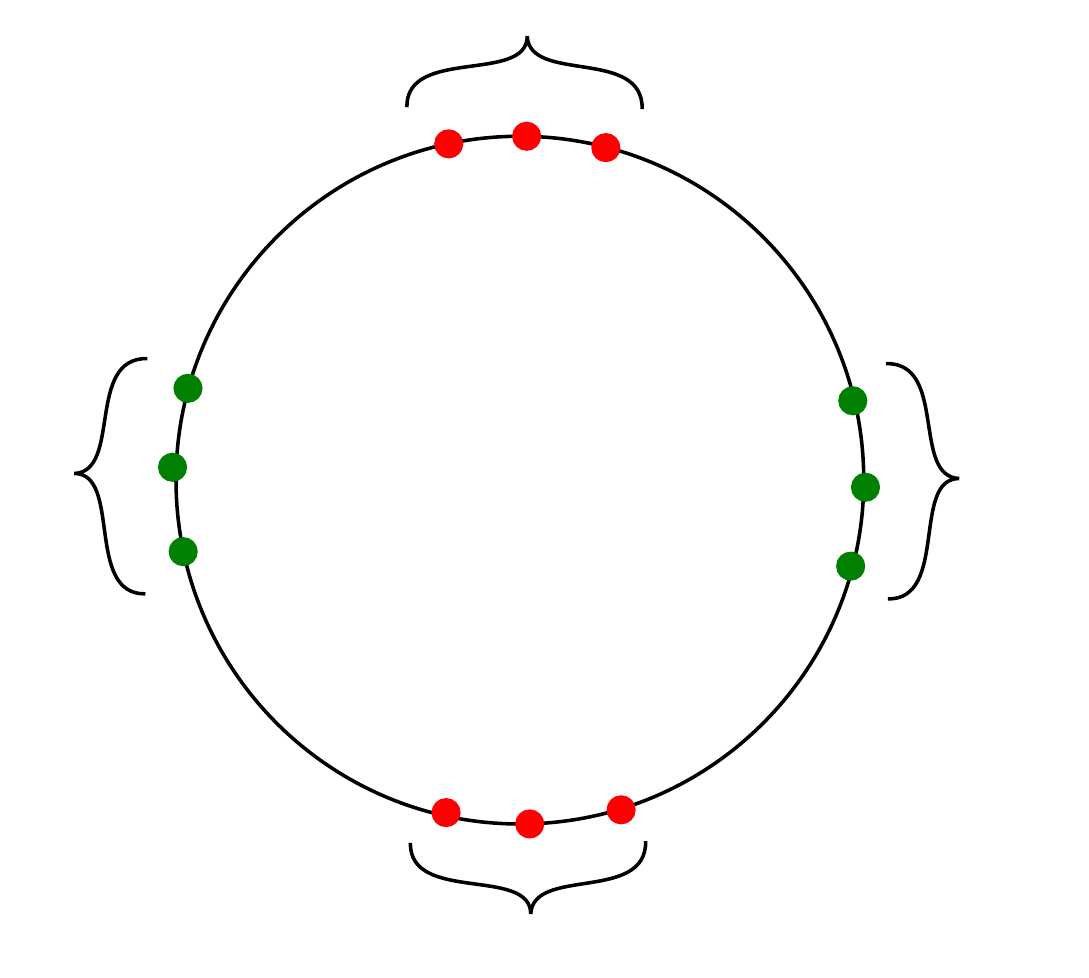}}%
    \put(0.00032394,0.4579532){\color[rgb]{0,0.50196078,0}\makebox(0,0)[lt]{\lineheight{1.25}\smash{\begin{tabular}[t]{l}$\alpha^-$\end{tabular}}}}%
    \put(0.47380203,0.00466429){\color[rgb]{1,0,0}\makebox(0,0)[lt]{\lineheight{1.25}\smash{\begin{tabular}[t]{l}$\beta^-$\end{tabular}}}}%
    \put(0.89459098,0.44901107){\color[rgb]{0,0.50196078,0}\makebox(0,0)[lt]{\lineheight{1.25}\smash{\begin{tabular}[t]{l}$\alpha^+$\end{tabular}}}}%
    \put(0.46571,0.88274722){\color[rgb]{1,0,0}\makebox(0,0)[lt]{\lineheight{1.25}\smash{\begin{tabular}[t]{l}$\beta^+$\end{tabular}}}}%
    \put(0.1655946,0.17459862){\color[rgb]{0,0,0}\makebox(0,0)[lt]{\lineheight{1.25}\smash{\begin{tabular}[t]{l}$X_0$\end{tabular}}}}%
  \end{picture}%
\endgroup%

  \caption{Clusters of points on $X_0\simeq S^1$ corresponding to the strips of the regular square lattice, ensuring the Kasteleyn property}
  \label{Fig sq clusters}
\end{figure}

\section{Graphs on a torus and Harnack curves}
\label{sec:7_Torus_and_Harnack}

Starting from here, we consider a $\mathbb Z^2$-periodic graph $G$ on the plane $\mathbb R^2$. We can consider such a graph as a lift to $\mathbb R^2$ of the graph $G_1=G/\mathbb Z^2$ on the torus $\mathbb R^2/\mathbb Z^2$.

Every strip of $G$ projects to a directed closed curve (loop) on the torus. Let $(m_\sigma,n_\sigma)$ be the homology class of a loop $\sigma$, so that  $\gamma_1=[(0,0),(1,0)]\in H_1(\mathbb T^2,\mathbb Z)$ intersects the loop $\sigma$ from left to right $n_\sigma$ times, while $\gamma_2=[(0,0),(0,1)]\in H_1(\mathbb T^2,\mathbb Z)$ intersects the loop $\sigma$ from right to left $m_\sigma$ times. One has
\begin{equation}\label{00}
\sum_\sigma n_\sigma=0, \quad \sum_\sigma m_\sigma=0.
\end{equation}

As seen in \cite{goncharov_dimers_2013}, one can draw the vectors $(m_\sigma,n_\sigma)$ in $\mathbb Z^2$ in a uniquely defined order so that 
\begin{itemize}
\item the initial point of each vector is the end point of the previous one, and 
\item they build the boundary of a convex polygon, the \emph{Newton polygon} $N(G_1)$. 
\end{itemize}
Parallel strips correspond to vectors aligned along the same side of the Newton polygon. The criterium for the Kasteleyn property for the periodic case can be formulated as follows.
\begin{theorem}\label{th alphas vs Newton polygon}
A sufficient condition for the Fock weights defined by an $M$-curve to satisfy the Kasteleyn property is that the points $\alpha_\sigma$ are ordered on $X_0$ in the same cyclic order as the corresponding directed edges of the Newton polygon $N(G_1)$, with clusters of $\alpha_\sigma\in X_0$ corresponding to the sides of $N(G_1)$. 
\end{theorem} 

One easily sees:
\begin{proposition} \label{prop periodic Dirac}
The discrete Abel map $\eta$, and therefore the discrete Dirac operator on $G$ are $\mathbb Z^2$-periodic iff 
\begin{equation}\label{00 in J}
\sum_{\sigma} n_\sigma A(\alpha_\sigma)= 0\; in\; J(\mathcal{R}), \quad \sum_{\sigma} m_\sigma A(\alpha_\sigma)= 0\; in\; J(\mathcal{R}),
\end{equation}
i.e., if both divisors $\sum_\sigma n_\sigma\alpha_\sigma$ and $\sum_\sigma m_\sigma\alpha_\sigma$ are principal (are divisors of certain meromorphic functions on $\mathcal{R}$).
\end{proposition}
Moreover, we see from \eqref{black BA} that these  meromorphic functions are nothing but the monodromies of the Baker-Akhiezer functions along the loops $\gamma_1,\gamma_2$:
\begin{equation}
\label{monodromies}
z(P):=\frac{\psi_{(1,0)}(P)}{\psi_{(0,0)}(P)}=\prod_{\sigma} (E(P,\alpha_\sigma))^{-n_\sigma}, \quad w(P):=\frac{\psi_{(0,1)}(P)}{\psi_{(0,0)}(P)}=\prod_{\sigma} (E(P,\alpha_\sigma))^{m_\sigma}.
\end{equation}
By the general theory, the points $(z(P),w(P))\in\mathbb CP^1\times \mathbb CP^1$ lie on an algebraic curve $\mathcal{P}(z,w)=0$, called the \emph{spectral curve} and parametrized by the Riemann surface $\mathcal{R}$ and having therefore geometric genus $g$ (but usually a much higher arithmetic genus due to numerous singularities). 

For an arbitrary Dirac operator $K:\bbC^B\to\bbC^W$ on a bipartite $\bbZ^2$-periodic graph with $\bbZ^2$-periodic coefficients, the equation $\mathcal{P}(z,w)=0$ of the spectral curve can be obtained as follows. Choose a fundamental domain $\Omega\subset \bbR^2$ (so that the translations $\Omega+\bbZ^2$ are disjoint and their union is $\bbR^2$). Set $\widetilde B=B\cap \Omega$ and $\widetilde W=W\cap\Omega$. These sets represent the black and the white vertices of the graph $G_1$ on $\mathbb T^2$. We now define an operator $\widetilde K(z,w): \bbC^{\widetilde B}\to\bbC^{\widetilde W}$ with coefficients depending on $z,w$ such that
\begin{equation}
\label{eq:spectral_curve}
\mathcal{P}(z,w)=\det \widetilde K(z,w).
\end{equation}
For each $w\in\widetilde W$, consider its star in $\bbR^2$ consisting of $b_1,b_2,\ldots,b_n\in B$ and the corresponding Dirac equation \eqref{black Dirac}. For those black vertices $b_k$ lying outside $\Omega$ (i.e., for such that the edge $(wb_k)$ intersects the boundary $\partial\Omega$), determine $(i,j)\in \bbZ^2$ so that $\widetilde b_k:=b_k-(i,j)\in\Omega$. Set in the Dirac equation $\psi_{b_k}=z^iw^j\psi_{\widetilde b_k}$. This is equivalent to multiplying the weight $K_{wb_k}$ by $z^iw^j$. Thus, the Dirac equation becomes
$$
\sum_{\widetilde b_k\sim w} \widetilde K_{w\widetilde b_k}\psi_{\widetilde b_k}=0, \quad {\rm where}\quad  \widetilde K_{w\widetilde b_k}=z^iw^jK_{wb_k},
$$
and the incidence $\widetilde b_k\sim w$ is understood relative to the graph $G_1$ on $\mathbb T^2$. A different choice of the fundamental domain leads to the multiplication of $\det(\widetilde K)$ by a factor $z^kw^\ell$ and does not chage the spectral curve.

According to \cite{Kenyon_Okounkov_2006} if the coefficients of the Dirac operator satisfy the Kasteleyn condition then the spectral curve \eqref{eq:spectral_curve} is a \emph{Harnack curve}. Moreover all Harnack curves can be obtained this way. We use this characterizatiuon of the Harnack curves as M-curve $\mathcal{R}$ with two principal divisors $\sum_\sigma n_\sigma\alpha_\sigma$ and $\sum_\sigma m_\sigma\alpha_\sigma$ satisfying clustering conditions. 

\subsection{Example: regular hexagonal lattice}

Consider the case of a $d\times d$ piece of the regular hexagonal lattice as the fundamental domain (in Fig.~\ref{Fig hex lattice} $d=3$ and one should assume $\alpha_4=\alpha_1$,  $\beta_4=\beta_1$, $\gamma_4=\gamma_1$; the coordinate directions are defined by the shifts $(0,0)\to (1,0)$ resp. $(0,0)\to (0,1)$). The homology class of all $\alpha_i$-strips is $(0,-1)$, the homology class of all $\beta_j$-strips is $(1,0)$, and the homology class of all $\gamma_k$-strips is $(-1,1)$. The Newton polygon is shown in Fig.~\ref{Fig Newton hex}. 
\begin{figure}[h]
    \centering
    \fontsize{10pt}{12pt}\selectfont
    \def\svgwidth{0.3\textwidth}
    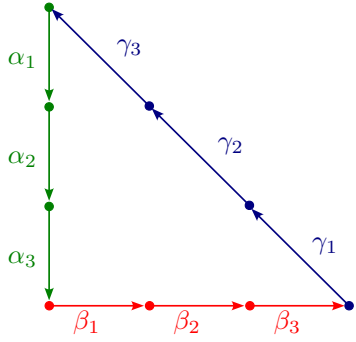
  \caption{Newton polygon for a $3\times 3$ fundamental domain of the hexagonal lattice}
  \label{Fig Newton hex}
\end{figure}

The monodromies are
\begin{equation}
\label{monodromies hex}
z(P)=\frac{\psi_{(1,0)}(P)}{\psi_{(0,0)}(P)}=\prod_{k=1}^d \frac{E(P,\alpha_k)}{E(P,\gamma_k)}, \quad w(P)=\frac{\psi_{(0,1)}(P)}{\psi_{(0,0)}(P)}=\prod_{k=1}^d \frac{E(P,\beta_k)}{E(P,\gamma_k)}.
\end{equation}
The periodicity condition of the Dirac operator is
\begin{equation}
\sum_{k=1}^d (A(\alpha_k)-A(\gamma_k)) = 0\; {\rm in}\; J(\mathcal{R}), \quad \sum_{k=1}^d (A(\beta_k)-A(\gamma_k)) = 0\; {\rm in}\; J(\mathcal{R}),
\end{equation}
so that both divisors $\sum_{k=1}^d (\alpha_k-\gamma_k)$ and $\sum_{k=1}^d (\beta_k-\gamma_k)$ are principal. Thus, the problem consists in finding three equivalent divisors $\sum_{k=1}^d \alpha_k$, $\sum_{k=1}^d\beta_k$ and $\sum_{k=1}^d\gamma_k$ of points on the oval $X_0$ of the $M$-curve $\mathcal{R}$ satisfying clustering condition \eqref{clusters hex}.

\subsection{Example: regular square lattice}

Consider the case of a $d\times d$ piece of the regular square lattice as the fundamental domain (in Fig.~\ref{Fig sq lattice} $d=3$; the coordinate directions are defined by the shifts $(0,0)\to (1,0)$ resp. $(0,0)\to (0,1)$). The homology class of all $\alpha_i^\pm$-strips is $(0,\pm 1)$, the homology class of all $\beta_j^\pm $-strips is $(\mp 1,0)$. The Newton polygon is shown in Fig.~\ref{Fig Newton sq}. 
\begin{figure}[h]
    \centering
    \fontsize{10pt}{12pt}\selectfont
    \def\svgwidth{0.3\textwidth}
    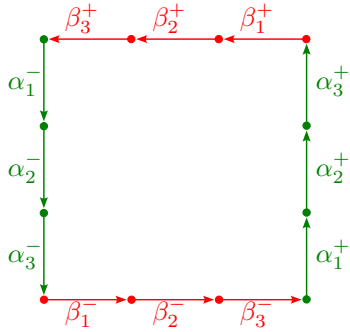
  \caption{Newton polygon for a $3\times 3$ fundamental domain of the square lattice}
  \label{Fig Newton sq}
\end{figure}

The monodromies are
\begin{equation}\label{monodromies sq}
z(P)=\frac{\psi_{(1,0)}(P)}{\psi_{(0,0)}(P)}=\prod_{k=1}^d \frac{E(P,\alpha_k^-)}{E(P,\alpha_k^+)}, \quad w(P)=\frac{\psi_{(0,1)}(P)}{\psi_{(0,0)}(P)}=\prod_{k=1}^d \frac{E(P,\beta_k^-)}{E(P,\beta_k^+)}.
\end{equation}
The periodicity condition of the Dirac operator is
\begin{equation}
\label{ch07:eq_periodicDiracQuadGrid}
\sum_{k=1}^d (A(\alpha_k^-)-A(\alpha_k^+)) = 0\; {\rm in}\; J(\mathcal{R}), \quad \sum_{k=1}^d (A(\beta_k^-)-A(\beta_k^+)) = 0\; {\rm in}\; J(\mathcal{R}),
\end{equation}
so that both divisors $\sum_{k=1}^d (\alpha_k^--\alpha_k^+)$ and $\sum_{k=1}^d (\beta_k^--\beta_k^+)$ are principal. Thus, the problem consists in finding two pairs of equivalent divisors $\sum_{k=1}^d \alpha_k^\pm$ and $\sum_{k=1}^d\beta_k^\pm$ of points on the oval $X_0$ of the $M$-curve $\mathcal{R}$ satisfying clustering condition \eqref{clusters sq}.

\subsection{Existence for Harnack curves over a given $M$-curve}

In this Section we will show that on any M-curve there exists a collection of equivalent clustered divisors on $X_0$, as required by Theorem~\ref{th alphas vs Newton polygon} and Proposition~\ref{prop periodic Dirac}. These are 
 three equivalent clustered divisors for the hexagon lattice, or two pairs of equivalent clustered divisors for the square lattice. As explained above , the corresponding functions \eqref{monodromies} define a Harnack curve $\mathcal{P}(z,w)=0$.
 We formulate the following result for the square lattice for the sake of simplicity only, it is actually of a quite general nature.
 
 We first formulate the following Lemma, which we will use also in Section \ref{sec:9_heightFunction}.
 
 \begin{lemma}
 \label{lemma clusters}
Let $\ell^\pm$  be two arcs in $X_0$. Then there exists $L\in \mathbb{N}$ such that for any $R \in T^g = \mathbb{R}^g \Big/ \mathbb{Z}^g = \myRe\br{\Jac(\mathcal{R})}$ there exists a divisor $\sum_j^L (s_j^- - s_j^+)$  with $s_j^\pm \in \ell^\pm$ and
\[ \sum_{j=1}^L A(s_j^- -s_j^+) = R. \]
\end{lemma}
\begin{proof}  Any arc supports a non-special divisor, since they are generic. The set
$$
S_g(\ell^\pm)=\{ \sum_{i=1}^g A(t_i^- - t_i^+)\in\mathbb R^g | t_i^\pm\in \ell^\pm \}
$$
is an open set in $\mathbb R^g$. There follows that the set $N S_g(\ell^\pm)\subset\mathbb R^g$ (the homothetic image of $S_g(\ell^\pm)$ scaled by $N$) for a sufficiently large $N$ contains the torus $T^g$. This completes the proof for $L=Ng$.
\end{proof}

\begin{theorem}
\label{th existence}
Let $\mathcal{R}$ be an $M$-curve.
There exists $N_0\in \mathbb N$ such that for any $m,n\ge N_0$ there exist principal divisors $\sum_{i=1}^m (\alpha_i^- - \alpha_i^+), \sum_{j=1}^n (\beta_j^- - \beta_j^+)$ satisfying clustering conditions \eqref{clusters sq}. Functions \eqref{monodromies sq} define a Harnack curve with the Riemann surface $\mathcal{R}$. 
\end{theorem}
\begin{proof}
Choose four disjoint arcs satisfying clustering condition, see Fig.~\ref{fig:clustering_for_regularization} and apply Lemma~\ref{lemma clusters} with $R=0$.
 \end{proof}

\begin{figure}[h]
    \centering
    \fontsize{10pt}{12pt}\selectfont
    \def\svgwidth{0.3\textwidth}
    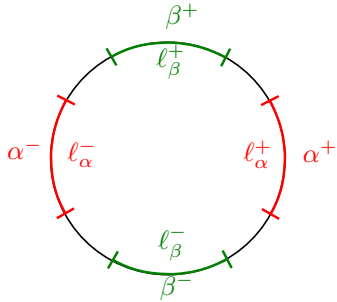
  \caption{Distjoint clustering arcs for the square grid.}
  \label{fig:clustering_for_regularization}
\end{figure}

\section{The Ronkin function}
\label{sec:8_Ronkin_function}

In this and the following section for simplicity we consider the case of the square grid only. General considerations are similar.

Let $\mathcal{R}$ be an $M$-curve, $\tau$ its anti-holomorphic involution and $X = \bigcup_{i=0}^g X_i$ the set of real ovals. $X$ decomposes $\mathcal{R}$ into two components $\mathcal{R}_+$ and $\mathcal{R}_- = \tau(\mathcal{R}_+)$. We have $\mathcal{R_+} = \mathcal{R}/\tau$ and $\mathcal{R}_+ \cap \mathcal{R}_- = X$. We also denote the corresponding open components by $\mathcal{R}^\circ_\pm := \mathcal{R}_\pm \setminus X$.

Let $X_0$ be a distinguished real oval and $\cubr{\alpha, \beta } := \cubr{\alpha_1^\pm, \ldots, \alpha_m^\pm, \beta_1^\pm, \ldots, \beta_n^\pm} \subset X_0$ be points on it clustered as in \eqref{clusters sq}. They separate $X_0$ into $2m+2n$ arcs $\delta_i$.

We also choose a canonical basis of cycles as in Section~\ref{sec:6_M-curves} , i.e. $a_i = X_i$, $i = 1, \ldots, g$, and $\tau a_i = a_i$, $\tau b_i = -b_i$.

\subsection{The Ronkin function on Harnack data}

Following Krichever \cite{krichever_amoebas_2014}
\[ \mathcal{S} = \cubr{\mathcal{R}, \cubr{\alpha, \beta}}\]
is called \emph{Harnack data} for the $M$-curve $\mathcal{R}$. 

Denote by $d\zeta^{\alpha_i}, d\zeta^{\beta_i}$ the Abelian differential of the third kind with simple poles at $\alpha_i^\pm$ or respectively $\beta_i^\pm$ and 
\begin{equation}
	\label{eq:residues_of_dZeta_i}
	\res_{\alpha_i^-} d\zeta^{\alpha_i} = - \res_{\alpha_i^+} d\zeta^{\alpha_i} = 1 = \res_{\beta_i^-} d\zeta^{\beta_i} = - \res_{\beta_i^+} d\zeta^{\beta_i},
\end{equation}
and imaginary periods. Introduce two differentials 
\begin{equation}
	\label{eq:definition_of_dXis}
	d\zeta_1 = \sum_{i = 1}^m d\zeta^{\alpha_i}, \quad d\zeta_2 = \sum_{i = 1}^n d\zeta^{\beta_i}.
\end{equation}
They are real on $X$, moreover $d\zeta_k(\tau P) = \overline{d\zeta_k(P)}$. 
The corresponding integrals 
\[ \zeta_k(P) = \int_Q^P d\zeta_k \] 
with $Q \in X_0$ have jumps by $2\pi i$ on some paths $[\alpha_i^-, \alpha_i^+]$ and $[\beta_j^-, \beta_j^+]$. We chose these paths to lie in $\mathcal{R}_-$. Since the periods may be only imaginary, all periods of $\zeta_k$ over the real ovals vanish, and we obtain functions $\zeta_k$ well defined on $\mathcal{R}_+ \setminus \cubr{\alpha, \beta}$. 

Denote the real and imaginary parts of $\zeta_k$ by
\[ \zeta_k(P) = x_k + i y_k, \quad i = 1,2. \]
\begin{definition}
\label{def:amoeba_map_polygon_map}
	The map
	\[ \mathcal{A}(P) := (x_1, x_2) = (\myRe \zeta_1(P), \myRe \zeta_2(P)) \]
	is called the \emph{amoeba map}, and its image 
	\[ \mathcal{A}(\mathcal{R}_+ \setminus \cubr{\alpha, \beta}) =: \mathcal{A}_\mathcal{S} \subset \mathbb{R}^2 \]
	is called the \emph{amoeba}.
	
	Similarly the map
	\[ \Delta(P) = (s_1, s_2) = J(y_1, y_2) := \frac{1}{\pi}(-y_2, y_1) = \frac{1}{\pi}(- \myIm \zeta_2(P), \myIm \zeta_1(P)) \]
	we call the \emph{polygon map}, and its image
	\[ \Delta(\mathcal{R}_+ \setminus \cubr{\alpha, \beta}) =: \Delta_\mathcal{S} \subset \mathbb{R}^2 \]
	is called the \emph{Newton polygon}.
\end{definition} 

\begin{rem}
	The scaling and rotation $J$ in the definition of $\Delta$ is introduced so that it matches the classical Newton polygon. It will often be more convenient to work directly with $(y_1, y_2)$ as it produces more symmetric formulas.
\end{rem}
	
Under the amoeba map $X_0$ is mapped  to the concave boundary of $\mathcal{A}_\mathcal{S}$, where the points $\cubr{\alpha, \beta}$ are mapped to $2m + 2n$ infinite tentactles of the amoeba. The ovals $X_1, \ldots, X_g$ are mapped to $g$ convex ovals bounding interior islands in the amoeba. Under the polygon map the (interior) ovals $X_i$ are mapped to the interior points of $\Delta_\mathcal{S}$ given by the b-periods
\begin{equation}
\label{eq:b-periods}
\int_{b_i} d\zeta_k= {i y_k}_{| X_i}.
\end{equation}
The singular points $\cubr{\alpha, \beta}$ separate $X_0$ into $2m + 2n$ connected components $\delta_1, \ldots, \delta_{2m + 2n}$. Each $\delta_l$ is mapped under the amoeba map to an external boundary component of $\mathcal{A}_\mathcal{S}$, and to an exterior vertex of $\Delta_S$ under the polygon map. See Fig.~\ref{fig:AmoebaAndNewtonPolygon} for an illustration.

\begin{figure}[h]
\centering
\centering
\fontsize{10pt}{12pt}\selectfont
\def\svgwidth{\textwidth}
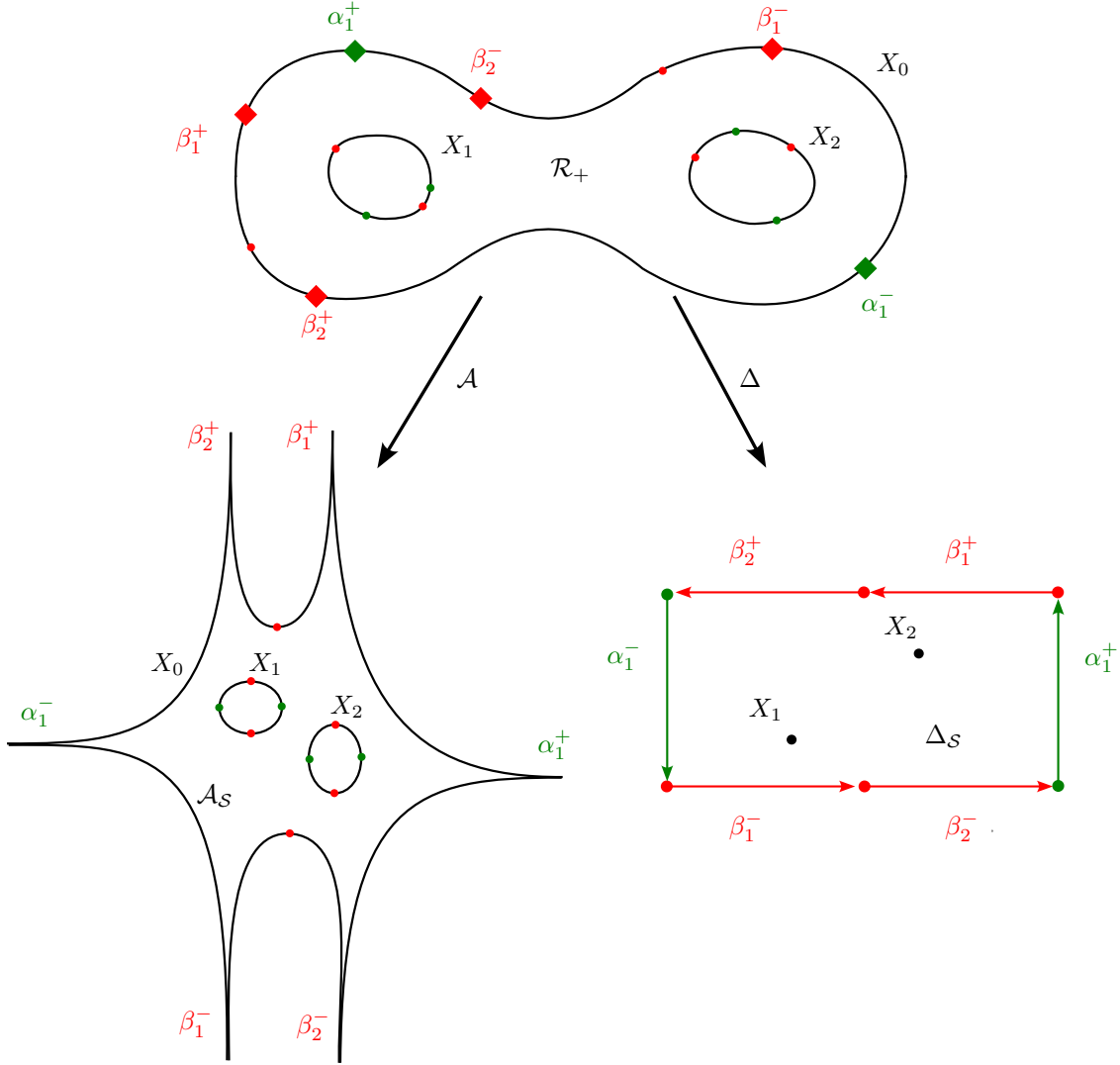
\caption{The two diffeomorphisms $\mathcal{A}$ and $\Delta$ coming from the M-curve $\mathcal{R}$. Amoeba $\mathcal{A}_\mathcal{S}$ (left) and Newton polygon $\Delta_\mathcal{S}$ (right). Zeros of $d\zeta_1$, $d\zeta_2$ as green and red circles respectively. Note that the $X_i$ do not need to be mapped to integer points in $\Delta_\mathcal{S}$ as opposed to the algebraic case.}
\label{fig:AmoebaAndNewtonPolygon}
\end{figure}

The differential $d\zeta_1$ has $2g-2 +2m$ zeros, and $d\zeta_2$ has $2g-2 +2n$ zeros. We observe that we know locations of all these zeros. 

\begin{lemma}
\label{lem:zeros_dzeta}
	The differentials $d\zeta_1, d\zeta_2$ have two zeros on each real oval $X_i, i=1,\ldots, g$. Each of the components $(\alpha_i^+, \alpha_{i+1}^+), (\alpha_i^-, \alpha_{i+1}^-)$ of $X_0$ contains exactly one of the remaining $2m-2$ zeros of $d\zeta_1$. Each of the components $(\beta_j^+, \beta_{j+1}^+), (\beta_j^-, \beta_{j+1}^-)$ of $X_0$ contains exactly one of the remaining $2n-2$ zeros of $d\zeta_2$.
	Moreover, $2g+(2m-2)+(2n-2)$ zeros of the differential $d\zeta_{\alpha,\beta}:=\alpha d\zeta_1+\beta d\zeta_2$ for any $\alpha,\beta\in \mathbb{R}_*$ are also located in $X$ as described above; and the remaining two zeros lie in two of the four arcs $(\alpha^\pm, \beta^\pm)$. 
\end{lemma}
\begin{proof}
	$d\zeta_k$ must have two zeros on each $X_i, i=1, \ldots, g$ since they are real and the periods over $X_i$ vanish. The facts about the zeros on $X_0$ follow from the continuity of $d\zeta_k$ on the corresponding component $\delta_i$ and the fact that it has poles of the same sign at both ends of the component. 
	The arguments for $d\zeta_{\alpha,\beta}$ are the same.
\end{proof}
See Fig.~\ref{fig:AmoebaAndNewtonPolygon} for an illustration of the zeros.

$\zeta_{\alpha,\beta}$ are periodic functions on $X_i$ with exactly two critical points. Also on each $\delta_k\in X_0$ the integral $\zeta_{\alpha,\beta}$ has at most one critical point. We conclude that the lines $\alpha x_1+\beta x_2=\gamma$ intersect the components of $X$ at most in two points.
\begin{corollary}
\label{cor:convex}
	The images $\mathcal{A}(X_i)\subset \mathbb{R}^2, i=1,\ldots,g$ are convex ovals. The images $\mathcal{A}(\delta_k)\subset\mathbb{R}^2, k=1,\ldots, 2m+2n $ are infinite concave arcs. 
\end{corollary}

\begin{corollary}
	$R := \frac{d\zeta_1}{d\zeta_2}$ is a holomorphic function without zeros on $\mathcal{R}_+^\circ$.
\end{corollary}

The identity $d\zeta_1 = R d\zeta_2$ is equivalent to
	\[ 1 + i \frac{\partial y_1}{\partial x_1} = i R \frac{\partial y_2}{\partial x_1}, \quad i \frac{\partial y_1}{\partial x_2} = R \br{1 + i \frac{\partial y_2}{\partial x_2}}. \]
	Equalizing the real and imaginary parts in these identities one easily arrives at the following formulas for the derivatives:
\begin{equation}
\label{eq:partial_derivatives_R}
 \frac{\partial y_2}{\partial x_1} = -\frac{1}{\myIm R}, \quad \frac{\partial y_2}{\partial x_2} = - \frac{\partial y_1}{\partial x_1} = \frac{\myRe R}{\myIm R}, \quad \frac{\partial y_1}{\partial x_2} = \frac{|R|^2}{\myIm R}. 
 \end{equation}

We denote the open sets
$$
\mathcal{A}^\circ_\mathcal{S}:=\mathcal{A}(\mathcal{R}^\circ_+), \quad \Delta^\circ_\mathcal{S}:=\Delta(\mathcal{R}^\circ_+).
$$
$\Delta_\mathcal{S}^\circ$ is the interior of a polygon with $g$ points removed. The latter are the images under $\Delta$ of the real ovals $X_1,\ldots,X_g$.
The following result is due to Krichever \cite{krichever_amoebas_2014}.  We present its proof for completeness.

\begin{lemma}
	The amoeba map $\mathcal{A}: \mathcal{R}^\circ_+ \to \mathcal{A}^\circ_\mathcal{S}$ and the polygon map $\Delta: \mathcal{R}^\circ_+ \to \Delta^\circ_\mathcal{S}$ are diffeomorphisms.  
\end{lemma}
%

\begin{proof}
First we show that the set of critical points of the amoeba map coincides with the set of real ovals $X$. At regular points $(x_1,x_2)$ are local coordinates, which yields $dx_1\wedge dx_2\neq 0$. The representation 
$d\zeta_2= i \frac{\partial y_2}{\partial x_1}dx_1+(1+ i \frac{\partial y_2}{\partial x_2})dx_2$ implies 
$d\zeta_2\wedge d\bar{\zeta_2}=2i \frac{\partial y_2}{\partial x_1}dx_1\wedge dx_2$. Using (\ref{eq:partial_derivatives_R}) for the partial derivative we obtain
$$ 
2dx_1\wedge dx_2=i \myIm R\ d\zeta_2\wedge d\bar{\zeta_2}.
$$ 
We know that all zeros of $d\zeta_2$ are in $X$, therefore a point $P_0\in \mathcal{R}_+^\circ$ is critical if and only if $R(P_0)\in \mathbb{R}$. Obviously $d\zeta :=d\zeta_1-R(P_0) d\zeta_2$ vanishes at $P_0$. On the other hand, all $2g+2$ zeros of $d\zeta$ must lie in $X$. Indeed, each oval $X_i, i=1,\ldots, g$ contains 2 zeros because of the vanishing periods, and $X_0$ also contains 2 zeros between the poles with the residues of the same sign. Thus there are no critical points in $\mathcal{R}_+^\circ$.

The map $\mathcal{A}$ is injective on $X_0$ as well as on each $X_i, i=1,\ldots, g$ because of the monotonicity of $d\zeta_i$ on the corresponding components of $X$, see Lemma~\ref{lem:zeros_dzeta}. Due to Corollary~\ref{cor:convex} the set $\{(x_1,x_2)\in \mathcal{A}_{\mathcal{S}}^\circ | x_1=a \}$ consists of disjoint intervals. Since $\mathcal{A}$ is a local diffeomorphism on $\mathcal{R}_+^\circ$, the derivative of $x_2$ on the segments of the level line $x_1^{-1}(a)\subset \mathcal{R}_+^\circ$ cannot vanish, thus the function $x_2$ is monotonic. This proves the injectivity of $\mathcal{A}$. 

For the polygon map the proof is the same.
\end{proof}

We see that both $(x_1, x_2) \in \mathcal{A}_\mathcal{S}$ and $(y_1, y_2) \in J^{-1}\Delta_\mathcal{S}$ as well as $(s_1, s_2) \in \Delta_\mathcal{S}$ can be treated as global variables on $\mathcal{R}^\circ_+$, and one obtains well defined functions
\[ y_1(x_1, x_2), y_2(x_1, x_2); \quad x_1(y_1, y_2), x_2(y_1, y_2). \]
Note that $\zeta_k$ and therefore $x_k$ are defined up to real constants that correspond to choice of $Q$. This is just a global shift of the amoeba map.

Let us introduce the function
\[ h(P) :=\frac{1}{\pi }\myIm \int_\ell \zeta_2 d\zeta_1 = \frac{1}{2\pi i} \int_{\mathcal L} \zeta_2 d\zeta_1
\]
on $\mathcal{R}_+ \setminus \cubr{\alpha, \beta}$. Here $\ell\subset\mathcal{R}_+$ as a path from the boundary $Q\in X_0$ to $P$.
The path $\mathcal{L}=-\tau\ell + \ell $ from $\tau P \in \mathcal{R}_-$ to $P \in \mathcal{R}_+$ has the symmetry $\tau \mathcal{L}= -\mathcal{L}$. 
This is a real valued function. It vanishes on the external boundary component $\delta_\mathcal{L}\ni Q$ of $\mathcal{A}_\mathcal{S}$ which is crossed by $\mathcal{L}$. Functions $h$ that correspond to different choices of $\delta_\mathcal{L}$ differ by constants.

Considering $(x_1, x_2)$ as variables on $\mathcal{R}_+ \setminus \cubr{\alpha, \beta}$ we obtain
\[ d\zeta_1 = \br{ 1 + i \frac{\partial y_1}{\partial x_1}} dx_1 + i \frac{\partial y_1}{\partial x_2} dx_2. \]
This implies the following formulas for the derivatives of $h$:

\begin{align*}
	\frac{\partial h}{\partial x_1} &= \frac{1}{\pi} \myIm \br{\br{x_2 + i y_2} \br{1 + i \frac{\partial y_1}{\partial x_1}}} = \frac{1}{\pi} \br{x_2 \frac{\partial y_1}{\partial x_1} + y_2}, \\
	\frac{\partial h}{\partial x_2} &= \frac{1}{\pi} x_2 \frac{\partial y_1}{\partial x_2}.
\end{align*}

Similarly in the coordinates $(y_1, y_2)$ one obtains
\[d\zeta_1 = \br{ i +  \frac{\partial x_1}{\partial y_1}} dy_1 +  \frac{\partial x_1}{\partial y_2} dy_2 .\]
and thus

\begin{align*}
	\frac{\partial h}{\partial y_1} &= \frac{1}{\pi} \br{y_2 \frac{\partial x_1}{\partial y_1} + x_2}, \\
	\frac{\partial h}{\partial y_2} &= \frac{1}{\pi} y_2 \frac{\partial x_1}{\partial y_2}.
\end{align*}

Introducing 
\begin{equation}
	\label{eq:def_rho}
	\rho(x_1, x_2) := -h(P) + \frac{1}{\pi} x_2 y_1
\end{equation}
and its Legendre dual
\begin{equation}
	\label{eq:def_rhoTilde}
	\sigma(y_1, y_2) := h(P) - \frac{1}{\pi} x_1 y_2 = -\rho(x_1, x_2) + \frac{1}{\pi}\br{x_2 y_1 - x_1 y_2},
\end{equation}
we get

\begin{equation}
	\label{eq:partial_derivatives_of_rhos}
	\begin{alignedat}{2}
		\frac{\partial \rho}{\partial x_1} &= -\frac{1}{\pi} y_2, \quad &&\frac{\partial \rho}{\partial x_2} = \frac{1}{\pi} y_1, \\
		\frac{\partial \sigma}{\partial y_1} &= \frac{1}{\pi} x_2, \quad &&\frac{\partial\sigma}{\partial y_2} = -\frac{1}{\pi} x_1,
	\end{alignedat}
\end{equation}

More symmetric representations for $\rho$ and $\sigma$ are given by

\begin{align}
	\rho(x_1, x_2) &= H(P) + \frac{1}{2 \pi} \br{x_2 y_1 - x_1 y_2} \\
	\label{eq:rho_sigma_Legendre_duality}
	\sigma(y_1, y_2) &= -H(P) + \frac{1}{2 \pi} \br{x_2 y_1 - x_1 y_2} = -\rho(x_1, x_2) + \frac{1}{\pi} \br{x_2 y_1 - x_1 y_2}
\end{align}
with 
\[H(P) = \frac{1}{4\pi i} \int_\mathcal{L} \zeta_1 d\zeta_2 - \zeta_2 d\zeta_1=\frac{1}{2\pi}\myIm \int_\ell \zeta_1 d\zeta_2 - \zeta_2 d\zeta_1
.\]
Note that in $(s_1, s_2)$ coordinates this recovers the standard Legendre duality
\begin{equation}
	\label{eq:rho_sigma_Legendre_duality_s_coords}
	\sigma(s_1, s_2) = -\rho(x_1, x_2) + x_1 s_1 + x_2 s_2.
\end{equation}

\begin{theorem}
	Both $\rho(x_1, x_2)$ and $\sigma(y_1, y_2)$ are strictly convex functions on $\mathcal{A}_\mathcal{S}^\circ$ and $J^{-1}\Delta_\mathcal{S}^\circ$ respectively. Their Hessians are equal to 
	\[ \Hess (\rho) = \Hess (\sigma) = \frac{1}{\pi \myIm R} \begin{pmatrix}
		1 & -\myRe R \\
		- \myRe R & |R|^2
	\end{pmatrix} .\]
\end{theorem}

\begin{proof}
Formulas for the coefficients of 
	\[ \Hess(\rho) = \begin{pmatrix}
		\frac{\partial^2 \rho}{\partial x_1^2} & \frac{\partial^2 \rho}{\partial x_1 \partial x_2} \\
		\frac{\partial^2 \rho}{\partial x_1 \partial x_2} & \frac{\partial^2 \rho}{\partial x_2^2}
	\end{pmatrix}. \]
follow directly from \eqref{eq:partial_derivatives_of_rhos} and \eqref{eq:partial_derivatives_R}. 
		
We omit the computation for the Legendre dual, which is the same.
\end{proof}

\begin{proposition}
	The function $\rho: \mathcal{A}_\mathcal{S} \to \mathbb{R}$ can be continuously differentiably extended to the whole $\mathbb{R}^2$ plane (i.e. to the exterior of the amoeba and its islands) by affine functions. It vanishes on the exterior component of $\mathbb{R}^2 \setminus \mathcal{A}_\mathcal{S}$ bounded by $\mathcal{A}(\delta_\mathcal{L})$, which is crossed by $\mathcal{L}$.
\end{proposition}
\begin{proof}
Since both differentials $d\zeta_k$ are real on $X$ the imaginary components $y_1, y_2$ are constant on every arc $\delta$ of $X_0$ built by the points $\alpha, \beta$. Wenn we pass one of these points the corresponding value of $y_1$ or $y_2$ changes by $\pi$. On the oval $X_1,\ldots,X_g$ the imaginary components $y_1, y_2$ are also constant since they are the b-periods (\ref{eq:b-periods}).
Thus the gradient of $\rho$ on each component of the boundary $\mathcal{A}(X)$ of amoeba is constant (\ref{eq:partial_derivatives_of_rhos}), which completes the proof.
\end{proof}

\begin{definition}
	The function $\rho: \mathbb{R}^2 \to \mathbb{R}$ is called the \emph{Ronkin function} of the Harnack data $\mathcal{S}$.
\end{definition}

The notion of a Ronkin function of Harnack data was originally introduced by Krichiver in \cite{krichever_amoebas_2014} as an integral of a two-form over the lower left quadrant in $\mathbb{R}^2$ with the corner at $(x_1, x_2)$. Our representation via an integral of a holomorphic one form is more convenient for investigation and computation (see Section~\ref{sec:12_computation}).

\subsection{Doubly periodic weights and identification with free energy} 

Consider the case when the weights are doubly periodic. The graph $G_1$ on a torus contains $m$ pairs of (horizontal) $\alpha_i^\pm$ train tracks and $n$ pairs of (vertical) $\beta_j^\pm$ train tracks. The periodicity condition~\eqref{ch07:eq_periodicDiracQuadGrid} reads as
\[ \sum_{i=1}^m (A(\alpha_i^-)-A(\alpha_i^+)) = 0, \quad \sum_{j=1}^n (A(\beta_j^-)-A(\beta_j^+)) = 0 \] 
The differentials $d\zeta_1, d\zeta_2$ are given by
\[ d\zeta_1 = \frac{dz}{z}, \quad d\zeta_2 = \frac{dw}{w}, \]
where $z$ and $w$ are meromorphic functions \eqref{monodromies sq} on $\mathcal{R}$.

The Kasteleyn operator $K(z,w)$ with monodromies $z,w$ determines the spectral curve $ \mathcal{P}(z,w) = 0$, where 
\[ \mathcal{P}(z,w) = \det K(z,w). \]
Its Riemann surface is $\mathcal{R}$.

Consider the graph $G_N$ on a torus which is the $N \times N$ covering of $G_1$ (see Fig.~\ref{fig:regularizationGraph}) with the same weights $\nu$ as of $G_1$. The thermodynamic limit of the free energy is defined as 
\[ \log Z(\nu) := \lim_{n \to \infty} \frac{1}{N^2} \log Z(G_N, \nu). \]
This limit exists and depends on the weights $\nu$ on $G_1$. In our case we denote it
\[ E(\mathcal{R}, \cubr{\alpha, \beta}) = \log Z(\nu) \]
since it is a function of the Harnack data. Note that it is independent of the free vector $Z \in J(\mathcal{R})$ which enters the formulas \eqref{Fock coeff} for the weights.

In \cite{Kenyon_Okounkov_Sheffield_2006} it was shown that for doubly-periodic weights $E = F(0,0)$, where
\[ F(x_1, x_2) := \frac{1}{(2\pi i)^2} \int_{\mathbb{T}^2} \log |\mathcal{P}(e^{x_1} z, e^{x_2} w)| \frac{dz}{z} \frac{dw}{w}. \]
Here the integration is taken over the torus
\[ \mathbb{T}^2 = \cubr{(z,w) \in \mathbb{C}^2 | \; |z| = |w| = 1}. \]
The curve $\mathcal{P}(z,w) = 0$ is a Harnack curve. Its Riemann surface is an M-curve $\mathcal{R}$.

The classical amoeba map $\mathcal{R}_+^\circ \to \mathbb{R}^2$ is given by $(\log |z|, \log |w|)$. Comparing to the one from Definition~\ref{def:amoeba_map_polygon_map} it is shifted by the values of $z$ and $w$ at the starting integration point $Q$ of $\zeta_k$:
$$
\log |z(P)|-\log |z(Q)|=\zeta_1(P), \quad \log |w(P)|-\log |w(Q)|=\zeta_2(P).
$$
Its image $\mathcal{A}$ is $\mathcal{A}_\mathcal{S}$ shifted by the same constant.

Generally $F(x_1, x_2)$ is interpreted as the free energy of the dimer model on a torus with magnetic field. The behaviour of the model qualitatively depends on the value of the magnetic field:
\begin{enumerate}
	\item For $(x_1, x_2) \in \mathcal{A}_\mathcal{S}$ the model is in the \emph{liquid} phase.
	\item For $(x_1, x_2)$ in the infinite exterior domain of the amoeba the model is in the \emph{solid} phase.
	\item For $(x_1, x_2)$ in the amoeba islands the model is in the \emph{gas} phase.
\end{enumerate}

The function $F(x_1, x_2)$ is the classical \emph{Ronkin function of a Harnack curve}. It is a continuously differentiable convex function on $\mathbb{R}^2$, which is piecewise affine on $\mathbb{R}^2 \setminus \mathcal{A}$ and smooth and strictly convex on $\mathcal{A}$.

\begin{theorem}
\label{thm:Our_and_periodic_ronkin_agree}
	$F(x_1, x_2) = \rho(x_1-c_1, x_2-c_2) + \log|\mathcal{P}(0,0)|, \quad c_1=\log |z(Q)|, c_2=\log |w(Q)|.$
\end{theorem}
\begin{proof}
	As it was shown in \cite{passare_amoebas_2004}, the Hessian of $F(x_1, x_2)$ computed at the point $P$, i.e. at $(x_1,x_2)=(\log |z(P)|, \log |w(P)|)$ is given by $\Hess(\rho(P))$, see also \cite{krichever_amoebas_2014}. Thus the difference $F - \rho$ is an affine function $a_1 x_1 + a_2 x_2 + a_3$. With our choice the integration contour $\mathcal{L}$ intersecting the component $(\alpha_m^-, \beta_1^-)$, $\rho$ identically vanishes on the infinite domain of $\mathbb{R}^2 \setminus \mathcal{A}_\mathcal{S}$ bounded by $\mathcal{A}((\alpha_m^-, \beta_1^-))$. In the limit we have
	\[ \lim_{\begin{smallmatrix} x_1 \to -\infty & \\ x_2 \to -\infty \end{smallmatrix}} F(x_1, x_2) = \log |\mathcal{P}(0,0)|. \]
Thus $a_1=a_2=0$.
\end{proof}
\begin{rem}
	We note that $\mathcal{P}(0,0)$ is the partition function of the finite planar graph we get if for fixed $i,j$ we remove the edges crossing the train tracks $\alpha_i^+$ or $\beta_j^+$ from the graph $G_1$ on a torus. This partition function is independent of the choice of $i,j$ but does depend on the cut more generally. In this sense $\rho$ is a correction term from the finite planar case to the thermodynamic limit on a torus.
\end{rem}

From \eqref{eq:definition_of_dXis} we obtain explicit formulas for $\rho$ and $\sigma$ in the form of a sum over the edges of $G_1$. Note the similarity of these functions when represented in terms of the Riemann surface.
\begin{align}
\label{eq:rho_explicit}
	&\rho(x_1, x_2) = \sum_{i=1}^m \sum_{j=1}^n \br{-\frac{1}{ \pi } \myIm \int_\ell \zeta^{\beta_j} d\zeta^{\alpha_i} + \frac{1}{\pi} \myRe \zeta^{\beta_j}(P) \myIm \zeta^{\alpha_i}(P) } \\
	&= \sum_{i=1}^m \sum_{j=1}^n \br{- \frac{1}{2 \pi } \myIm\int_\ell \zeta^{\beta_j} d\zeta^{\alpha_i} - \zeta^{\alpha_i} d\zeta^{\beta_j} + \frac{1}{2\pi} \br{\myRe \zeta^{\beta_j}(P) \myIm \zeta^{\alpha_i} - \myRe \zeta^{\alpha_i}(P) \myIm \zeta^{\beta_j}(P)} } \nonumber
\end{align}
\begin{align}
\label{eq:sigma_explicit}
	&\sigma(s_1, s_2) = \sum_{i=1}^m \sum_{j=1}^n \br{\frac{1}{ \pi } \myIm \int_\ell \zeta^{\beta_j} d\zeta^{\alpha_i} - \frac{1}{\pi} \myRe \zeta^{\alpha_i}(P) \myIm \zeta^{\beta_j}(P) }\\
	&= \sum_{i=1}^m \sum_{j=1}^n \br{ \frac{1}{2 \pi } \myIm\int_\ell \zeta^{\beta_j} d\zeta^{\alpha_i} - \zeta^{\alpha_i} d\zeta^{\beta_j} + \frac{1}{2\pi} \br{\myRe \zeta^{\beta_j}(P) \myIm \zeta^{\alpha_i} - \myRe \zeta^{\alpha_i}(P) \myIm \zeta^{\beta_j}(P)} }. \nonumber
\end{align}



\section{Height Function}
\label{sec:9_heightFunction}



Let $G$ be an infinite $\mathbb{Z}^2$-invariant bipartite planar graph with doubly periodic dimer weights and a finite fundamental domain $G_1$ which is embedded in the unit square and let $G_N$ be its $N \times N$ covering. A Gibbs measure on $G$ is a probability measure $\mu(G)$ on infinite dimer configurations with the property that if we fix a matching on a topological annulus the matching inside it is independent of the matching outside. Furthermore the measure restricted to the inside of the annulus is the Boltzmann measure defined in \eqref{eq:def_Boltzmann_measure}.

It was shown in \cite{sheffield_random_2006} that for any $(s_1, s_2) \in \Delta_\mathcal{S}$ there exists a unique ergodic Gibbs measure $\mu_{(s_1, s_2)}$ and that the measures $\mu(G_N)$ conditioned on having height change 
\begin{equation}
	\label{eq:definition_height_change}
	(h(f + (0,N)) - h(f), h(f + (-N,0)) - h(f)) = \br{\lfloor N s_1 \rfloor, \lfloor N s_2 \rfloor }
\end{equation}
converge to $\mu_{(s_1, s_2)}$ as $N$ goes to infinity.

Let now $\mathcal{M}_{(s_1, s_2)}(G_N)$ be the set of all dimer configurations with this given height change on $G_N$ and
\[ Z_{(s_1, s_2)}(G_N) = \sum_{M \in \mathcal{M}_{(s_1, s_2)}(G_N)} \prod_{e \in M} \nu(e) \]
be the partition function of $\mu_{(s_1, s_2)}(G_N)$. The surface tension is then defined as a function of $(s_1, s_2)$ through the free energy per fundamental domain of $\mu_{(s_1, s_2)}$:
\[ -\varsigma(s_1, s_2) = \log Z_{(s_1, s_2)} = \lim_{N \to \infty} \frac{1}{N^2} \log Z_{(s_1, s_2)}(G_N).\]

It was shown in \cite{Kenyon_Okounkov_Sheffield_2006} that this surface tension is the Legendre dual of the Ronkin function $F$. That is

\[ F(x_1, x_2) = \max_{(s_1, s_2)} \br{-\varsigma(s_1, s_2) + x_1 s_1 + x_2 s_2}. \]

\begin{rem}
	For clarity we note here that our setup differs from that of \cite{Kenyon_Okounkov_Sheffield_2006} by a rotation of all coordinates by $\pi$. In their work the monodromies $z,w$ go in the directions $(0,-1), (1,0)$ and the height function change is considered in the directions $(1,0), (0,1)$.
\end{rem}

\subsection{Regularized convergence of the surface tension}

We have seen in Theorem~\ref{thm:Our_and_periodic_ronkin_agree} that $F = \rho$ up to a shift of the argument and an additive constant in the case of doubly periodic weights. Since $\varsigma, \sigma$ are Legendre duals of $F, \rho$ respectively we obtain that they coincide up to an affine function
\begin{equation}
\label{eq:affine_sigma}
\varsigma(s_1,s_2)=\sigma(s_1,s_2) +b_1 s_1+ b_2 s_2 + b_3. 
\end{equation}
This affine function dissapears in the variational description of the height function (see Remark~\ref{re:affine_sigma}), and one obtains an equivalence class $[\sigma]$ of surface tension functions that differ by affine functions. We consider particular representatives in the equivalence class to formulate the regularized convergence results for $\sigma$.

Let $\zeta_1, \zeta_2$ in general (quasi-periodic) case be given by
\[ \zeta_1 = \sum_{i=1}^m \zeta^{\alpha_i}, \; \zeta_2 = \sum_{j=1}^n \zeta^{\beta_j}. \]
Let $N$ be the parameter of the thermodynamic limit and $L$ fixed (See Lemma~\ref{lemma clusters}) with $\br{\gamma^\pm_1, \ldots, \gamma^\pm_L} \in \ell_\alpha^\pm, \; \br{\delta^\pm_1, \ldots, \delta^\pm_L} \in \ell_\beta^\pm$ clustered points (See Fig.~\ref{fig:clustering_for_regularization}) satisfying
\begin{align*}
	N \sum_{i=1}^m \br{\alpha_i^- - \alpha_i^+} + \sum_{l=1}^L \br{\gamma_l^- - \gamma_l^+} &= 0\\
	N \sum_{j=1}^m \br{\beta_i^- - \beta_i^+} + \sum_{l=1}^L \br{\delta_l^- - \delta_l^+} &= 0
\end{align*}

Now let us consider the regularized graph $\hat{G}_N$ with the Fock weights defined by the train tracks $\alpha, \beta, \gamma, \delta$ (See Fig.~\ref{fig:regularizationGraph}). Note that the weights on various copies of $G_1$ in $\hat{G}_N$ generically do not coincide (they do under the periodicity condition \eqref{ch07:eq_periodicDiracQuadGrid}).

The integrals
\begin{align*}
	\zeta_1^N &= N\sum_{i=1}^m \zeta^{\alpha_i} + \sum_{l=1}^L \zeta^{\gamma_l}, \\
	\zeta_2^N &= N\sum_{j=1}^n \zeta^{\beta_j} + \sum_{l=1}^L \zeta^{\delta_l}. 
\end{align*}
correspond to periodic weights on $\hat{G}_N$. Denote their real and imaginary parts by
\[ \zeta_k^N(P) = x_k^N + i y_k^N, \; i=1,2 \]
and again introduce $s_1^N, s_2^N := J(y_1^N, y_2^N)$. The surface tension on the periodic graph $\hat{G}_N$ is given by
\[ \sigma^N(s_1^N, s_2^N) = \frac{1}{\pi} \myIm \int_\ell \zeta_2^N d\zeta_1^N + x_1^N s_1^N. \]
This function is defined on the Newton polygon $\hat{\Delta}_N^\circ$. Observe that the Legendre dual Ronkin function $\sigma$ of $G_1$ is defined (\ref{eq:sigma_explicit}) on the Newton polygon of $G_1$, $\Delta^\circ \subset \frac{1}{N} \hat{\Delta}_N^\circ$.
We define the regularized surface tension as a function on $\frac{1}{N} \hat{\Delta}_N^\circ$:
\begin{equation}
	\label{eq:}
	\hat{\sigma}^N(\frac{s^N_1}{N}, \frac{s^N_2}{N}) := \frac{1}{N^2} \sigma^N(s^N_1, s^N_2).
\end{equation}

\begin{lemma}
	The regularized surface tension function $\hat{\sigma}^N$ converges on $\Delta^\circ$ to $\sigma$ given by (\ref{eq:sigma_explicit}) in the supremum norm. That is
	\[ \sup_{\Delta^\circ} \abs{\hat{\sigma}^N - \sigma} \xrightarrow{N \to \infty} 0, \; and \quad \hat{\Delta}_N^\circ \xrightarrow{N \to \infty} \Delta^\circ. \] 
\end{lemma}
\begin{proof}
	A point $P \in \mathcal{R}^\circ_+$ corresponds to $(s_1, s_2) = J(\myIm \zeta_1, \myIm \zeta_2)(P) \in \Delta^\circ$ and to $\frac{1}{N}(s^N_1, s^N_2) = \frac{1}{N}J(\myIm \zeta^N_1, \myIm \zeta^N_2)(P) \in \hat{\Delta}_N^\circ$. 

	We have
	\begin{equation}
		\label{eq:estimate_regularized_sigma}
		\abs{\hat{\sigma}^N(s_1, s_2) - \sigma(s_1, s_2)} \leq \abs{\frac{1}{N^2} \sigma^N(N s_1, N s_2) - \frac{1}{N^2} \sigma^N(s^N_1, s^N_2) } + \abs{\frac{1}{N^2} \sigma^N(s^N_1, s^N_2) - \sigma(s_1, s_2).}
	\end{equation}
	Observe that the arguments of the second term correspond to the same point $P \in \mathcal{R}^\circ_+$ and thus all the terms corresponding to the edges labelled by $(\alpha, \beta)$ cancel. Thus this term consists of contributions of $N$ edges $(\alpha_i, \delta_l)$, $N$ edges $(\beta_j, \gamma_l)$ and one edge $\gamma_{l'}, \delta_l$ for all $i, j, l, l'$. All these entries are of the form \eqref{eq:sigma_explicit} and are bounded. Thus the second term in \eqref{eq:estimate_regularized_sigma} is $\mathcal{O}\br{\frac{1}{N}}$.

	The first term in \eqref{eq:estimate_regularized_sigma} can be easily estimated through the derivatives. We obtain two terms of the form
	\[ \frac{1}{N^2} \abs{\frac{\partial \sigma^N}{\partial s^N_k}} \abs{s^N_k - Ns_k}. \]
	For $k=1$ this gives (and the same estimate for $k=2$)
	\[ \frac{1}{N^2}\abs{x^N_1} \abs{\sum_{l=1}^L \myIm \zeta^{\delta_l}} \leq \frac{1}{N} \br{\sum_{i=1}^m \sum_{l=1}^L \abs{\myRe \zeta^{\alpha_i}(P)} \abs{\myIm \zeta^{\delta_l}(P)} } + \mathcal{O}(\frac{1}{N^2}), \]
	where we have used \eqref{eq:def_rhoTilde}.

	To show that this term is $\mathcal{O}(\frac{1}{N})$ one observes that $\abs{\myRe \zeta^{\alpha_i}(P)} \abs{\myIm \zeta^{\delta_l}(P)}$ remain bounded at the singularities $\alpha_i^\pm$. Indeed at $z \to \alpha_i$ these terms have the following behaviour:
	\[ \abs{\myRe \zeta^{\alpha_i}(P)} \approx \log\abs{z - \alpha_i}, \quad \abs{\myIm \zeta^{\delta_l}(P)} \approx \abs{\myIm(z-\alpha_i)} \leq \abs{z-\alpha_i}. \]
	The product can be estimated by $\abs{z-\alpha_i} \log\abs{z-\alpha_i} \to 0$. Finally we obtain that both terms in \eqref{eq:estimate_regularized_sigma} are uniformly bounded by $\mathcal{O}\br{\frac{1}{N}}$.
\end{proof}

\begin{figure}[h]
\centering
\fontsize{10pt}{12pt}\selectfont
\def\svgwidth{.7\textwidth}
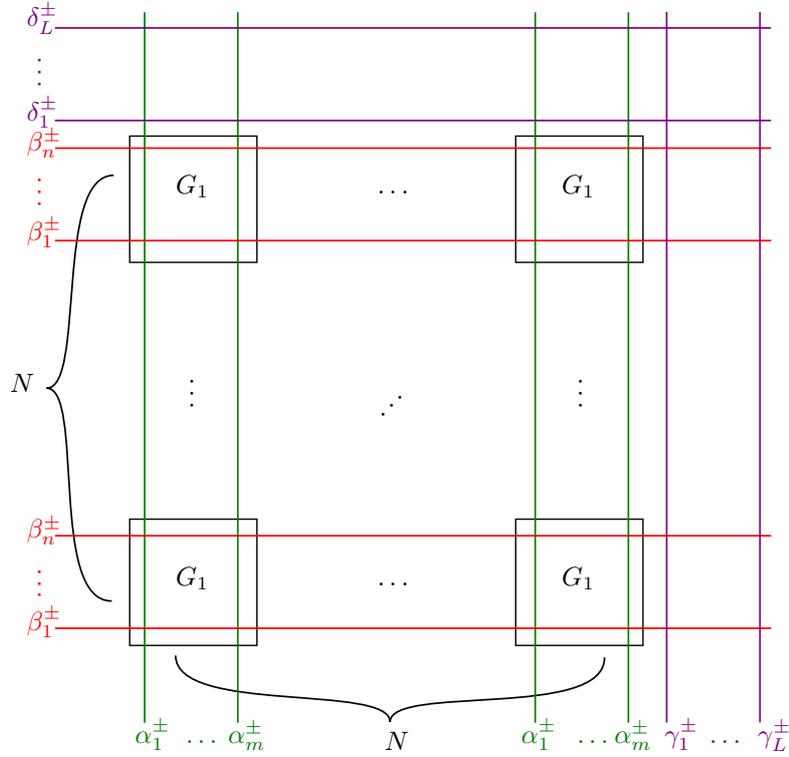
\caption{Graph $G_N$ and regularized graph $\hat{G}_N$. $G_N$ is composed of $N\times N$ copies of $G_1$ with doubly periodic train track parameters and thus quasi periodic weights. $\hat{G}_N$ is the graph $G_N$ extended by two strips with $L$ trains tracks each. $\hat{G}_N$ has zero monodromy for its Baker-Akhiezer functions and therefore defines doubly periodic weights. Each graph $G_1$ is a diagonal square lattice as in Fig.~\ref{Fig sq lattice}. }
\label{fig:regularizationGraph}
\end{figure}

\subsection{Regularized convergence of the height function}

Surface tension is important because it defines a functional that is minimized by the limiting height function of the corresponding dimer model.

For a domain $\Omega \subset \mathbb{R}^2$ and a set $\Delta \subset \mathbb{R}^2$ let us denote the space of functions 
\[ \Lip_{\Delta}(\Omega) := \cubr{f \in \Lip(\Omega, \mathbb{R}) \;|\; \nabla f \in \Delta \; \text{a.e.}}, \quad \Lip_{\Delta}(\Omega, h_b) := \cubr{f \in \Lip_\Delta(\Omega) \;|\; \restr{f}{\partial\Omega} = h_b }. \]
Note that the existence of the gradient is guaranteed almost everywhere due to Rademacher's theorem. This is the natural family of spaces for limiting height functions with Newton polygon $\Delta$. We omit the index $\Delta$ when it is clear from context.

The following is a result from \cite{cohn_variational_2000} and can also be found in \cite{sheffield_random_2006}.

\begin{theorem}
\label{ch10_thm:doublyPeriodicHeightFunctionConvergence}
	Let $\Omega \subset \mathbb{R}^2$ be a simply connected region with piecewise smooth boundary and $h_b$ be a function in $\Lip_{[0,m]\times [0,n]}(\Omega)$ restricted to $\partial\Omega$. Let $G_k \subset \frac{1}{k}\br{ \frac{1}{m} \mathbb{Z} \times \frac{1}{n} \mathbb{Z}}$ be a sequence of graphs that approximate $(\Omega, h_b)$ in the following sense:
	\begin{itemize}
		\item $G_k \subset \Omega$ for all $k$.
		\item The Hausdorff distance between the boundary of $G_k$ and $\partial\Omega$ tends to 0.
		\item $G_k$ permits dimer configurations and the height function $h_k$ on the boundary approximates the boundary conditions. That is for $x \in \partial\Omega$ if we define $h_k(x)$ as $h_k$ of the face closest to the point $x$ then $\norm{\frac{h_k}{k} - h_b}_\infty \to 0$.
	\end{itemize}

	Fix a fundamental domain of size $(m\times n)$ with some fixed dimer weights on it and populate all of $G_k$ doubly periodically with those weights. Let $\mathbb{P}_k$ be the corresponding dimer Boltzmann measures. Let $\sigma$ be the surface tension corresponding to these weights. Then for the limiting height function defined as the minimizer
	\begin{equation}
		\label{eq:surface_tension_minimization}
		h \coloneqq \argmin_{h \in \Lip(\Omega, h_b)} \int_\Omega \sigma(\nabla h)
	\end{equation}
we have the large deviations principle 
	\[ \mathbb{P}_k(\norm{\frac{h_k}{k} - h}_\infty > \epsilon) \approx e^{- k^2 \epsilon} .\]
\end{theorem}

The surface tension minimizer is unique due to the strict convexity of $\sigma$.
\begin{rem}
\label{re:affine_sigma}
Surface tension functions that differ by an affine function $b_1 y_1+b_2 y_2 +b_3$ lead to the same variational description. Indeed, by Stokes theorem we have
$$
\int_\Omega \nabla h\ du dv=\int_{\partial\Omega} h\ n \ ds,
$$
where $n$ is the outward normal to the boundary curve $\partial\Omega$. The last expression is given in terms of the boundary data.
\end{rem}

We will now use Theorem~\ref{ch10_thm:doublyPeriodicHeightFunctionConvergence} for doubly periodic weights to establish convergence of height functions in a regularized sense for quasi periodic weights.

To that end we consider a sequence of graphs $G_k \subset \frac{1}{k}\br{ \frac{1}{m} \mathbb{Z} \times \frac{1}{n} \mathbb{Z}}$ that approximate a domain $\Omega$ with boundary height $h_b$ as above.

Let $\alpha^\pm_i, i \leq m$, $\beta^\pm_j, j \leq n$ be the train track angles repeating periodically. As before we consider a fundamental domain $\hat{G}_N$ of size $(N m + L) \times (N n + L)$ with $4L$ added regularizing train tracks $\gamma^\pm_k, \delta^\pm_l$ such that the periodicity condition \eqref{ch07:eq_periodicDiracQuadGrid} is satisfied on this domain. We denote by $\sigma^N$ the corresponding surface tension and $\hat{\sigma}^N(\frac{s^N_1}{N}, \frac{s^N_2}{N}) := \frac{1}{N^2} \sigma^N(s^N_1, s^N_2)$ its normalized version defined on $\hat{\Delta}^\circ_N$. Henceforth we denote the surface tension functional as $S_N(h) \coloneqq \int \hat{\sigma}^N(\nabla h)$ and $S(h) \coloneqq \int \sigma(\nabla h)$.

For each of these fundamental domains we are in the doubly periodic case so by a scaled version of Theorem~\ref{ch10_thm:doublyPeriodicHeightFunctionConvergence} we have the limiting average height functions 
\[ h_N = \argmin_{h \in \Lip_{\hat{\Delta}^\circ_N}(\Omega, h_b)} S_N(h).\]
It was shown in \cite{cohn_variational_2000} that the functional $S$ is lower semicontinuous with respect to the supremum norm. This is the key ingredient to establishing convergence of the minimizing functions $h_N$.

\begin{theorem}
\label{ch10:thm_regHeightFunctionConv}
	The limiting height functions $h_N$ converge to the minimizer of the surface tension $\sigma$. That is
	\[ h_N \xrightarrow{\norm{\cdot}_\infty} h \coloneqq \argmin_{f \in \Lip_{\Delta^\circ}(\Omega, h_b)} S(f).\]
\end{theorem}
\begin{proof}
	We have seen that $\hat{\sigma}^N \xrightarrow{\norm{\cdot}_\infty} \sigma$. Since $h_N \in \Lip_{\hat{\Delta}^\circ_N}(\Omega, h_b) \supset \Lip_{\Delta^\circ}(\Omega, h_b)$ for all $N$ we get by Arzela-Ascoli the existence of a convergent subsequence $h_{N_k} \xrightarrow{k\to\infty} g\in \Lip_{\Delta^\circ}(\Omega, h_b)$. We show that this limit minimizes $S(\cdot)$.

	Suppose that $S(g')$ = $S(g) - \delta$ for some $\delta > 0$. Then for $\epsilon = \abs{\Omega} \norm{\sigma - \hat{\sigma}^N}_\infty$ and due to lower semicontinuity of $S$ we have
	\begin{align*}
		S_N(g') &\leq S(g') + \epsilon = S(g) + \epsilon - \delta \\
						 &\leq S(h_N) + \epsilon + \epsilon' - \delta \\
						 &\leq S_N(h_N) + 2\epsilon + \epsilon' - \delta \\
						 &< S_N(h_N) \quad\quad \text{for N large enough.}
	\end{align*}
	This contradicts $h_N$ being the minimizer of $S_N$. Therefore $g$ is indeed a minimizer of $S$. Since $\sigma$ is stricly convex, the minimizer is unique. Therefore in the sequence $h_N$ any subsequence has a subsequence converging to the same limit $h$. This yields convergence $h_N \xrightarrow{\norm{}_\infty} h$.
\end{proof}

We thus have a limiting height function $h$ under the regularized limit. We conjecture that this convergence also holds without a regularization step in a direct sense akin to Theorem~\ref{ch10_thm:doublyPeriodicHeightFunctionConvergence}.

Imposing a volume constraint of $V \in \mathbb{R}$ on the minimization problem \eqref{eq:surface_tension_minimization} yields
\begin{equation}
	\label{eq:surface_tension_minimization_constrained}
	h \coloneqq \argmin_{h \in \Lip(\Omega, h_b)} \int_\Omega \sigma(\nabla h), \; \text{ s.t. } \int_\Omega h = V.
\end{equation}
\[  \]

It is well known (See e.g. \cite{Kenyon_Okounkov_Sheffield_2006} for details and further references) that the minimizer $h$ satisfies the Euler-Lagrange equation
\[ \Div \br{\nabla\sigma \circ \nabla h(x) } = \lambda \]
with some Lagrange multiplier $\lambda\in \mathbb{R}$ (depending on $V$ and $h_b$) where $h$ is $C^2$, and $\sigma$ is smooth at $h(x)$. Here $\lambda=0$ corresponds to the unconstrained minimization problem.

Since $\rho$ is the Legendre dual of $\sigma$ we have 
\[ \Div \br{\nabla\sigma \circ \nabla \rho } = \Div \br{\text{Id}} = 2. \]
The Ronkin function thus satisfies \eqref{eq:surface_tension_minimization_constrained} for some volume $V$ and its own boundary conditions. See Figure~\ref{fig:simulation_with_amoeba} for an example of this match of theoretical prediction and simulation and Section~\ref{sec:12_computation} for a description of the computational aspects.

\section{Schottky Uniformization}
\label{sec:10_Schottky_Uniformization}

\addtocontents{toc}{\protect\setcounter{tocdepth}{1}}

Formulas for weights, surface tension and other characteristics of dimer models presented in this paper are parametrized by compact Riemann surfaces. This seems to be a rather complicated parametrization if one wants to investigate these dimer models or just plot them. Fortunately this problem can be solved using the Schottky uniformization of Riemann surfaces. This method was first suggested in \cite{Bobenko_1987} (see also \cite{belokolos_algebro-geometric_1994}) for a similar problem of computation of algebro-geometric solutions of nonlinear integrable equations. We shortly present it here.

Let $\{C_1, C_1', \ldots, C_g, C_g'\}$ be a collection of $2g$ mutually disjoint Jordan curves that form the boundary of a $2g$-connected domain $\Pi\subset \mathbb{C}\cup \{ \infty\}$. The fractional-linear transformation $\sigma\in  \text{PSL}(2, \mathbb{C})$ given by

\[ \frac{\sigma_n(z) - B_n}{\sigma_n(z) - A_n} = \mu_n \frac{z - B_n}{z - A_n}, \quad |\mu_n| < 1.\]

maps the outside of a boundary curve $C_n$ onto the inside of the boundary curve $C_n'$, $\sigma_n C_n=C_n'$. Its fixed point $A_n$ lies inside $C_n$ and $B_n$ lies inside $C_n'$. 

The freely generated group $G = <\sigma_1, \ldots, \sigma_g>$ is called a \emph {Schottky group}. If all boundary curves $C_n, C_n'$ are circles the Schottky group is called \emph{classical}.

Let $\Omega$ be the discontinuity set of a Schottky group $G$. It is a classical result  (see e.g. \cite{ford_automorphic_2004}) that $\mathcal{R} = \Omega / G$ is a compact Riemann surface of genus $g$. Note that $G$ acts freely on $\Omega$ and that $\Omega / G \approx \Pi(G) \coloneqq \mathbb{C} \setminus \{2g \text{ disks}\}$. Here $\Pi(G)$ is a fundamental domain of $G$. Observe that in fact the topology is that of a sphere with $2g$ holes glued in pairs, thus a surface with $g$ handles.

Note that given a Riemann surface $\mathcal{R}$, the choice of Schottky group representing it is not unique. 
The Schottky uniformization theorem states that any compact Riemann surface $\mathcal{R}$  with a choice of homologically independent simple disjoint loops $v_1, \ldots, v_g$ can be realized as $\mathcal{R} = \Omega(G) / G$, with $G$ Schottky. It is, however, an open question whether $G$ can be chosen to be classical.

The uniformization parameters $\{A_n, B_n, \mu_n \}$ and $ \{\tau A_n, \tau B_n, \mu_n\}$ that differ by a fractional-linear mapping $\tau$ represent the same Riemann surface. The complex dimension of the uniformization space is thus $1$ for $g=1$ and $3g - 3$ otherwise. This corresponds to the complex dimension of the moduli space of compact Riemann surfaces of genus $g$ as one would expect from the Schottky uniformization theorem.

In order to compute the period matrix and Abel map on $\mathcal{R}$, we need to fix a basis of cycles. Given a Schottky group $G$ we chose   a \emph{canonical homology basis} $a_1, \ldots, a_g$, $b_1, \ldots, b_g$ such that the $a_n$ are positively oriented cycles around the $B_n$, $b_n$ are paths from some $z_n \in \sigma_n^{-1}(a_n)$ to $\sigma_n(z_n) \in a_n$ such that the intersection indices are $a_k \circ b_l = \delta_{kl}$ and $a_k \circ a_l = b_k \circ b_l = 0$ $\forall k,l$.

Any element in $G$ can be represented as a word consisting of its generators and inverses
\[G \ni \sigma = \sigma_{i_1}^{j_1} \ldots \sigma_{i_k}^{j_k}, \quad i_1, \ldots i_k \in \{1, \ldots, g\}, \; j_1, \ldots, j_k \in \mathbb{Z}.\]

This representation is unique if we impose that $i_k \neq i_{k+1}$. We furthermore denote by $G / G_n \coloneqq \{\sigma \;|\; i_k \neq n\}$ the elements whose words do not end in $\sigma_n$, and similarly ${G_m\setminus G/G_n} \coloneqq \{\sigma \;|\; i_1 \neq m, i_k \neq n\}$. The following theorem gives expressions for the holomorphic differentials and the period matrix in terms of series over the group $G$, see \cite{burnside_1892}.

\begin{theorem}
	\label{thm:Schottky_series}
	Given a Schottky group $G$ with a canonical homology basis. If
	\[ \omega_n(z) = \frac{1}{2\pi i}\sum_{\sigma \in G / G_n}  \br{ \frac{1}{z-\sigma B_n} - \frac{1}{z - \sigma A_n} }  dz \]
	converges absolutely, then $\omega_1, \ldots \omega_g$ forms a basis of holomorphic differentials dual to the given homology basis.

	Furthermore the period matrix is then given by
	\begin{align*}
		B_{nm} &=  \frac{1}{2\pi i} \sum_{\sigma \in G_m \backslash G / G_n} \log \cubr{B_m, \sigma B_n, A_m, \sigma A_n}, \quad \text{if } n \neq m \\
		B_{nn} &=  \frac{1}{2\pi i}\log \mu_n +  \frac{1}{2\pi i}\sum_{\sigma \in G_n \backslash G / G_n, \sigma \neq I} \log \cubr{B_n, \sigma B_n, A_n, \sigma A_n}
	\end{align*}
	where the curly brackets indicate the cross-ratio $\cubr{z_1, z_2, z_3, z_4} = \frac{z_1 - z_2}{z_2 - z_3} \frac{z_3 - z_4}{z_4 - z_1}$.
\end{theorem}

To see this, first observe that $\omega_n$ is indeed invariant under action in $G$ and thus well defined on $\mathcal{R}$. Furthermore it has no poles in the fundamental domain and is thus holomorphic. The normalization $\int_{a_m} \omega_n = \delta_{nm}$ is true by Cauchy's residue theorem. A slightly more involved calculation yields the period matrix.

In general the question of absolute convergence of these Poincaré series is open. The convergence is known \cite{burnside_1892, baker_1897} for sufficiently small $C_1, C_1', \ldots, C_g, C_g'$ or, equivalently, for sufficiently small $|\mu |$'s. Another classical result \cite{schottky_1887, fricke-klein_1965} is that the convergence is guaranteed in the circle-decomposable case, i.e. when the fundamental domain $\Pi (G)$ of a classical Schottky group can be decomposed into non-intersecting pairs of pants by circles. See \cite{belokolos_algebro-geometric_1994} for more details and \cite{schmies_computational_2005} for error estimates.

\subsection*{Schottky uniformization of M-curves}

The problem of convergence of the Poincaré series, as well as the problem of the description of the Schottky space
$$
S=\{A_1, B_1, \mu_1,\ldots, A_g, B_g, \mu_g\},
$$
can be solved for M-curves, and more general for real Riemann surfaces. In this section we describe two Schottky uniformizations of M-curves, called U1 and U2, that are useful in different regimes, see \cite[Chapter~5]{belokolos_algebro-geometric_1994}. 

Let $\mathcal{R}$ be an M-curve with antiholomorphic involution $\tau$ that is separated into two parts $\mathcal{R}_+, \mathcal{R}_-$ by the real ovals $X_0, \ldots, X_g$. 

The surface $\mathcal{R}_+$ can be conformally mapped to the upper half plane with $g$ round holes, so that $X_0$ being mapped to the real axis. Applying the anti-holomorphic symmetry $\tau z = \bar{z}$ we obtain a fundamental domain of a classical Schottky group, see Fig.~\ref{fig:Schottky+Amoeba}. The boundary circles $C_n$ and $C_n'$ are mutually complex conjugated $C_n$ and $C_n'$, and represent  
other real ovals $X_n, n=1,\ldots, g$. We call this U2-uniformization of the M-curve $\mathcal{R}$.

The parameters of this Schottky group satisfy
\[ A_n = \overline{B_n}, \; \myIm(A_n) > 0 ,\; 0 < \mu_n < 1, \; \forall n.\]
Furthermore it is easy to deduce inequalities for the discs bounded by $C_n$ to be disjoint:
$$
\left| \frac{A_n - \mu_n \bar{A}_n}{1-\mu_n}-\frac{A_m - \mu_m \bar{A}_m}{1-\mu_m} \right |>
\frac{2\myIm A_n}{1-\mu_n}+\frac{2\myIm A_m}{1-\mu_m}, \; \forall m,n.
$$
This gives a complete description of the Schottky space $S$.
Note that in general this uniformization is not circle decomposable and convergence for the differentials is not guaranteed. However it is circle decomposable for $g \leq 2$ as well as for sufficiently small $\mu$ which corresponds to amoebas with small hole sizes. 

The meromorphic differentials $d\zeta_i$ can be written in terms of Poincaré series as well. The formulas in for the U2-uniformization are especially simple. 

\begin{proposition}
	Let $G$ be a Schottky group in U2 form. Then the differentials $d\zeta_k$ and integrals $\zeta_k$ are given by the following series (provided they converge):
	\begin{eqnarray}
	& d\zeta^{\alpha_i}(z) = \sum_{\sigma \in G} \br{\frac{1}{ z - \sigma \alpha_i^-} - \frac{1}{ z - \sigma \alpha_i^+} } dz,\quad
	d\zeta^{\beta_j}(z) = \sum_{\sigma \in G} \br{\frac{1}{ z - \sigma \beta_j^-} - \frac{1}{ z - \sigma \beta_j^+} } dz,
	\label{eq:dzeta_Schottky}\\
	& \zeta^{\alpha_i}(z) = \sum_{\sigma \in G} \log \{  z, \sigma\alpha_i^-, z_0, \sigma \alpha_i^+\}, \quad
	\zeta^{\beta_j}(z) = \sum_{\sigma \in G} \log \{  z, \sigma\beta_j^-, z_0, \sigma \beta_j^+\}.
	\label{eq:zeta_Schottky}
	\end{eqnarray}
	Here $z_0 \in \mathbb{R}$ is the starting integration point and the Harnack data are cyclicaly ordered $\beta_j^- < \alpha_i^+ < \beta_j^+ < \alpha_l^-  \in \mathbb{R}$.
\end{proposition}
\begin{proof}
It is easy to check that the differentials are invariant with respect to $G$, and clearly all $a$-periods vanish. Because of the complex conjugation symmetry all $b$-periods are imaginary, so the differentials given by formulas (\ref{eq:dzeta_Schottky}) are indeed the differentials  $d\zeta_k$.
\end{proof}

Another natural Schottky uniformization for arbitrary M-curves comes from the Fuchsian uniformization $H/G$ of $\mathcal{R}_+$. Here $H$ is the upper half-plane and $G$ is the Fuchsian group of the second kind. In this case all $A_n, B_n$ are real valued, and are in nested pairs with non-intersecting circles, forming a fundamental domain of $G$. Extending its action to the lower half-plane we obtain a classical Schottky group uniformizing $\mathcal{R}$. It follows both from circle decomposability as well as from general theory of Fuchsian groups that in this uniformization, which we call U1, the series in Theorem~\ref{thm:Schottky_series} always converge.



\section{Computation}
\label{sec:12_computation}

For computations  we use the U2 uniformization. 


 We use the jtem library introduced in \cite{schmies_computational_2005} for calculation of Abelian differentials and integrals coming from a given Schottky uniformization as well as corresponding theta functions in order to get the weights. Its functionality was extended  to compute the differentials $d\zeta_k$ and their integrals $\zeta_k$ thus in particular yielding the amoeba map and generalized Newton polygon. 
 
 The amoeba map is given by the real part of (\ref{eq:zeta_Schottky}):
 \begin{align*}
	\mathcal{A}(z) &= (\myRe(\zeta_1), \myRe(\zeta_2)) \\
								 &= \br{ \sum_i \sum_{\sigma \in G} \log \abs*{\cubr{ z, \sigma\alpha_i^-, z_0, \sigma\alpha_i^+}},  \sum_j \sum_{\sigma \in G} \log \abs*{\cubr{ z, \sigma\beta_j^-, z_0, \sigma\beta_j^+}} }.
\end{align*}
We choose the starting integration point $Q$  between $\alpha^-$ and $\beta^-$ so we set $z_0 = \infty$.
 
The group's elements $\sigma \in G$ can be written as words consisting of the generators $(\sigma_1, \ldots, \sigma_g)$. By limiting the maximal word length $L$ we have a finite sum approximation of $d\zeta_k$. For approximation error estimates and more sophisticated approximation schemes we refer the reader to \cite{schmies_computational_2005}.
Fig.~\ref{fig:Schottky+Amoeba} presents the image of the amoeba map of the real ovals.


\begin{figure}
\centering
\begin{subfigure}[p]{.45\textwidth}
	\centering
	\fontsize{10pt}{12pt}\selectfont
	\def\svgwidth{0.9\linewidth}
	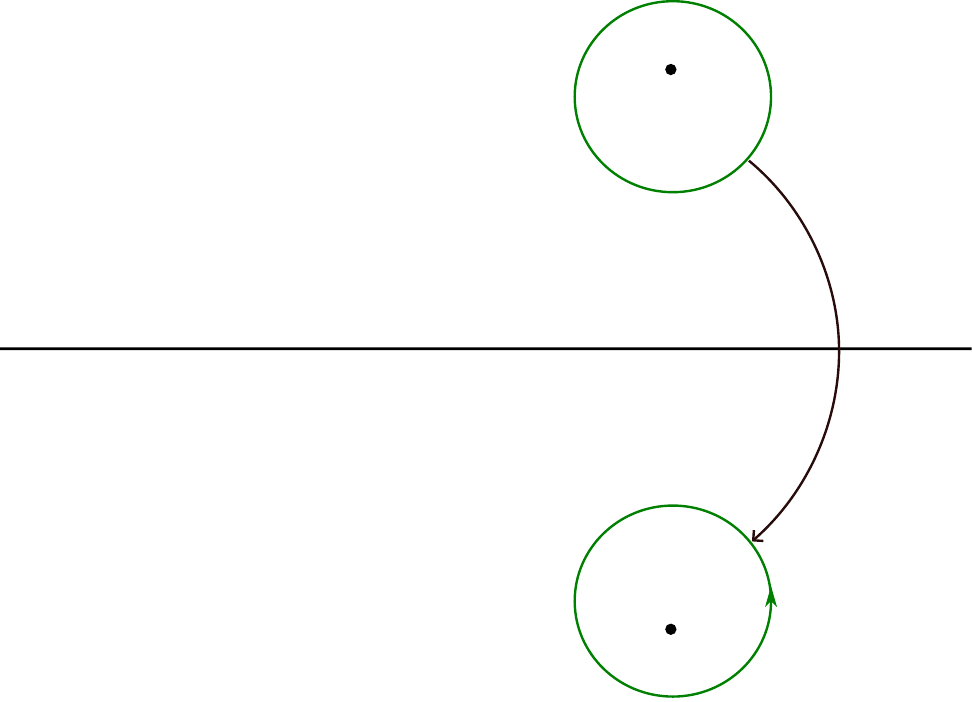
\end{subfigure}
\begin{subfigure}[p]{.45\textwidth}
	\centering
	\fontsize{10pt}{12pt}\selectfont
	\def\svgwidth{0.9\linewidth}
\begingroup%
  \makeatletter%
  \providecommand\color[2][]{%
    \errmessage{(Inkscape) Color is used for the text in Inkscape, but the package 'color.sty' is not loaded}%
    \renewcommand\color[2][]{}%
  }%
  \providecommand\transparent[1]{%
    \errmessage{(Inkscape) Transparency is used (non-zero) for the text in Inkscape, but the package 'transparent.sty' is not loaded}%
    \renewcommand\transparent[1]{}%
  }%
  \providecommand\rotatebox[2]{#2}%
  \newcommand*\fsize{\dimexpr\f@size pt\relax}%
  \newcommand*\lineheight[1]{\fontsize{\fsize}{#1\fsize}\selectfont}%
  \ifx\svgwidth\undefined%
    \setlength{\unitlength}{553.86338624bp}%
    \ifx\svgscale\undefined%
      \relax%
    \else%
      \setlength{\unitlength}{\unitlength * \real{\svgscale}}%
    \fi%
  \else%
    \setlength{\unitlength}{\svgwidth}%
  \fi%
  \global\let\svgwidth\undefined%
  \global\let\svgscale\undefined%
  \makeatother%
  \begin{picture}(1,0.96575796)%
    \lineheight{1}%
    \setlength\tabcolsep{0pt}%
    \put(0,0){\includegraphics[width=\unitlength,page=1]{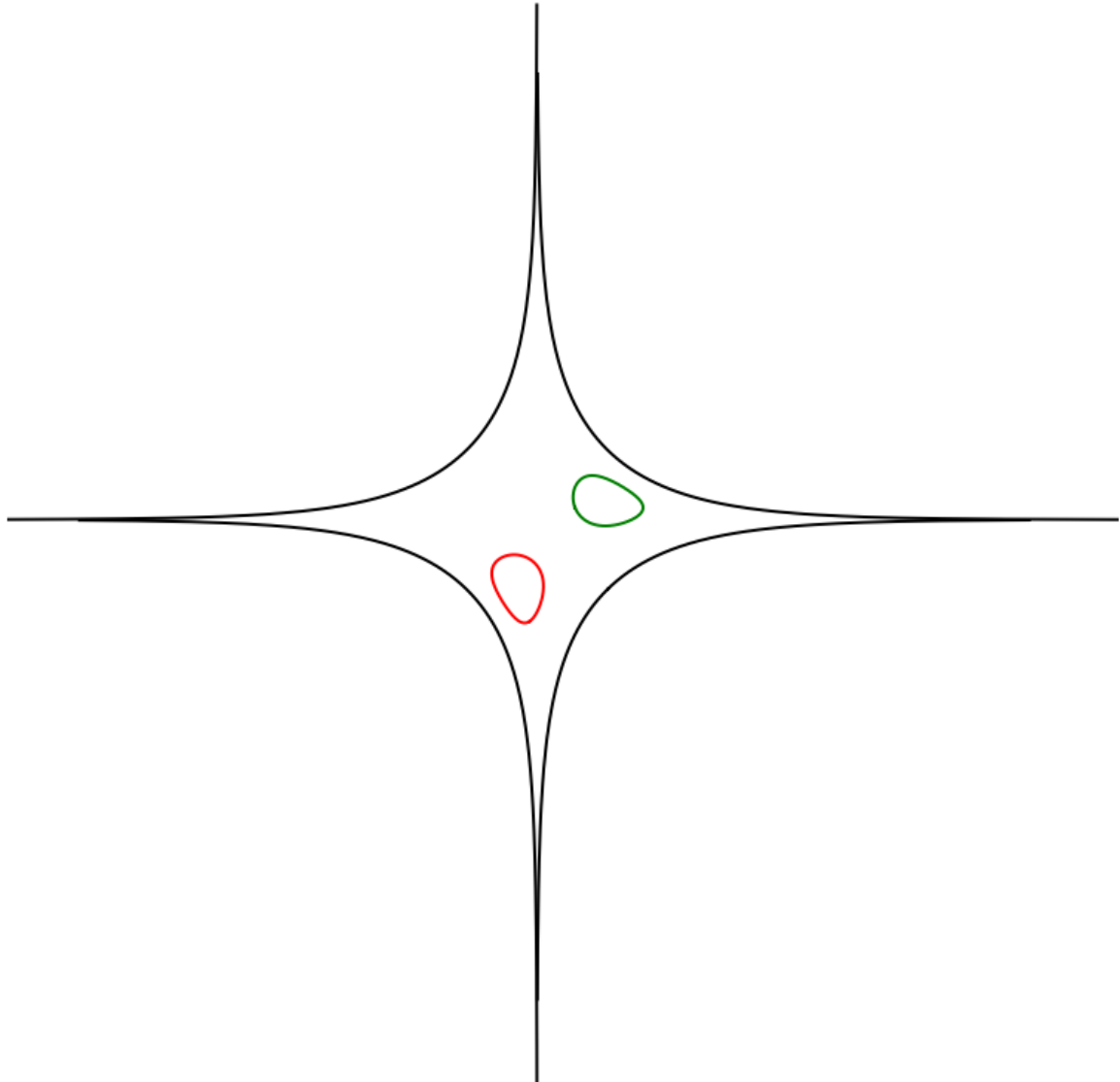}}%
    \put(0.03996588,0.53616149){\color[rgb]{0,0,0}\makebox(0,0)[lt]{\lineheight{1.25}\smash{\begin{tabular}[t]{l}$\alpha^-$\end{tabular}}}}%
    \put(0.5165483,0.04477838){\color[rgb]{0,0,0}\makebox(0,0)[lt]{\lineheight{1.25}\smash{\begin{tabular}[t]{l}$\beta^-$\end{tabular}}}}%
    \put(0.90210632,0.53098126){\color[rgb]{0,0,0}\makebox(0,0)[lt]{\lineheight{1.25}\smash{\begin{tabular}[t]{l}$\alpha^+$\end{tabular}}}}%
    \put(0.52172853,0.87657749){\color[rgb]{0,0,0}\makebox(0,0)[lt]{\lineheight{1.25}\smash{\begin{tabular}[t]{l}$\beta^+$\end{tabular}}}}%
    \put(0.42700403,0.52284089){\color[rgb]{0,0,0}\makebox(0,0)[lt]{\lineheight{1.25}\smash{\begin{tabular}[t]{l}$\mathcal{A}_\mathcal{S}$\end{tabular}}}}%
  \end{picture}%
\endgroup%

\end{subfigure}
\caption{Schottky uniformization U2 and the corresponding amoeba.}
\label{fig:Schottky+Amoeba}
\end{figure}

We now take an $N \times N$ square grid with repeating train track angles $\alpha^-, \beta^-, \alpha^+, \beta^+$. The face weight $W_f$ defined in \eqref{gauge inv} of a face $f$ surrounded by faces $f_n, f_w, f_s, f_e$ is then given by

\begin{equation}
\label{eq:face_weight_sq_lattice}
	W_f = \frac{\theta[\Delta](\int_{\alpha^\pm}^{\beta^\pm} \omega)}{\theta[\Delta](\int_{\beta^\pm}^{\alpha^\mp} \omega)} \frac{\theta[\Delta](\int_{\alpha^\mp}^{\beta^\mp} \omega)}{\theta[\Delta](\int_{\beta^\mp}^{\alpha^\pm}\omega)} \frac{\theta(\eta(f_e) + D) \theta(\eta(f_w) + D)}{\theta(\eta(f_n) + D) \theta(\eta(f_s) + D)}.
\end{equation}

\begin{figure}
\centering
\begin{subfigure}[p]{.45\textwidth}
	\centering
	\includegraphics[width=\linewidth]{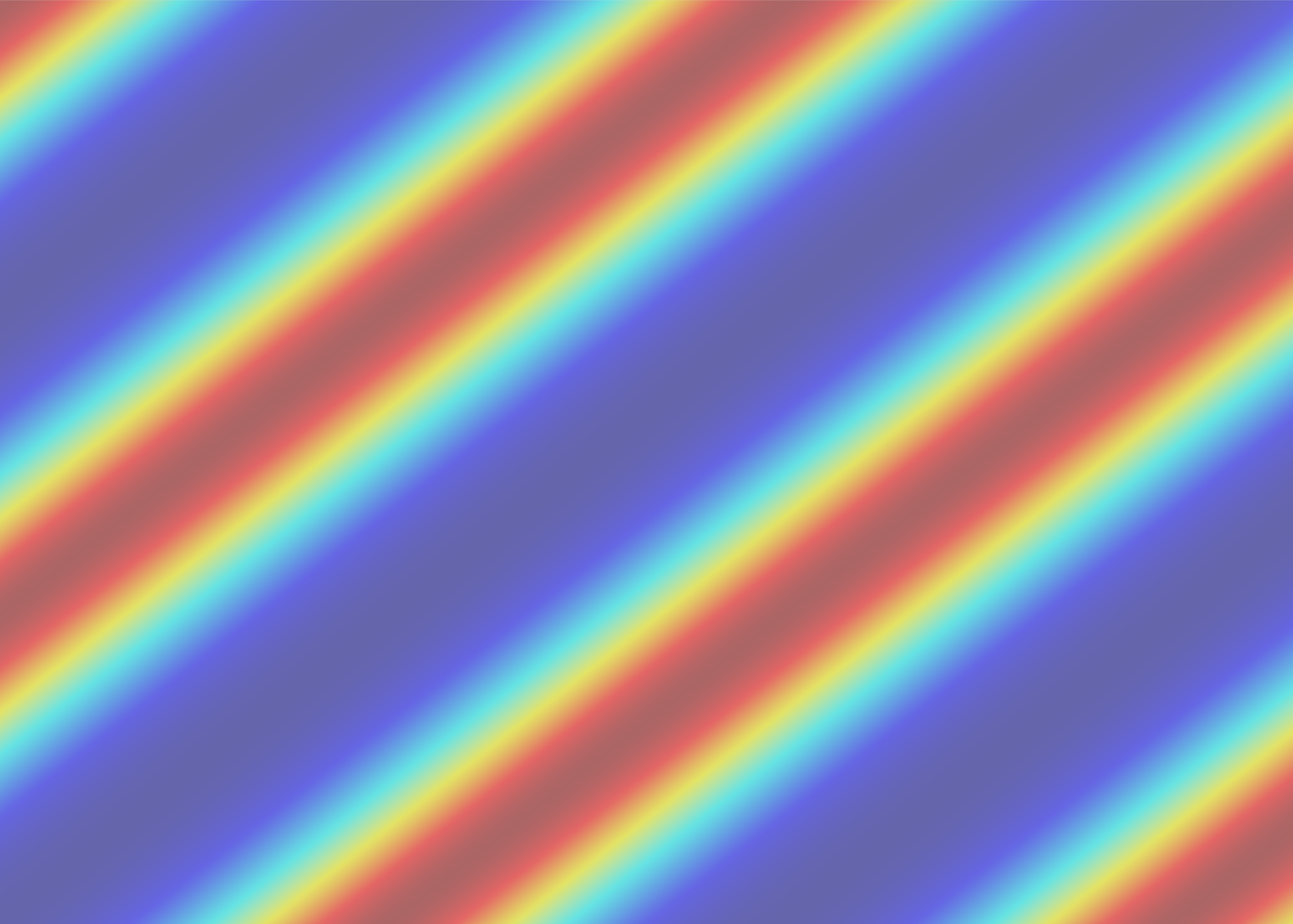}
\end{subfigure}
\begin{subfigure}[p]{.45\textwidth}
	\centering
	\includegraphics[width=\linewidth]{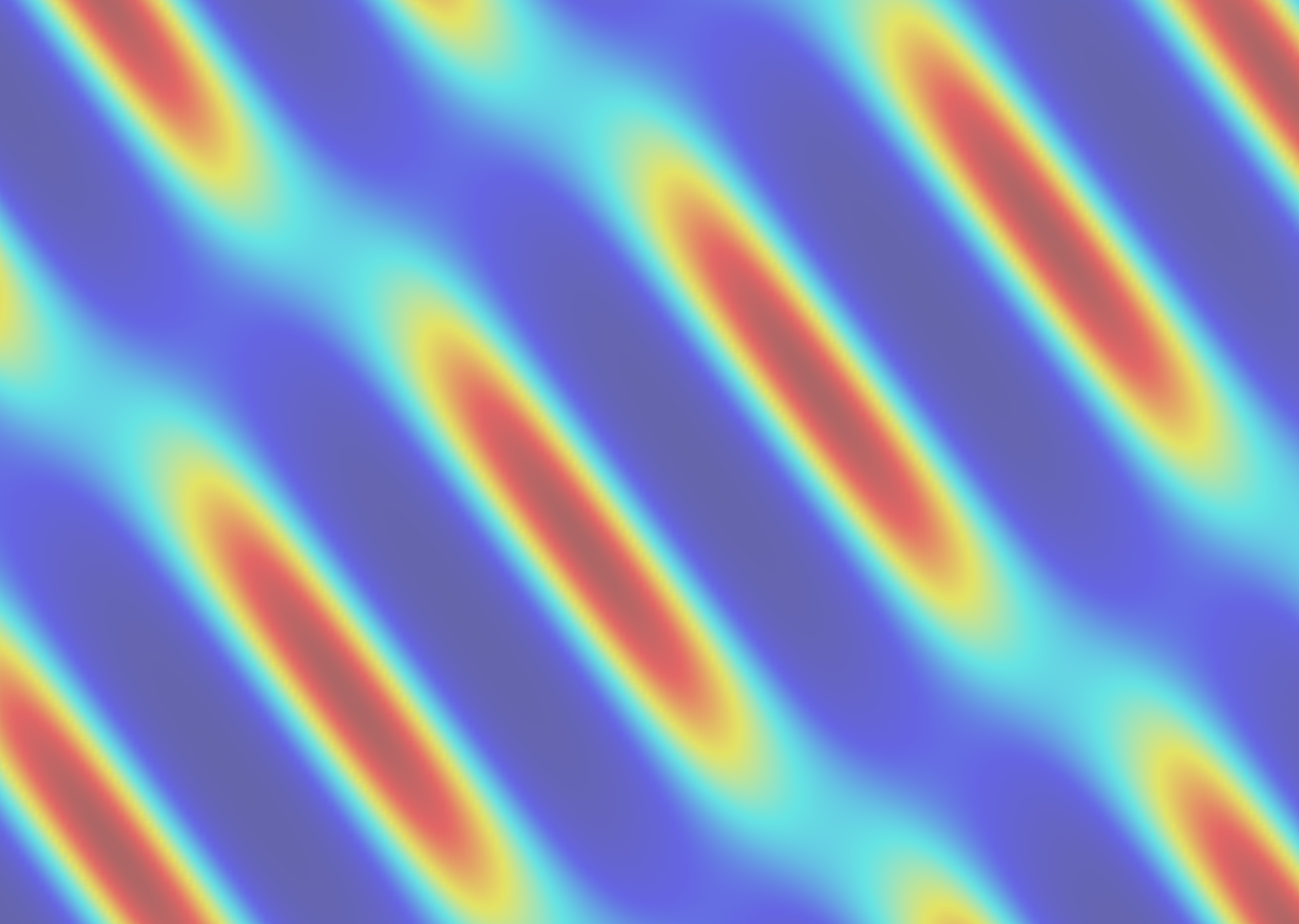}
\end{subfigure}
\caption{Two examples $W_f$ as a continuous function of $\eta(f)$ \eqref{eq:face_weight_sq_lattice} with g=1 and g=2. Computed for Schottky data with $\cubr{\beta^-, \alpha^+, \beta^+, \alpha^-} = \cubr{-2.4, -0.4, 0.4, 2.4}$ and $(A,B,\mu) = (0.1 + i, 0.1 - i, 0.02)$(left) and $(A_1, B_1, \mu_1, A_2, B_2, \mu_2) = (1.2 + i 1.3, 1.2 - i 1.3, 0.08, -0.4 + i 0.6, -0.4 - i 0.6, 0.03)$(right). }
\label{fig:FaceWeightsExamples}
\end{figure}

The choice of signs here depends on the order of traintracks surrounding the face $f$. Note that the ratio of prime forms simplifies to a ratio of theta functions with odd characteristics because the holomorphic spinors in \eqref{E thru theta} cancel. Hence, to compute $W_f$ we need to compute the period matrix, the Abel maps and the corresponding theta functions. See \cite{schmies_computational_2005} for a description of both pointwise and uniform approximation schemes as well as error estimates and numerically stable decompositions. 

Intuitively the weights represent a non-linear interaction of $g$ waves analogous to multiphase solutions of nonlinear integrable equations (see, for example, \cite{belokolos_algebro-geometric_1994}). Fig.~\ref{fig:FaceWeightsExamples} illustrates $W_f$ in \eqref{eq:face_weight_sq_lattice} for two given Harnack data as a continuous function of $\eta(f)$. The face weights on $\mathbb{Z}^2$ are then given by evaluating this function on a periodic lattice. The periodicity condition \eqref{ch07:eq_periodicDiracQuadGrid} is equivalent to the fact that the periods of the lattice and of this function match.

We sample a random height function on the square grid with Fock weights via the Metropolis-Hastings algorithm to illustrate the convergence results in Section~\ref{sec:9_heightFunction}.\footnote{All code and further examples can be found at \url{github.com/nikolaibobenko/FockDimerSimulation}.}
A dimer configuration on the square grid can be encoded via $4$-bit incidence variables $n_f$ at every face indicating which edges are part of the dimer configuration with $(n, w, s, e)$ corresponding to $(1, 2, 4, 8)$. If north and south are present, i.e. $n_f = 5$ or east and west are present, i.e. $n_f = 10$ we can perform a local flip operation by updating $n_f$ and the neighboring values. This can be done efficiently via bitwise XOR operations with appropriate filters.

In \cite{thurston1990tilings} it was shown that any dimer configuration can be reached via local flips from any other. Since our state space is connected we can run the Metropolis-Hastings algorithm on the random walk defined by the local flip operation weighted by the dimer weights to sample a random height function according to the measure induced by the chosen weights. See also \cite{keating_random_2018} for more detail and extension to GPU-based parallelization.

In order to get behaviour of the type of the Ronkin function we need to impose a volume constraint as described in Section~\ref{sec:9_heightFunction}. This can be done by first starting with a configuration of prescribed volume $\int h = V$ and then running the Markov chain under the constraint of volume preservation. This is easiest achieved by flipping in pairs that cancel the volume change. 
Note that this construction of flipping two faces in opposite directions is not as amenable to parallelization as sampling from the unconditioned measure.

One could get similar results by sampling from a volume weighted measure. That is for some $q < 1$ the probability of a configuration $M$ becomes $\mathbb{P}(M) = \frac{1}{Z} q^{|V_M - V|} \prod_{e\in M} w_e$ where $V_M$ is the total volume of the height function corresponding to $M$. This is naively parallelizable and has similar behaviour to the volume constrained approach. See \cite{sheffield_random_2006} for more detail.


As we saw in Section~\ref{sec:9_heightFunction} the average height function under a volume constraint converges to the minimizer of the corresponding surface tension under a volume constraint. The volume constraint constant as well as boundary values of the height function can be tricky to choose to match the Ronkin function in the general case but we show an example hereof in Figs.~\ref{fig:simulation_with_amoeba} and \ref{fig:height3d}.

\begin{figure}
	\centering
	\includegraphics[width=.6\linewidth]{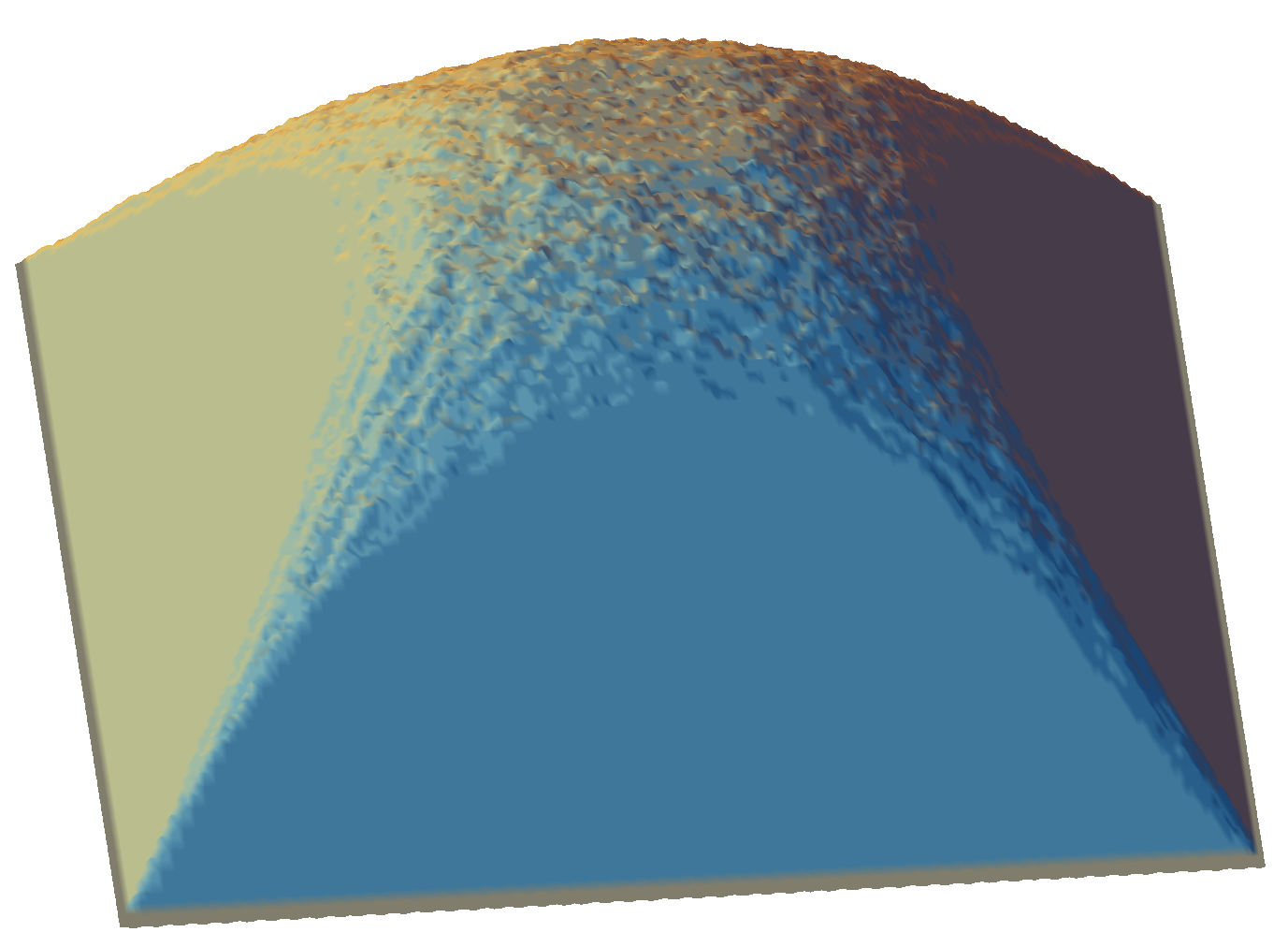}
	\caption{Sampled height function. This is same matching as seen in Fig.~\ref{fig:simulation_with_amoeba}.}
	\label{fig:height3d}
\end{figure}

\bibliographystyle{acm}

\bibliography{refs}

\end{document}